%% file: pmc_output_linear_open_access.tex
\documentclass[11pt,letterpaper]{article}

\usepackage[utf8]{inputenc}
\usepackage{amsmath,amsthm,amssymb}
\usepackage{fullpage}
\usepackage[ruled,vlined]{algorithm2e}
\usepackage{tikz}
\usepackage{esvect}
\usepackage{caption}
\usepackage{subcaption}
\usepackage[left= 3 cm,right=3 cm,top=3 cm,bottom=3 cm]{geometry}
\usepackage{mathtools}
\usepackage[colorlinks = true, citecolor = {blue}]{hyperref}
\usepackage{array}
\usepackage{comment}
\usepackage{authblk}
\usepackage{thm-restate}
\usepackage{enumitem}
\usepackage{todonotes}
\usetikzlibrary{positioning}
\usetikzlibrary{patterns}
\usetikzlibrary{decorations}
\usetikzlibrary{shapes.geometric}
\usetikzlibrary{shapes.misc}
\usetikzlibrary{arrows.meta}
\usetikzlibrary{patterns}
\usetikzlibrary{calc}

\tikzstyle{noeud}=[scale=0.8,draw, circle, fill=black]

\tikzset{
bicolor/.style 2 args={
  dashed,dash pattern=on 10pt off 10pt,#1,
  postaction={draw,dashed,dash pattern=on 10pt off 10pt,#2,dash phase=10pt}
  },
}

\newtheorem{theorem}{Theorem}
\newtheorem{lemma}{Lemma}
\newtheorem{corollary}{Corollary}
\newtheorem{definition}{Definition}

\newtheorem{observation}{Observation}

\newcommand{\defdecproblem}[3]{
 \vspace{1mm}
\noindent\fbox{
 \begin{minipage}{0.96\textwidth}
 #1 \\
 {\bf{Input:}} #2 \\
 {\bf{Question:}} #3
 \end{minipage}
 }
 \vspace{1mm}
}

\newcommand{\absep}{$(a,b)$-separator}
\newcommand{\notxi}[2]{\overline{#1_{#2}}}
\newcommand{\sepa}[2]{\mathcal{S}(#1,#2)}
\newcommand{\sepag}[3]{\mathcal{S}_{#1}(#2,#3)}
\newcommand{\torso}{\textsf{torso}}
\newcommand{\removal}{\textsc{conflict-greedy}}
\newcommand{\btenum}{\textsc{BT-enum}}

\newcommand{\lis}[1]{\textsc{lis}-(tw $\le #1)$}
\newcommand{\cmso}[1]{\textsc{max}-(tw $\le #1,\varphi)$}

\newcommand{\ALGO}[0]{\textsc{PMC-Traversal}}

\title{Output-linear time enumeration of potential maximal cliques\footnote{This work was supported by the French ``Agence Nationale de la Recherche'' project SPAWN (ANR-25-CE48-1750) and project MIAI (ANR-23-IACL-0006).}}

\author[1]{Pierre Berg\'e}
\author[2]{Vincent Limouzy}
\author[2]{Nicolas Schivre\footnote{Corresponding author}}

\affil[1]{Université Grenoble Alpes, CNRS, Grenoble INP, LIG, 38000 Grenoble, France}
\affil[2]{Université Clermont-Auvergne, CNRS, Mines de Saint-Etienne, Clermont-Auvergne-INP, LIMOS, 63000 Clermont-Ferrand, France}

\date{}

\begin{document}

\maketitle

\begin{abstract}
First defined by Bouchitté and Todinca, Potential Maximal Cliques (PMC) are a well studied notion in graph theory due to its numerous applications. The PMCs of an undirected graph $G$ are defined as the maximal cliques of the minimal triangulations of $G$. Their number can be exponential in the size of the graph. PMCs were crucial in the design of a polynomial-time algorithm which solves \textsc{Maximum Weight Independent Set} on $P_6$-free graphs (Grzesik, Klimošová, Pilipczuk, Pilipczuk, 2020). They also play a central role in the proposal of a polynomial-time algorithm to solve \textsc{Largest Induced Subgraph of treewidth $\le t$} (Fomin, Todinca, Villanger, 2013). 

Last but not least, certainly one of the main application of PMCs is the existence of an exact dynamic-programming algorithm to determine the treewidth of a graph (Fomin, Kratsch, Todinca, Villanger, 2008) which uses its list of PMCs as an input. In fact, the enumeration of PMCs stands as the bottleneck for the computation of treewidth with this method: the PMC-enumeration step takes output-quadratic time (Bouchitté, Todinca, 2002) while the dynamic-programming one takes a time linear in the number of PMCs. This two-step procedure was proven to be among the most performing approaches in practice to determine treewidth. 

Answering a longstanding open question, we propose an algorithm to enumerate PMCs in output-linear time, {\em i.e.} a running time linear in the number of PMCs. This offers nice perspectives not only from the theoretical viewpoint - better running time for the PMC-based methods determining treewidth, for \textsc{Minimum Fill-In} and \textsc{Largest Induced Subgraph of treewidth $\le t$} - but also from the practical side, with the possibility to propose top-performance implementation.
\end{abstract}

\newpage

\tableofcontents

\newpage

\section{Introduction} \label{sec:intro}

The theory of potential maximal cliques (PMC) was pioneered more than twenty years ago by Bouchitté and Todinca~\cite{BoTo98,BoTo99,BoTo01,BoTo02}. This combinatorial structure in graphs was introduced in the first place to compute efficiently the treewidth. Through the years, other applications of PMCs arose, a famous example being the design of polynomial-time algorithms for \textsc{maximum weight independent set}~\cite{GrKlPiPi19,LoVaVi14}. Also, many efficient algorithms for NP-hard problems on graphs use the collection of PMCs as an input to achieve their goal, making thus the enumeration of PMCs be the bottleneck for numerous algorithmic challenges.

\subsection{Motivations}

Considering an undirected connected graph $G = (V,E)$, a vertex set $X \subseteq V$ is a \emph{minimal separator} of $G$ if there exists a pair $a,b \in V \setminus X$ such that not only $a$ and $b$ are disconnected in $G \setminus X$ but also no proper subset of $X$ disconnects $a$ from $b$. We denote by $\Delta_G$ the set of minimal separators of $G$. An apparently unrelated notion - but in fact with strong connections - is the concept of \emph{potential maximal clique} (PMC). A vertex set $X \subseteq V$ is a PMC of $G$ if it is a maximal clique of some \emph{minimal triangulation} of $G$, {\em i.e.} an inclusion-wise minimal add-on of edges to make $G$ become chordal. Let $\Pi_G$ denote the set of all PMCs of $G$. Concretely, due to the strong relationship between triangulations and tree decompositions, PMCs can be seen as ``candidates for bags of useful tree decompositions'', echoing the words of Grzesik, Klimosov, Pilipczuk and Pilipczuk in~\cite{GrKlPiPi19}. This reason explains why PMCs play a crucial role in the design of exact exponential-time algorithms, as detailed below. Note that PMCs can be defined for disconnected graphs, but since the set of PMCs of a disconnected graph is simply the disjoint union of PMCs of its connected components, all graphs considered in this paper are connected.

We call \textsc{Treewidth} the problem of determining the treewidth of a graph. Coming back to the analogy between PMCs and bags of tree decompositions, computing the set $\Pi_G$ seems to be a natural approach to obtain the treewidth of $G$ and/or an optimal tree decomposition. This is what Bouchitté and Todinca stated formally in~\cite{BoTo01}. They proposed a dynamic programming procedure which solves \textsc{Treewidth} in time $O(n^2 \cdot \vert \Delta_G \vert \cdot \vert \Pi_G \vert )$, given that the sets $\Delta_G$ and $\Pi_G$ are provided, with $n = \vert V \vert$. This was later improved~\cite{FoKrToVi08} with a running time $O(n^3 \vert \Pi_G \vert)$. Such execution times are exponential since a graph may admit a number of minimal separators and/or PMCs exponential in $n$.

In these aforementioned algorithms for \textsc{Treewidth}, PMCs are essential to achieve the divide-and-conquer scheme which determines the treewidth of successive graphs until going up to the input one. However, if we aim at using such algorithms in practice, we must be able to enumerate PMCs before the execution, as efficiently as possible. Minimal separators of a graph $G$ can be enumerated in polynomial delay~\cite{BeBoCo99}, which implies a global enumeration time of $O(n^3\vert \Delta_G \vert)$, which is output-linear. Bouchitté and Todinca~\cite{BoTo02} proposed an algorithm (that we call \btenum ) enumerating all PMCs of a graph with a running time in $O(n^2m\vert \Delta_G \vert^2)$, where $m = \vert E \vert$. Additionally, they established the following inequality which strengthens even more the links between minimal separators and PMCs: $\frac{\vert \Delta_G \vert}{n} \le \vert \Pi_G \vert \le n\vert \Delta_G \vert^2 + n\vert \Delta_G \vert + 1$. Due to the left-hand side of this inequality, their algorithm \btenum\ guarantees an output-quadratic time: there can be instances on which the time taken by the enumeration is quadratic in the size of the output $\Pi_G$. Thanks to \btenum , by assembling results together, \textsc{Treewidth} can be solved in time $O(n^2m\vert \Delta_G \vert^2 + n^3\vert \Pi_G \vert) = O(n^4 \vert \Delta_G \vert^2)$, which can be quadratic in the number of PMCs (if we neglect polynomial factors).

An important challenge consists thus in identifying efficient algorithms for the enumeration of PMCs. Since \btenum\ is output-quadratic, a natural question is whether an output-linear time PMC enumeration algorithm can be built. This question is part of the list of open problems in parameterized and exact computations at IPEC 2006 conference~\cite{BoCaChFe06}, and was also stated in several articles~\cite{FoKrToVi08,GrKoTaZa25,Ko20,KoBeJa19,Vi05} as a way to obtain a faster exponential time algorithm for \textsc{Treewidth}. Indeed, a PMC enumeration algorithm running in $O(n^{O(1)}\vert \Pi_G \vert) = \tilde{O}(\vert \Pi_G \vert)$\footnote{Notation $\tilde{O}$ neglects factors which are polynomial in the number $n$ of vertices.} would imply the same time complexity for \textsc{Treewidth}, improving the current $\tilde{O}(\vert \Delta_G \vert^2)$, especially on graphs where values $\vert \Delta_G \vert$ and $\vert \Pi_G \vert$ are close. In this paper, we answer positively to this open question.

A bunch of results on PMCs deals with their number in function of $n$. Fomin and Villanger presented the two most recent ones. In the first one~\cite{FoVi08}, they prove that $\vert \Pi_G \vert = O(1.7549^n)$ and, moreover, the list $\Pi_G$ can be computed in time $O(1.7549^n)$, matching tightly this upper bound. In the second one~\cite{FoVi10}, they establish exactly the same contributions, but with the new bound $O(1.7346^n)$. Since the computation of PMCs is the bottleneck for \textsc{Treewidth}, this latter result offers the possibility to solve it in exponential time $O(1.7346^n)$. This is, to the best of our knowledge, the lower running time for \textsc{Treewidth} if we look at exact exponential-time algorithms, function of $n$. Unfortunately, this enumeration algorithm for PMCs admits some limits in practice, since it may take time $\Theta(1.7346^n)$ while the size of $\Pi_G$ is much smaller. Moreover, this bound is probably not tight. For these reasons, the importance of designing an output-linear time algorithm for PMC enumeration is reinforced.

While Fomin and Villanger studied deeply the enumeration of PMCs through extremal combinatorics, they also identified other applications~\cite{FoKrToVi08,FoVi08,FoVi10}. One of their interesting results follows~\cite{FoVi10}.
Given a graph $G$, an integer $t \ge 0$, and sets $\Delta_G$ and $\Pi_G$, one can compute the largest induced subgraph of treewidth at most $t$ in time $O(n^{O(t)}(\vert \Pi_G \vert + \vert \Delta_G \vert))$. We call this problem \lis{t}.
Since enumerating PMCs can be done in time $O(n^2m\vert \Delta_G \vert^2)$ as written above, this result can be restated as function of the input graph only. Given a graph $G$ and an integer $t \ge 0$, \lis{t}\ can be solved in time $O(n^{O(t)}\vert \Delta_G \vert^2)$. Again, our main contribution to PMC enumeration will allow us to improve this statement. 

Furthermore, a few years later, this result was generalized to a larger set of problems in~\cite{FoToVi14}. In fact, given a graph $G$ and its set of PMCs $\Pi_G$, an integer $t \ge 0$, and a CMSO\footnote{Counting Monadic Second Order logic.} formula $\varphi$, one can compute in time $O(n^{O(t)}\cdot f(t,\varphi) \cdot \vert \Pi_G \vert)$ the largest set of vertices $X \subseteq F$ such that $F$ has treewidth at most $t$ and $(G[F],X) \models \varphi$, where $f$ is some exponential function. Behind this rather abstract formulation - let us call it \cmso{t} - we can find many problems, classical ones as \textsc{Maximum Independent Set} and \textsc{Maximum Induced Forest} and more exotic ones as \textsc{Independent Cycle Packing} and \textsc{Maximum Induced Subgraph with $\le \ell$ copies of $p$-cycles}.

In 2014, a major algorithmic impact of PMCs came with the design of a polynomial-time algorithm for \textsc{Maximum Weight Independent Set} (MWIS) on $P_5$-free graphs by Lokshtanov, Vatshelle and Villanger~\cite{LoVaVi14}. The algorithm consists in two steps: first the computation of a polynomial-size family $\mathcal{F}$ of particular PMCs, second a dynamic programming procedure based on $\mathcal{F}$ to retrieve the maximum-weighted independent set. Inspired by this technique, the authors of~\cite{GrKlPiPi19} proposed a polynomial-time algorithm for MWIS on $P_6$-free graphs. This breakthrough strengthened the general belief that MWIS is polynomial-time solvable on $P_k$-free graphs, for any constant $k$. Very recently, the authors of~\cite{ChMcPiPiRz24} designed a similar technique to tackle \cmso{t} on $P_6$-free graphs in polynomial-time for constant $t$. Furthermore, $K_{2,q}$-induced-minor-free graphs were characterized thanks to PMCs, enabling MWIS to be solved in polynomial time in this graph class~\cite{DaMiSt24}.

Another application of PMCs is the study of the \textsc{Minimum fill-in} problem: can a graph $G$ become chordal with the add-on of at most $k$ edges? This question is naturally related to PMCs as it deals with triangulations. As for \textsc{Treewidth}, Bouchitté and Todinca~\cite{BoTo01} proved that, given a graph $G$ with its lists of minimal separators $\Delta_G$ and PMCs $\Pi_G$, we can solve \textsc{Minimum fill-in} in time $O(n^2 \cdot \vert \Delta_G \vert \cdot \vert \Pi_G \vert )$. This was later improved~\cite{FoKrToVi08} with a running time $O(n^3 \vert \Pi_G \vert)$, depending only on list $\Pi_G$. Again, PMC enumeration stands as a bottleneck for solving \textsc{Minimum fill-in}. From a parameterized point of view, an exact subexponential parameterized algorithm was proposed~\cite{FoVi12}, with a running time $O(2^{O(\sqrt{k}\log k)} + k^2 nm)$. It enumerates, as a subroutine, a list of certain PMCs called \emph{vital} PMCs.

From a more practical viewpoint, Korhonen presented a few years ago an experimental study of the enumeration algorithms based on PMCs in~\cite{Ko20}. It confirmed that going through PMC enumeration is among the most effective methods to solve \textsc{Treewidth}, outperforming declarative approaches. Furthermore, the PMC-based algorithm evoked earlier for \textsc{Minimum fill-in} obtains the best performance on benchmarks, defeating integer-programming alternatives. The latter conclusion also holds for PMC-based approaches on other applications: computing generalized and fractional \emph{hypertreewidth}, identifying the \emph{perfect phylogeny}, or determining the \emph{treelength}. As a last remark, another work proposed alternative practical techniques to compute the list of PMCs, using SAT and ASP declarative solvers~\cite{KoBeJa19}.

\subsection{Enumeration algorithms}

Let us dive now into the details of PMC enumeration. Several results have been stated above, however the algorithm \btenum\ \cite{BoTo02} stands as the reference for enumerating PMCs in general graphs, even in practice~\cite{Ko20}. We recall that its time complexity is $\tilde{O}(\vert \Delta_G \vert^2)$, without any guarantee on the delay between two solution outputs (all PMCs might be returned to the very end of the execution of the algorithm). \btenum\ uses simple but insightful observations on PMCs. It is a recursive algorithm: given a graph $G$, it begins with a recursive call on graph $G \setminus \{a\}$ for some arbitrary vertex $a \in V$. Indeed, given the list of PMCs of  $G \setminus \{a\}$, one can obtain all PMCs of $G$ following certain rules, recalled below.

On one hand, Bouchitté and Todinca~\cite{BoTo02} showed that the PMCs of $G$ containing $a$ can be written either $S \cup \{a\}$ or $\Omega' \cup \{a\}$, where $S$ is a minimal separator of $G$ and $\Omega'$ is a PMC of $G \setminus \{a\}$. Since the enumeration of minimal separators is achieved in output-linear time~\cite{BeBoCo99} and that $\Delta_G \le n\vert \Pi_G \vert$, checking whether all such sets are PMCs can be done in output-linear time after the recursive call. On the other hand, keeping identical notation, the PMCs which do not contain $a$ can be written either $\Omega'$ or $S \cup (C \cap T)$, where the additional variable $T$ is also a minimal separator of $G$ and $C$ is a connected component of $G \setminus S$ fully adjacent to $S$. The time complexity of \btenum\ comes from this last costly case, where the check of all such sets  $S \cup (C \cap T)$ takes $\tilde{O}(\vert \Delta_G \vert^2) = \tilde{O}(\vert \Pi_G \vert^2)$. \btenum\ highlights the fact that PMCs can be enumerated in output-polynomial time, and more precisely in output-quadratic time. Handling differently the PMCs provided by this last costly case is the strategy we will adopt to decrease the running time for PMC enumeration.

Some alternative results on specific graph classes were proposed to decrease this output-quadratic running time. Several graph classes admit a number of PMCs which is polynomial in $n$: weakly chordal, polygon-circle, circular-arc, and $d$-trapezoid~\cite{BoTo99,FoToVi14}. \btenum\ guarantees that, on graph classes with a number of PMCs (and hence minimal separators) polynomial in $n$, its running time is polynomial in $n$ too.
Then, a very recent contribution was proposed in~\cite{GrKoTaZa25}: PMCs can be enumerated in polynomial delay on triconnected planar graphs, which gives thus an output-linear time on this large family of graphs.
From a more parameterized viewpoint, two PMC enumeration algorithms were proposed~\cite{FoLiMoTo14} where the running time depends exponentially on respectively two parameters: \emph{vertex cover} and \emph{modular width}.

\btenum\ might use a quantity of memory exponential in $n$, since it is necessary to store all PMCs of $G \setminus \{a\}$ but also all minimal separators of $G$ to return the final set $\Pi_G$. It was shown very recently~\cite{BrCoLiPu24} that it is possible to modify \btenum\ so that its running time stays output-polynomial while the memory used in the process is polynomial in $n$.

A natural technique in enumeration consists in looking at the \emph{extension version} of the enumeration problem: if this extension version can be solved in polynomial time, then it gives an enumeration algorithm with polynomial delay and space~\cite{ReTa75}. Such an approach - called \emph{Flashlight Search}, but also sometimes \emph{Backtrack} or \emph{Binary Partition} - could be tried for PMC enumeration. The following problem is called \textsc{PMC Extension}: given a graph $G$, a \emph{mandatory} set $M \subseteq V(G)$ and a \emph{forbidden} one $F \subseteq V(G)$, is there a PMC $\Omega \in \Pi_G$ satisfying $M \subseteq \Omega$ and $\Omega \cap F = \emptyset$? Observe that \textsc{PMC Extension} is thus a decision problem. Being able to solve it in polynomial time would offer the possibility to ``flashlight'' PMCs, {\em i.e.} construct a binary decision tree, where each node represents a suitable partition $(M,F)$ and each leaf is related to a unique PMC. Flashlight Search works for the enumeration of minimal separators in polynomial delay and space~\cite{Ta10}. Unfortunately, we will show that such an approach will unlikely not work for PMC, strengthening our choice to propose different techniques.

\subsection{Contributions and consequences}

Our main contribution is the proposal of an output-linear time algorithm enumerating all PMCs of a graph.

\begin{restatable}{theorem}{mainResult}
	\label{th:main}
	Given a graph $G$ with $n = \vert V(G) \vert$, one can compute the list of all PMCs of $G$ in output-linear time $\tilde{O}(\vert \Pi_G \vert) = O(n^{4}m\vert \Pi_G \vert)$.
\end{restatable}

Many literature results mentioned earlier looks like: given a graph $G$ with its lists of PMCs $\Pi_G$, one can solve a problem $\mathcal{P}$ in time $O(\vert \Pi_G \vert f_{\mathcal{P}}(n))$, where $f_{\mathcal{P}}$ is some function depending on $\mathcal{P}$. Thanks to \autoref{th:main}, we can now get rid of this assumption that set $\Pi_G$ is given. We can simply list the PMCs in output-linear time and then set $\Pi_G$ is available, without affecting asymptotically the running time for problem $\mathcal{P}$.

For example, as the time complexity of the enumeration of PMCs used to be in $\tilde{O}(\vert \Delta_G \vert^2)$, we already had that the treewidth can be computed in $\tilde{O}(\vert \Delta_G \vert^2)$, but this was potentially quadratic in the number of PMCs, since we have $\vert \Pi_G \vert = \tilde{O}(\vert \Delta_G \vert)$ on certain graph classes. Our output-linear time enumeration of PMCs thus allows us to deduce an exact algorithm for treewidth whose running time depends linearly on the number of PMCs. The same observation holds for \textsc{minimum fill-in}, \lis{t}\ and \cmso{t}. Combining \autoref{th:main} with contributions in~\cite{FoKrToVi08,FoToVi14,FoVi10}, we obtain the following corollary.

\begin{corollary}
Given a graph $G$, an integer $t \ge 0$ and a CMSO formula $\varphi$, we have:
\begin{itemize}
\item \textsc{treewidth} and \textsc{minimum fill-in} can be solved in $\tilde{O}(\vert \Pi_G \vert)$,
\item \lis{t} can be solved in $O(n^{O(t)}\vert \Pi_G \vert)$,
\item \cmso{t} can be solved in $O(n^{O(t)}\cdot f(t,\varphi) \cdot \vert \Pi_G \vert)$,
\end{itemize}
without the need to possess initially the list of PMCs.
\label{co:treewidth}
\end{corollary}

Let us describe briefly our approach for \autoref{th:main}. Our algorithm, called \ALGO , reuses some principles of \btenum\ while developing a new technique which allows to decrease the overall running time. It is also recursive: it enumerates all PMCs by launching first a recursive call on a subgraph $G \setminus \{a\}$ with an arbitrary vertex $a \in V$. The PMCs of $G$ containing $a$ - formed like $S \cup \{a\}$ or $\Omega' \cup \{a\}$ - are listed as in \btenum . Also, as in \btenum , all PMCs of $G$ which were PMCs of $G \setminus \{a\}$ are enumerated easily thanks to the recursive call. The major difference between both algorithms lies in the way to enumerate PMCs of $G$ which are neither PMCs of $G \setminus \{a\}$ nor containing $a$. Instead of listing brute-forcely all pairs of minimal separators of $G$, we introduce a novel and completely different technique to achieve this task: we execute a traversal of some ``PMC-graph'', whose nodes are PMCs and edges define a notion of neighborhood between PMCs. This approach was inspired by the \emph{Proximity Search} framework~\cite{CoUn19}.

More concretely, we define a neighborhood relation $\mathcal{N} \in \Pi_G \times \Pi_G$ which associates any PMC with a collection of neighbors. Relation $\mathcal{N}$ is fast to handle, since any PMC admits at most $n^2$ neighbors and, moreover, these neighbors can be computed in time polynomial in $n$. The PMC-graph mentioned above is then the graph $(\Pi_G,\mathcal{N})$. Naturally, this graph is partially known even after exploiting the recursive call, since otherwise all PMCs would have already been listed. Our idea consists of traversing the PMC-graph, starting from the 
PMCs already enumerated (PMCs containing $a$ or PMCs of $G \setminus \{a\}$), exploring their neighbors through $\mathcal{N}$ and so on. Said differently, we launch a \emph{Breadth-First Search} (BFS) whose initial queue contains the PMCs already enumerated, and the edges visited are given by relation $\mathcal{N}$. We show that this procedure allows us to attain any PMC of $G$.

By the way, a large part of the paper is dedicated to the proof of this latter statement. An essential tool to complete this proof consists, at some moment, in deciding whether some given set $X$ contains (or not) a PMC $\Omega \subseteq X$, and, if yes, producing such $\Omega$. We call this problem \textsc{PMC Superset}. Observe that \textsc{PMC Superset} is a special case of \textsc{PMC Extension} defined above, fixing $M = \emptyset$. We show that \textsc{PMC Superset} is polynomial-time solvable.

\begin{restatable}{theorem}{pmcSuperset}
	\label{th:pmc-superset}
	There is an algorithm deciding \textsc{PMC Superset} in time $O(n^3 m)$. Furthermore, for a yes-instance, it produces a solution PMC.
\end{restatable}


One may wonder what is the computational complexity of the more general \textsc{PMC Extension} problem, since a polynomial-time algorithm for it would imply polynomial delay for PMC enumeration, which is stronger than~\autoref{th:main}. We provide a NP-hardness proof, using a $3$-SAT reduction. This proof is preceded by another NP-hardness reduction for \textsc{Minimal $(a,b)$-separator Extension} which stands as an intermediary step.

\begin{restatable}{theorem}{hardness}
	\label{th:pmc-hardness}
	\textsc{PMC Extension} is NP-complete.
\end{restatable}

\subsection{Perspectives}

Even if our algorithm produces PMCs in output-linear time, it provides no guarantee on the delay between two outputs. In fact, we are even sure that the time taken before the return of the first PMC of $G$ cannot be less than exponential in $n$, because of the recursive nature of our procedure (as for \btenum ). Hence, finding a polynomial-delay enumeration algorithm for PMCs remains an open challenge. We wonder whether the traversal of the PMC-graph - our approach introduced to improve \btenum\ - could not be itself a polynomial-delay algorithm enumerating PMCs. Concretely, one could compute one arbitrary PMC of $G$, fix it as a single element in the starting queue of a BFS, and explore the PMC-graph using our neighborhood relation $\mathcal{N}$. For sure, the running time between two solutions would be polynomial in $n$ as all neighbors of a PMC can be determined fast. However, the main question is whether all PMCs would be enumerated with this method. We did not identify promising directions to prove this statement so far. But this BFS did not miss any PMCs on the numerous graph classes we tested experimentally, although we are blocked computationally by the brute-force checking (at most $20$ vertices). 

Another natural question is whether output-linear time can be maintained while only polynomial space is used. To reach this goal, one could think of exploring the recursion tree in a DFS manner, as done in~\cite{BrCoLiPu24}. Perhaps such an approach stands also as another promising direction to obtain polynomial-time delay.

Fomin and Villanger~\cite{FoVi08} showed that, for any integer $k \ge 1$, knowing the list $\Pi_G^{k+1}$ of PMCs of size at most $k+1$ allows checking whether the input graph $G$ has treewidth at most $k$ in time linear in $\vert \Pi_G^{k+1} \vert$.\footnote{Separators of size at most $k$ are also needed.} As a direct consequence, the existence of an FPT-output-linear time algorithm enumerating PMCs of small size would be a significant result, {\em i.e.} with running time $O(f(k)n^{O(1)}\vert \Pi_G^k \vert)$ where $f$ is an arbitrary function and $\Pi_G^k$ is the set of PMCs of size at most $k$. Indeed, it would give an algorithm deciding whether treewidth is at most $k$ in time $\tilde{O}(f(k)\vert \Pi_G^k \vert)$, which could be very efficient in practice. 

A final perspective is simply to implement our algorithm. The experimental study produced in~\cite{Ko20}, which displays numerous benchmarks of algorithms/instances, could be extended with our improvement for PMC enumeration. Such an experimental contribution would allow assessing in practice the performance gain offered by output-linear time, instead of output-quadratic one. 

\section{Preliminaries}

\input{preliminaries}

\section{PMC superset problem} \label{sec:superset}

In this section, we will focus on the following question: given a set of vertices $X \subseteq V(G)$, is there any chance that some subset of $X$ is a PMC? This question was defined as the \textsc{PMC Superset} problem in \autoref{subsec:problems}. Note that a naive algorithm would check all subsets of $X$ in time exponential in the size of $X$. Our objective is to present an algorithm for this decision problem in time polynomial in $n = \vert V(G) \vert$. This contribution will be important for showing that our final algorithm \ALGO\ finds all PMCs.

\subsection{A polynomial-time algorithm for reducing shy sets} \label{subsec:removal}

As a first step, we present a very simple and natural algorithm called \removal\ which, given a shy set $X \subseteq V(G)$ of some graph $G$, computes a subset of $X$. \removal\ will be used in the next subsections as a crucial subroutine to solve \textsc{PMC Superset}.

Before presenting this subroutine, we explain why it does not solve \textsc{PMC Superset} itself. Unfortunately, the subset computed by \removal\ is not necessarily a PMC. This is due to this simple fact: there exists shy sets which are not PMC supersets.

\begin{figure}[h]
\centering
\begin{subfigure}{0.45\textwidth}
\centering
\scalebox{0.8}{\input{tikz/negative-shy-set}}
\caption{An example of shy set $X$ (in red) which is not a PMC and such that any of its proper subset has a full component.}
\label{subfig:negative-shy-set}
\end{subfigure}
\hspace{0.05\textwidth}
\begin{subfigure}{0.45\textwidth}
\centering
\scalebox{0.8}{\input{tikz/positive-shy-set}}
\caption{Almost the same graph (an edge was removed between $a'$ and $S$): here $X$ is still shy but is a PMC superset with $\Omega = S \cup \{a\}$.}
\label{subfig:positive-shy-set}
\end{subfigure}
\caption{Two graphs with shy sets $X$ but different decision answer to \textsc{PMC Superset}.}
\label{fig:example-superset}
\end{figure}

\autoref{subfig:negative-shy-set} provides such example. The set $X = S \cup \{a,b\}$ represented is shy as connected components $\{a'\}$ and $\{b'\}$ are not full for $X$. However, $X$ is not a PMC: it does not fulfil the reachability property because $b \notin R_G(a,X)$. Furthermore, no proper subset of $X$ is shy. Removing $a$ creates a full component: set $X' = X \setminus \{a\} = S \cup \{b\}$ admits a full component $\{b'\}$. By symmetry, the same statement holds for $X \setminus \{b\}$. Finally, for any $s \in S$, the set $X'' = X \setminus \{s\}$ admits also a full component $\{s,a',b'\}$. In summary, for any vertex $x \in X$, set $X \setminus \{x\}$ is not shy and hence is not a PMC. Since any proper subset of $X$ is a subset of some $X \setminus \{x\}$, according to \autoref{obs:shy-set}, any proper subset of $X$ is not shy and thus is not a PMC. \autoref{subfig:positive-shy-set} shows that with a tiny modification of the graph of \autoref{subfig:negative-shy-set}, the shy set $X$ is now a PMC superset. Indeed, $\Omega = S \cup \{a\}$ is a PMC as $\{a'\}$ is not a full component anymore.

Below is a technical lemma to help us analyze our algorithm \removal , presented on the next page.

\begin{lemma}\label{lem:minus-one-comp}
    Let $X \subseteq V(G)$ be any shy set of vertices of $G$. Assume that $(u,v)$ is a conflict pair for $X$. Let $X_v = X \setminus \{v\}$. The neighborhood of the connected component $C_v$ containing $v$ in $G \setminus X_v$ is: $N(C_v) = R_G(v,X)$.
\end{lemma}
\begin{proof}
    Set $N(C_v)$ can be seen as the subset of $X_v$ which contains all vertices $w \in X$, $w\neq v$, such that there exists a path from $w$ to $v$ whose internal vertices do not belong to $X$. This corresponds to the definition of reachability (\autoref{def:reachable-set}).
\end{proof}

We define the notion of \emph{brightening} of a shy set. The result of our algorithm \removal\ will be some brightening of the input shy set.

\begin{definition}[Brightening of a shy set]
Let $G$ be a graph and $X \subseteq V(G)$ be a shy set of $G$. A subset $X' \subseteq X$ is a \emph{brightening} of $X$ if one of the following statements hold:
\begin{itemize}
\item $X'$ is a PMC,
\item $X'$ admits only \emph{native full components}, {\em i.e.} each connected component $C$ of $G \setminus X'$ which is full for $X'$ does not contain any vertex removed from $X$: formally, $C \cap (X \setminus X') = \emptyset$.
\end{itemize}
\label{def:brightening}
\end{definition}

We introduce our algorithm \removal. Its principle is very simple: consider the input shy set $X$. If $X$ does not admit conflicts (\autoref{def:reachable}), then $X$ is a PMC and we return it. Otherwise, we select an arbitrary vertex $u \in X$ - that we call \emph{protected vertex} - which is in conflict with another one, and we remove from $X$ all the vertices which are in conflict with the selected $u$. If the obtained set is still shy, we repeat the process, otherwise we return the obtained set. The pseudocode of \removal\ is given in \autoref{alg:removal}.

\begin{algorithm}[h]
	\SetKw{And}{and}
    \SetKwFor{For}{for}{do}{\nl endfor}
    \SetKwFor{Forall}{for all}{do}{\nl endfor}
    \SetKwIF{If}{ElseIf}{Else}{if}{then}{else if}{else}{}
    \DontPrintSemicolon
    \SetNlSty{}{}{:}
    \SetAlgoNlRelativeSize{0}
    \SetNlSkip{1em}

 	\nl \KwIn{A graph $G$ and a shy set $X \subseteq V(G)$.}
 	\nl \KwOut{A subset of $X$ which is a brightening of $X$ (proven in~\autoref{th:brightening}).}

%
%

	\nl $U \leftarrow X$; \label{line:removal-starts}\;
 	\nl \While{$U$ has at least one conflict pair \label{line:while-removal}}{
 		\nl $u \leftarrow$ any vertex of $U$ involved in a conflict pair for $U$;\label{line:selection-removal}\; 
 		\nl \Forall{$v \in U$ such that $(u,v)$ is a conflict pair for $U$\label{line:for-removal}}{
 			\nl $U \leftarrow U \setminus \{v\}$;\label{line:minus-one-removal}\;
 		}
 	}
 	
 	\nl \Return{} $U$;\label{line:return-removal}\;
 	\caption{\removal}
 	\label{alg:removal}
 \end{algorithm}
 
\removal\ starts in fact by checking whether the given set is a PMC (line~\autoref{line:while-removal} of \autoref{alg:removal}). If it is, it does not enter in the while loop and returns $X$ which is thus a brightening of itself, by \autoref{def:brightening}. Otherwise, it enters the while loop and the removal process starts (lines~\autoref{line:selection-removal}-\autoref{line:minus-one-removal}). It consists in removing from $U$ all vertices in conflict with some arbitrarily fixed protected vertex $u \in U$. When the removal process is achieved, the condition of the while loop is checked: we verify whether the obtained set $U$ is still shy. If no, we return $U$ (line~\autoref{line:return-removal}). If it is still shy, we restart the removal process until no conflict pair remain. Refer to \autoref{fig:removal} as a visual aid on an iteration of \removal.
 
\begin{theorem}[Existence of a brightening]
For any shy set $X \subseteq V(G)$, the subset returned by \removal\ is a brightening of $X$.
As a consequence, each shy set admits at least one brightening, and, it can be computed in time $O(n^3)$.
\label{th:brightening}
\end{theorem}
\begin{proof}
First, we prove that the result of \removal\ is a brightening of the input shy set $X$. 
To achieve it, we propose a loop invariant for the while loop in line~\autoref{line:while-removal}. We index the iterations of this loop by index $i \ge 1$, the first iteration being the one where $U$ is initially equal to $X$ and some of its vertices are withdrawn. For each iteration $i \ge 1$, we fix $u_i$ as the protected vertex selected in line~\autoref{line:selection-removal}, and $U_i$ the set $U$ obtained at the end of the loop. To be consistent, we also fix $U_0 = X$. Our invariant proposal is the following:

\begin{center}
\textbf{Invariant}: At iteration $i \ge 1$, for any non-native connected component $C$ of graph $G \setminus U_i$, there exists an enlightened vertex $z_C \in U_i$ which is not visible from $C$, {\em i.e.} $z_C \notin N_G(C)$.
\end{center}

We need to provide some explanations on this invariant. By ``non-native connected component'' $C$, we mean a connected component of $G \setminus U_i$ which contains at least one vertex of $X \setminus U_i$. In other words, $C$ is a connected component of $G \setminus U_i$ but not of $G \setminus X$. Moreover, we recall that an enlightened vertex $z_C \in U_i$ is a vertex visible from any other vertex of $U_i$ despite $U_i$, in brief $R_G(z_C,U_i) = U_i$. See \autoref{def:reachable} for an exhaustive reminder on reachability vocabulary. We prove this invariant by induction, on the index $i$.

\textit{Base case.} The base case is $i = 1$. During the first iteration, we select a vertex $u_1 \in U_0$ which is involved in at least one conflict pair, and then we remove all vertices in conflict with $u_1$ for $U_0$. We denote by $v_0,v_1,\ldots,v_r$ the sequence of removed vertices in this order, with $r \ge 0$. We have $\{v_0,\ldots,v_r\} = U_0 \setminus U_1$. When $v_0$ is removed from the current $U = U_0$, this forms a new connected component $C_0$ in $G \setminus \left(U_0 \setminus \{v_0\}\right)$ containing $v_0$. Vertex $u_1$ is not adjacent to $C_0$: since $(u_1,v_0)$ is a conflict pair, $v_0 \notin R_G(u_1,U_0)$ and by \autoref{lem:minus-one-comp}, $u_1 \notin N_G(C_0)$. Again, when we remove $v_1$ from the current $U = U_0 \setminus \{v_0\}$, we form a new connected component $C_1$ in graph $G \setminus \left(U_0 \setminus \{v_0,v_1\}\right)$. Either $C_1$ is distinct from $C_0$: again $u_1 \notin N_G(C_1)$ as $(u_1,v_1)$ is a conflict pair (\autoref{lem:minus-one-comp}). Or $C_1$ contains $v_0$ and thus $C_0 \subsetneq C_1$: as $C_0$ did not see $u_1$ in $G \setminus \left(U_0 \setminus \{v_0\}\right)$ and $(u_1,v_1)$ is a conflict pair, necessarily $u_1 \notin N_G(C_1)$. We can pursue the reasoning: for each newly-created connected component $C_j$, with $0\le j\le r$, either $C_j$ is distinct from its predecessors or it is merged with at least one of them. In both cases, since $u_1$ is not adjacent to the predecessors and as $(u_1,v_j)$ is a conflict pair, then $u_1 \notin N_G(C_j)$. In summary, for each non-native connected component $C \in \{C_j\}_{0\le j\le r}$ created during the first iteration, vertex $u_1$ is not visible from $C$ despite $U_1$. Furthermore, as all vertices in conflict with $u_1$ for $U_0$ belong to $U_0 \setminus U_1$, vertex $u_1$ is enlightened for $U_1$.

\begin{figure}[h]
	\centering
	\begin{subfigure}{0.31\textwidth}
	\centering
	\scalebox{0.9}{\input{tikz/removal}}
	\end{subfigure}
	\begin{subfigure}{0.31\textwidth}
	\centering
	\scalebox{0.9}{\input{tikz/removal2}}
	\end{subfigure}
	\begin{subfigure}{0.34\textwidth}
	\centering
	\scalebox{0.9}{\input{tikz/removal3}}
	\end{subfigure}
	\caption{The first iteration of \removal\ on some example: set $X = \{x_1,\ldots,x_6\}$ is shy, no connected component of $G \setminus X$ is fully adjacent to $X$. There are several conflict pairs for $U = X$: $(x_1,x_3),(x_1,x_4),(x_1,x_5),(x_2,x_6),(x_3,x_6)$. We select arbitrarily $u_1 = x_1$ as the protected vertex, hence we remove successively vertices $x_3,x_4,x_5$. We remove first vertex $v_0 = x_3$: we obtain a non-native connected component $C_0$ which is not full for $X \setminus \{x_3\}$ as it does not see $x_1$. Note that $X \setminus \{x_3\}$ admits another connected component which is native. Second, we remove $v_1 = x_4$ and obtain a connected component $C_1$ which is merged with $C_0$: again, it does not see $x_1$. Removing finally $v_2 = x_5$ makes us obtain a shy subset without conflict pairs: $\{x_1,x_2,x_6\}$ is a PMC.}
	\label{fig:removal}
\end{figure}

\textit{Induction step.} We consider now any index $i\ge 1$. We assume the invariant true for $i\ge 1$ and prove it for $i+1$. As $U_{i+1}$ exists, then we know that $U_i$ has a conflict pair. Set $U_i \setminus U_{i+1}$ is thus nonempty and contains all vertices in conflict with $u_{i+1}$ for $U_i$. Let $C$ be a non-native connected component of $G \setminus U_{i+1}$. We distinguish two cases.

\textbf{Case A}: $C$ was already a connected component of $G \setminus U_i$, said differently $C \cap (U_i \setminus U_{i+1}) = \emptyset$. In this case, $C$ is non-native as it contains vertices from $X \setminus U_i$ but not from $U_i \setminus U_{i+1}$. By induction hypothesis on index $i$, an enlightened vertex $z_C \in U_i$ is not visible from $C$. As $z_C$ is enlightened, it is not involved in any conflict pair of $U_i$. So, $z_C$ is still enlightened in $U_{i+1}$.

\textbf{Case B}: $C$ contains vertices from $U_i \setminus U_{i+1}$. Here, the reasoning is identical to the one provided for the base case. At each removal of a vertex in conflict with $u_{i+1}$, either the newly-created connected component is distinct from its predecessors in iteration $i+1$, or it contains one of its predecessors. In both cases, as the vertex removed is in conflict with $u_{i+1}$ and that predecessors are too, we maintain the property that the connected components created by this vertex removals do not see $u_{i+1}$. In brief, $u_{i+1} \notin N_G(C)$. Furthermore, vertex $u_{i+1}$ is enlightened for $U_{i+1}$ as all its conflicts for $U_i$ belong to $U_i \setminus U_{i+1}$. 

\textit{Conclusion}. The invariant property that any non-native connected component $C$ does not see some enlightened vertex of $U_i$ is maintained. Moreover, \removal\ necessarily stops because set $U$ is finite and, at each iteration, at least one vertex is removed until no conflict pairs remains. As a consequence, the subset $U \subseteq X$ returned by \removal\ exists, is nonempty, and all non-native components of $U$ are not full components for $U$. This means that there are two possibilities: either $U$ has no full component, and, it is a PMC, or $U$ admits a full component which is a native one. In brief, $U$ is a brightening (\autoref{def:brightening}).

\textit{Complexity time.} At each iteration of the while-loop, we need to update the collection of conflict pairs for the current set $U_i$. This operation amounts to determine graph $\torso(G,U_i)$, which can be done in $O(m+\vert U_i \vert^2) = O(n^2)$. Consequently, as the total number of iterations is at most linear in $n$, the time complexity of \removal\ is in $O(n^3)$.
\end{proof}

Let us come back to the example given in \autoref{subfig:negative-shy-set}. As stated above, we could show that even if $X$ is a shy set, none of its proper subset is a PMC. Consequently, set \removal$(X)$ will not be a PMC and will admit only native full components. Indeed, the single conflict pair for $X$ is $(a,b)$. Assume w.l.o.g. that the vertex selected in line~\autoref{line:selection-removal} is $u = a$. Then we remove $b$ from the current set: $U = X \setminus \{b\} = S \cup \{a\}$. Now, set $U$ has no conflict pair since it admits a full component. Observe by the way that it has a unique full component $\{a'\}$. This component is indeed native: $\{a'\}$ was already a connected component of $G \setminus X$. As an example, the connected component $\{b,b'\}$ of $G \setminus U$ - which is not full - was not a connected component of $G \setminus X$: it is non-native.

On the contrary, in \autoref{subfig:positive-shy-set}, the set $X$ is a PMC superset, as $S \cup \{a\}$ is a PMC. However, \removal\ can fail in finding a PMC subset of $X$. Indeed, consider the case where the selected vertex is $u = b$ (in conflict with $a$). The obtained set after an iteration is $U = S \cup \{b\}$ which admits a (native) full component. Nevertheless, note that if we had selected $u = a$, then we would have obtained $U = S \cup \{a\}$. 

Both remarks on the examples of \autoref{fig:example-superset} ensure us that the algorithm \removal\ is definitely not a good candidate for identifying a PMC which is the subset of some given shy set $X$. However, this subroutine will be necessary in the remainder for establishing the main contributions of this paper.

\subsection{Shy separators}

As stated in the previous subsection, the algorithm \removal\ is not sufficient to show that \textsc{PMC Superset} can be solved in polynomial time. We introduce a concept which is essential to achieve this goal: \emph{shy separators}.

\begin{definition}[Shy separators]
    Consider a graph $G$ with a set $X \subseteq V(G)$ and $S \subsetneq X$ a minimal separator of $G$. Let $C$ be a full component for $S$ in $G \setminus S$. 
    
	We say that $S$ is an $X$-\emph{shy separator} with regard to $C$ if, for any connected component $B \subseteq C$ of graph $G \setminus X$, then $N_G(B) \cap S \neq S$.
    \label{def:shy-sep}
\end{definition}

This definition is based upon some vertex set $X$ (which represents a possible PMC superset in our context), a minimal separator $S$ of $G$ included into $X$, and a full component for $S$. For some $X$, a separator $S \subsetneq X$ can be $X$-shy with regard to a full component $C_1$ whereas it is not for another full component $C_2$. As a specific example, assume $X = S \cup C$: in this case, $S$ is naturally $X$-shy with regards to $C$ as graph $G \setminus X$ does not contain any connected component included into $C$, since $C \subsetneq X$.

We will see that, for some vertex set $X$, being a PMC superset is directly related to the presence of some $X$-shy separator. First, we establish a link between PMCs and shy separators.

\begin{lemma}\label{lem:no_full_comp_for_all_sep}
    For any PMC $\Omega$, there is a minimal separator $S \subsetneq \Omega$ of $G$ such that: 
\begin{itemize}
    \item $T = \Omega \setminus S$ is included into a connected component $C$ of $G \setminus S$ full for $S$: $N_G(C) = S$,
    \item $S$ is $\Omega$-shy with regard to $C$.
\end{itemize}
\end{lemma}

\begin{proof}
	We recall that for any PMC $\Omega$, given a minimal separator $S \subsetneq \Omega$, the set $\Omega \setminus S$ is included into one single full component of $S$ (\autoref{lem:PMC_in_one_FC_of_S}). We proceed by proving the following claim. 
	
	\textit{Claim}: Given a PMC $\Omega$, one of its minimal separator $S \subsetneq \Omega$ and $C$ the full component of $S$ containing $T = \Omega \setminus S$, at least one of the two possibilities holds:
	\begin{itemize}
	\item either there is no connected component $B$ of $G[C \setminus T]$ such that $S \subseteq N_G(B)$,
	\item or there is another minimal separator $S' \subsetneq \Omega$ such that $S \subsetneq S'$.
	\end{itemize}

	We start by proving this claim. Assume that the first item is not satisfied: there is a connected component $B$ of $G[C \setminus T]$ full for $S$. Hence, $S \subseteq N(B) \subsetneq \Omega$.
	
	At least one vertex $u \in T$ is adjacent to $B$ as $B \subseteq C$ and $C$ is a connected component of $G \setminus S$. Otherwise, $B$ would be a full component for $S$ different from $C$. Formally, $B \subseteq R_G^+(u,\Omega)$ or more concisely $u \in N(B)$. Meanwhile, $B$ obviously cannot see all the vertices of $\Omega$ as it would contradict that $\Omega$ is a PMC (and thus a shy set). It follows that there exists a vertex $v \in T$ which is not adjacent to $B$: $B \cap R_G^+(v,\Omega) = \emptyset$. 
	
\begin{figure}[h]
	\centering
	\begin{subfigure}{0.49\textwidth}
	\centering
	\scalebox{0.7}{\input{tikz/shy-sep}}
	\caption{An example of graph $G$ with a PMC $\Omega = S \cup T$ surrounded by red lines. The connected component $B$ of $G[C \setminus T]$ is full for $S$, hence $S$ is not $\Omega$-shy with regards to $C$. Separator $S' = N(B)$ - surrounded by yellow lines - is also a minimal separator of $G$, included into $\Omega$, satisfying $S \subsetneq S'$.}
	\end{subfigure}~
	\hspace{0.01\textwidth}
	\begin{subfigure}{0.49\textwidth}
	\centering
	\scalebox{0.7}{\input{tikz/shy-sep-2}}
	\caption{The new decomposition $(S',T')$ with separator $S'$ and $T' = \Omega \setminus S'$. Here, $C'$ represents the (unique) connected component of $G \setminus S'$ containing $T'$, observe that $B \cap C' = \emptyset$. No component of $G[C' \setminus T']$ is full for $S'$, so $S'$ is $\Omega$-shy with regard to $C'$.}
	\end{subfigure}
	\caption{Illustration of the proof: how $S$ can be extended if it does not satisfy \autoref{lem:no_full_comp_for_all_sep}.}
	\label{fig:shy-sep}
\end{figure}
	
	We fix $S' = N_G(B)$. Remember that $S'$ is a minimal separator of $G$, according to~\autoref{le:neighbor-sep}, because $B$ is a connected component of $G \setminus \Omega$. As stated above, $S \subsetneq N(B)$ and also $u \in N(B)$. Finally, $S \subsetneq S'$ and $S' \subsetneq \Omega$. \autoref{fig:shy-sep} illustrates this reasoning with an example.
	
We conclude about the statement of the lemma.  The claim above provides us with a procedure to build the larger minimal separator with the required property. Start with an arbitrary minimal separator $S \subsetneq \Omega$. By the claim, we know that either $S$ admits no full component in $C \setminus T$ (first item) - in this case, return $S$ - or we can augment $S$ into a larger minimal separator $S' \subsetneq \Omega$. Restart the same process on $S'$ and future augmented separators while the lemma property is not satisfied. The process necessarily stops since each separator is augmented by at least one vertex at each step meanwhile being included into $\Omega$. 
\end{proof}

The previous lemma allows us to define what is a \emph{shy decomposition} of a PMC $\Omega$, \emph{i.e.} the bi-partition between some $\Omega$-shy separator and the remainder of the PMC. Such shy decomposition can be obtained by picking any minimal separator of $\Omega$ and then applying the procedure of \autoref{lem:no_full_comp_for_all_sep} until obtaining an $\Omega$-shy separator $S$ with regard to the component containing $\Omega \setminus S$.

\begin{definition}[Shy decomposition]
Given a graph $G$ and one of its PMCs $\Omega$, a pair $(S,T)$ is a \emph{shy decomposition} of $\Omega$ if it satisfies the following properties:
\begin{itemize}
\item $(S,T)$ is a partition of $\Omega$: $S \cup T = \Omega$ and $S \cap T = \emptyset$,
\item $S$ is a minimal separator of $G$ which is $\Omega$-shy with regard to $C$, where $C = C_G(t,S)$ with $t \in T$.
\end{itemize}
\label{def:shy-decomp}
\end{definition}

\autoref{lem:no_full_comp_for_all_sep} guarantees that each PMC admits at least one shy decomposition. Before going further on \textsc{PMC Superset}, we state an important connection between the concepts of shy separators and active separators. The following result is used in \autoref{sec:linear}.

\begin{lemma}
Consider a graph $G$ with a PMC $\Omega$ which admits an active separator $S$. Let $C$ be the connected component of $G \setminus S$ containing $T = \Omega \setminus S$. Then, separator $S$ is $\Omega$-shy with regard to $C$. 
\label{le:shy-active}
\end{lemma}
\begin{proof}
Assume by way of contradiction that $S$ is not $\Omega$-shy with regard to $C$. There exists a connected component $B \subsetneq C$ in $G \setminus \Omega$ which is full for $S$: $S \subseteq N(B)$. Hence, for any pair $x,y \in S$, set $T$ does not separate pair $(x,y)$ in graph $G[C \cup \{x,y\}]$: indeed, there is a $(x,y)$-path whose internal vertices belong to $B$, as both $x,y \in N(B)$. This contradicts the fact that $S$ is active (\autoref{def:active-sep}).
\end{proof}

\subsection{A characterization of PMC supersets} \label{subsec:charac-pmc}

The presence of a shy separator is the key witness to assess whether some shy set is a PMC superset. We begin with ``if'' part: if some set $X$ is a PMC superset, then there is necessarily a shy separator $S \subsetneq X$. The formal statement is provided below.

\begin{theorem}\label{thm:direct_superpmc_charac}
Let $X \subseteq V(G)$ be a shy set. If $X$ is a PMC superset, then there exists a triplet $(a,b,S)$ satisfying the following properties:
\begin{itemize}
\item $a$ and $b$ are two distinct vertices and in particular $a \in X$,
\item $S$ is the $X$-close $(a,b)$-separator,
\item $S$ is an $X$-shy separator with regard to $C_G(a,S)$.
\end{itemize}
\end{theorem}
\begin{proof}
    As $X$ is a PMC superset, there exists a PMC $\Omega \subseteq X$. According to \autoref{lem:no_full_comp_for_all_sep}, $\Omega$ admits a shy decomposition $(S',T')$, where separator $S'$ is $\Omega$-shy with regards to $C'$, which is the connected component in $G \setminus S'$ containing $T'$. 
We fix arbitrarily $a \in T \subseteq C'$ and $b$ as a vertex belonging to any other full component for $S'$ in $G$. By definition of both of these vertices, $S'$ is a minimal $(a,b)$-separator, and $S' \subsetneq X$. See \autoref{fig:pmc-super-shy} for an illustration.
    
    Let $S$ be the $X$-close $(a,b)$-separator. Since $S' \subsetneq X$, from \autoref{lem:parallel-close-x}, $S$ and $S'$ are parallel and $C_G(a,S) \subseteq C_G(a,S') = C'$. We fix $C = C_G(a,S)$. Assume by way of contradiction that $S$ is not $X$-shy with regards to $C$. Then, there exists a connected component $B$ of $G[C \setminus X]$ which is fully adjacent to $S$. As $B \subseteq C \setminus X$, we also have $B \subseteq C' \setminus \Omega$ because $C \setminus X \subseteq C' \setminus X \subseteq C' \setminus \Omega$, using the inclusion properties stated so far. Let us call $B'$ the connected component of $G[C' \setminus \Omega]$ containing $B$. We show that $S' \subseteq N_G(B')$. 
    
    \begin{figure}[h]
	\centering
	\scalebox{0.75}{\input{tikz/pmc-super-shy}}
	\caption{Illustration of separators $S$ and $S'$ on some graph, with shy set $X$ (given by the three purple areas - one big and two small with a single vertex) and PMC $\Omega = S' \cup T' \subseteq X$. Separator $S'$ is $\Omega$-shy with regard to $C'$: no connected component of $G[C' \setminus T']$ (surrounded by dashed blue lines) is full for $S'$. Set $S$ (orange vertices) is the $X$-close $(a,b)$-separator: note that it is $X$-shy with regards to $C = C_G(a,S)$, which is represented with blue vertices.}
	\label{fig:pmc-super-shy}
\end{figure}
    
Let $s \in S'$. If $s \in S' \cap S$, then $s$ is adjacent to $B \subseteq B'$ because $B$ is full for $S$. Otherwise, assume $s \in S' \setminus S$. Since $\Omega$ is a PMC of $G$ and has the reachability property, there exists an $(a,s)$-path $P_s$ whose internal vertices belong to $\overline{\Omega}$ (if there are). As $S$ is closer to $a$ than $S'$, {\em i.e.} $C_G(a,S) \subseteq C_G(a,S')$, path $P_s$ intersects $S$. Said differently, $P_s$ contains a vertex $v \in S \cap \overline{\Omega}$. Since $B$ is fully adjacent to $S$ and $v \notin \Omega$, we have $v \in B'$. Due to the presence of the internal vertex $v \in B'$ in path $P_s$ - which has internal vertices in $\overline{\Omega}$ - then all internal vertices of $P_s$ belong to $B'$ and thus $s$ is adjacent to $B'$. 

In summary, $B'$ is a connected component of $G[C' \setminus \Omega]$ full for $S'$: a contradiction since we assumed that $S'$ is $\Omega$-shy with regard to $C'$. Consequently, $S$ is $X$-shy with regard to $C$, completing the proof.
\end{proof}

Based on \autoref{subsec:removal}, we focus now on the ``only if'' part. We show, given a shy set $X$ together with a shy separator $S \subsetneq X$, how our algorithm \removal\ allows us to compute a PMC inside $X$, proving by the way that $X$ is a PMC superset.

\begin{theorem}\label{thm:indirect_PMC_superset_characterisation}
Let $X \subseteq V(G)$ be a shy set. Assume there exists a triplet $(a,b,S)$ such that:
\begin{itemize}
\item $a$ and $b$ are two distinct vertices and in particular $a \in X$,
\item $S$ is the $X$-close $(a,b)$-separator,
\item $S$ is $X$-shy with regards to $C_G(a,S)$.
\end{itemize}
Then we can identify in time $O(n^3)$ a PMC $\Omega \subseteq X$ such that $S \subsetneq \Omega \subseteq S \cup \left(C_G(a,S) \cap X\right)$.
\end{theorem}

\begin{proof}
  The set $X \cap C_G(a,S)$ is nonempty because otherwise $S$ could not be $X$-shy with regard to $C_G(a,S)$, the latter being a full component for $S$ in $G$. Consider set $Y = S \cup (X \cap C_G(a,S))$ which is a subset of $X$. Set $Y$ admits no full component in $G \setminus Y$: the components full for $S$ that do not contain $a$ are not adjacent to $X \cap C_G(a,S)$, and the connected components of $G \setminus Y$ included into $C_G(a,S)$ are not full for $S$ because $S$ is $X$-shy with regard to $C_G(a,S)$. So, $Y$ is a shy set. Either $Y$ is a PMC, and in this case, $X$ is obviously a PMC superset. Or, there is at least one conflict  pair within $Y$. We apply \removal\ on set $Y$, but we orient the choice of the protected vertex at each iteration. In particular, we maintain the fact that $S \subseteq U$, where $U$ is the variable set in line~\autoref{line:removal-starts} in \autoref{alg:removal}, $U \subseteq Y$.
  
  For the selection of the protected vertex $u$ (line~\autoref{line:selection-removal} of \autoref{alg:removal}), we apply the following rule:
  \begin{itemize}
  \item if there is at least one conflict pair $(s,v)$ for $U$ which belongs to $S \times (C_G(a,S) \cap X)$, fix $u = s \in S$,
  \item otherwise, if there are only conflict pairs in $(C_G(a,S) \cap X) \times (C_G(a,S) \cap X)$, fix $u$ arbitrarily.
  \end{itemize}
  
  In other words, we fix as a priority that the protected vertex belongs to $S$. Observe that there cannot be a conflict pair in $S \times S$ because such pairs are visible despite $Y$, seeing each other through the component $C_G(b,S)$ which does not intersect $Y$ by definition. We prove that the result of \removal $(Y)$, with this specific selection process which does not influence the time complexity, is a PMC.
  
  Let $Z = \removal (Y)$. As we privilege vertices in $S$ for the protected vertex and there are no $S \times S$ conflict pairs, we have $S \subsetneq Z$. From \autoref{th:brightening}, \removal\ returns a brightening of set $Y$, regardless the choice of protected vertices. Hence, if $Z$ admits a full component, then it is a native component for $Y$, {\em i.e.} a connected component initially present in $G \setminus Y$. As a first remark, we cannot have $Z = S$, otherwise $C_G(a,S)$ would be a non-native full component for $S$: $C_G(a,S)$ would be non-native as $Y \cap C_G(a,S) = X \cap C_G(a,S) \neq \emptyset$. Therefore, $Z \cap C_G(a,S) \neq \emptyset$. Let $B$ be a (native) full component for $Z$. Due to the previous remark, $B$ cannot be a connected component of $G \setminus S$ different from $C_G(a,S)$: such a component would not be adjacent to $Z \cap C_G(a,S) \neq \emptyset$, contradicting the fact that $B$ is full for $Z$. So, $B \subsetneq C_G(a,S)$. In this case, as $B$ is native, this means there exists a connected component of $G \setminus Y$ which is fully adjacent to $Z$, and thus in particular to $S \subsetneq Z$. Note that, by definition of $Y$, $B$ is also a connected component of $G \setminus X$ included into $C_G(a,S)$. This yields a contradiction since $S$ is $X$-shy with regard to $C_G(a,S)$. As a consequence, $Z$ does not admit any native full component: $Z$ is a PMC. 
  
 In summary, \removal\ with a priority rule on $S$-vertices produces a PMC $Z \subseteq Y \subseteq X$. The running time needed is the one of \removal , already given in \autoref{th:brightening}.
\end{proof}	

Thanks to \autoref{thm:direct_superpmc_charac} and~\autoref{thm:indirect_PMC_superset_characterisation}, we obtain a complete characterization of PMC supersets. 

\begin{corollary}
Let $X \subseteq V(G)$ be a shy set. Set $X$ is a PMC superset if and only if there exists a triplet $(a,b,S)$ satisfying the following properties:
\begin{itemize}
\item $a$ and $b$ are two distinct vertices and in particular $a \in X$,
\item $S$ is the $X$-close $(a,b)$-separator,
\item $S$ is $X$-shy with regards to $C_G(a,S)$.
\end{itemize}
Moreover, given such triplet $(a,b,S)$, one can compute a PMC $\Omega \subseteq X$ thanks to \removal\ in $O(n^3)$.
\label{co:pmc-carac}
\end{corollary}

\subsection{A polynomial-time algorithm for PMC superset}

The characterization of PMC supersets established in \autoref{co:pmc-carac} provides us with a way to check whether some shy set $X$ is a PMC superset. Indeed, it suffices to determine the existence (or not) of a triplet $(a,b,S)$ satisfying the properties listed in \autoref{co:pmc-carac}. It becomes rapidly clear that answering the latter question can be achieved in time polynomial in $n$. One can simply enumerate all possible pairs $a,b \in V(G)$ and then compute the unique $X$-close $(a,b)$-separator $S$. The last check consists in looking at the shyness of $S$, which is a simple connectivity property. All in all, we obtain a polynomial-time complexity for \textsc{PMC Superset}.

\pmcSuperset*
\begin{proof}
We begin with the research of a triplet $(a,b,S)$. First, we enumerate all pairs of vertices $(a,b) \in X \times V(G)$ and check whether $(a,b) \notin E(G)$ and $X \setminus \{a\}$ is an $(a,b)$-separator. If these conditions are not satisfied, we move on to another pair. Otherwise, we compute the $X$-close $(a,b)$-separator $S$. By \autoref{def:x-close}, constructing the $X$-close $(a,b)$-separator consists first in determining set $A = C_G(a,X \setminus \{a\})$: it can be achieved with a BFS starting at $a$ in graph $G \setminus (X \setminus \{a\})$. Second, we compute $S = \sepag{G}{A}{b}$ in linear time (\autoref{lem:close-sep}). When $S$ is known, we check whether it is $X$-shy with regard to $C = C_G(a,S)$. Concretely, we compute graph $G[C \setminus X]$ with its connected components: for each of them, we look at their neighborhood and verify whether $S$ is included into it. This takes $O(nm)$. In summary, checking the presence of a triplet $(a,b,S)$ witnessing the presence of PMC inside $X$ can be achieved by enumerating all pairs $(a,b)$ and then verifying whether the $X$-close $(a,b)$-separator exists and is $X$-shy. The running time is $O(n^3 m)$.

If no triplet satisfies the conditions listed in \autoref{co:pmc-carac}, then we conclude that $X$ is not a PMC superset. But if one works, then we can compute a PMC $\Omega \subseteq X$ in $O(n^3)$ (\autoref{co:pmc-carac}) and we answer that $X$ is indeed a PMC superset. The overall running time to decide \textsc{PMC superset} is thus $O(n^3 m)$.
\end{proof}

\section{Output-linear algorithm for the enumeration of PMCs} \label{sec:linear}

This section is dedicated to the description and the analysis of our algorithm \ALGO\ which enumerates all PMCs of a graph in output-linear time. The previous section will be used as a toolkit to attain this objective.

\subsection{Presentation of the algorithm: an extension of Bouchitté-Todinca method}

Our algorithm \ALGO\ can be seen as an extension of \textsc{BT-Enum}. We will keep its outline while improving the bottleneck step. Indeed, the fourth step of \textsc{BT-Enum} determines completely its worst-case time complexity. 
For the three first steps, we still do the same as \textsc{BT-Enum}, since these steps already cost output-linear time. Then, for the fourth step, we will search for the $a$-nice PMCs of $G$ in a BFS manner, starting from the PMCs enumerated in the three first steps. In fact, we will navigate in a ``supergraph of solutions'' which admits PMCs as vertices and is directed. Its arcs will be defined through a neighborhood function $\mathcal{N}$. We will then show that all $a$-nice PMCs of $G$ can be found by exploring this (partially unknown) supergraph, starting from PMCs which are not $a$-nice. 

We define the notion of \emph{neighborhood} between PMCs.

\begin{definition}[PMC neighborhood]
\label{def:pmc_neighbour}
We define the \emph{PMC neighborhood} $\mathcal{N}$ as a set of ordered pairs of PMCs. Given two PMCs $\Omega'$ and $\Omega$ of $G$, the ordered pair $(\Omega',\Omega) \in \mathcal{N}$ if $\Omega$ can be obtained from $\Omega'$ via the following process:
\begin{enumerate}
\item Guess some vertex $u^* \in \Omega'$ and fix $X^* = \Omega' \cup N_G(u^*)$.
\item Guess a second vertex $b^* \in \Omega'$ and compute the $X^*$-close $(b^*,u^*)$-separator $S^*$.
\item Keep only the vertices of $X^*$ which belong to $C_G[b^*,S^*] = S^* \cup C_{G}(b^*,S^*)$, then $\Omega = X^* \cap C_G[b^*,S^*]$. 
\end{enumerate}
We can also write $\Omega \in \mathcal{N}(\Omega')$ instead of $(\Omega',\Omega) \in \mathcal{N}$.
\end{definition}

\begin{algorithm}[h]
	\SetKwFor{For}{for}{do}{\nl endfor}
    \SetKwFor{Forall}{for all}{do}{\nl endfor}
    \SetKwIF{If}{ElseIf}{Else}{if}{then}{else if}{else}{}
    \DontPrintSemicolon
    \SetNlSty{}{}{:}
    \SetAlgoNlRelativeSize{0}
    \SetNlSkip{1em}
	
	\nl \KwIn{A graph $G$ and a set $\widehat{\Pi}$ containing some PMCs of $G$.}
 	\nl \KwOut{A set $\Pi$ containing PMCs of $G$, with $\widehat{\Pi} \cap \Pi = \emptyset$.}

	\nl $\Pi \leftarrow \emptyset$ \# Result of \textsc{PMC-BFS}, initially fixed as empty \# \;
	\nl $L \leftarrow \widehat{\Pi}$ \# Queue of the BFS \# \;

	\nl \While{$L$ is not empty}{
		\nl $\Omega' \leftarrow$ First PMC of $L$\;
		\nl $L \leftarrow L \setminus \Omega'$\;
		\nl \For{$\Omega \in \mathcal{N}(\Omega')$}{
			\nl \If{$\Omega \notin (\Pi \cup \widehat{\Pi})$ \label{line:test-present}}{
				\nl $\Pi \leftarrow \Pi \cup \{\Omega\}$\;
				\nl $L \leftarrow L \cup \{\Omega\}$\;
			}
		}
	}
	\nl \Return{} $\Pi$\;

	\caption{{\sc PMC-BFS}}
 	\label{alg:pmc_bfs}

\end{algorithm} 

This neighborhood function size is polynomial in the graph size: concretely, every PMC admits at most $n^2$ PMC out-neighbors. Computing the out-neighbors of some given PMC consists in guessing two vertices of the input graph and then applying several deterministic rules. Observe that given an arbitrary PMC $\Omega'$ and two of its vertices $u^*,b^*$, it is possible that the process of~\autoref{def:pmc_neighbour} does not produce a PMC (\autoref{fig:neighbor}). Therefore, we only keep in collection $\mathcal{N}$ the pairs $(\Omega',\Omega)$ for which the result returned by the process is indeed a PMC. This PMC neighborhood naturally defines a directed PMC-graph, where the vertices are PMCs of $G$ and arcs are given by the ordered pairs of $\mathcal{N}$.

\begin{figure}[h!]
	\centering
	\begin{subfigure}{0.49\textwidth}
	\scalebox{0.9}{\input{tikz/neighbor-yes}}
	\caption{Example of a PMC $\Omega'$ admitting a neighbor $\Omega$ for a certain choice of $u^*$ and $b^*$.}
	\label{subfig:neighbor-yes}
	\end{subfigure}~
	\begin{subfigure}{0.49\textwidth}
	\scalebox{0.9}{\input{tikz/neighbor-no}}
	\caption{Another choice of $u^*$ and $b^*$ leads to a set $\tilde{\Omega}$ which is not a PMC: $y \notin R_G(x,\tilde{\Omega})$.}
	\label{subfig:neighbor-no}
	\end{subfigure}
	\caption{Consider graph $G$ with the PMC $\Omega'$ in the light orange area. On the left, a first choice for $u^*$ and $b^*$ produces a potential neighbor of $\Omega'$: set $N[u^*]$ is drawn with bold orange lines around vertices, the $X^*$-close $(b^*,u^*)$-separator $S^*$ contains the vertices filled with green, and the final set $\Omega = X^* \cap C_G[b^*,S^*]$ is surrounded in blue. $\Omega$ is a PMC. On the right, given the same $G$ and $\Omega'$, a different choice for $u^*$ and $b^*$ produces a set $\tilde{\Omega}$ which is not a PMC. Indeed, the reachability between $x$ and $y$ is not satisfied.}
	\label{fig:neighbor}
\end{figure}

\textsc{PMC-BFS}, whose pseudocode is given in~\autoref{alg:pmc_bfs}, consists in a traversal of this PMC-graph. It takes as input a graph $G$ together with a starting set of PMCs $\widehat{\Pi}$ which is put into the initial queue of the BFS. It returns the set of PMCs visited during the traversal, excluding the initial ones in $\widehat{\Pi}$.

As a small - but very important - implementation remark, we specify briefly here how we can avoid duplication of PMCs, more concretely how we can check fast whether a given PMC was already visited/stored in the enumeration process. This is a test we execute at line~\autoref{line:test-present} of~\autoref{alg:pmc_bfs} for example. Indeed, computing such a test in exponential time would ruin our efforts to obtain overall output-linear time. We show, given a large set $\Pi$ of PMCs, how to check whether some PMC $\Omega$ is included into it or not, in time $n^{O(1)}$. We actually store our PMCs in a binary tree where edges at depth $i$ represent the addition or the exclusion of vertex $v_i$ in the PMC. In such data structure, any path from the root to a leaf represents a unique PMC. We can not only verify the presence of a PMC but also add one in a running time $O(n)$, depending only on the depth of the tree. Such a data structure is also used in \ALGO , see lines~\ref{line:case-1},~\ref{line:case-2}, \ref{line:case-3}, \ref{line:completion} of~\autoref{alg:outputlinear}. Such a data structure is very common in enumeration algorithms: as an example, it was used to store minimal separators in the polynomial-delay algorithm of Berry {\em et al.}~\cite{BeBoCo99} under the name of ``SEARCH-INSERT`` procedure.

\begin{algorithm}[h]
    \SetKwFor{For}{for}{do}{\nl endfor}
    \SetKwFor{Forall}{for all}{do}{\nl endfor}
    \SetKwIF{If}{ElseIf}{Else}{if}{then}{else if}{else}{}
    \DontPrintSemicolon
    \SetNlSty{}{}{:}
    \SetAlgoNlRelativeSize{0}
    \SetNlSkip{1em}

 	\nl \KwIn{A graph $G$.}
 	\nl \KwOut{All PMCs of $G$.}
 	
 	\# Base case \# \;
	
 	\nl \If{$G$ is a clique}{
 	    \nl \Return{} $\{V(G)\}$ \label{line:base}\;
 	}
 	\# Induction step \# \;
 	\nl $\mathcal{C} \leftarrow \ALGO (G \setminus \{a\})$\label{line:recursion} ~\#  $\mathcal{C}$ is a collection of vertex sets \# \;
 	\nl $\Pi \leftarrow$ empty collection of vertex sets\;
 	\nl \For{$\Omega' \in \mathcal{C}$}{
 	    \nl \lIf{$\Omega'$ PMC of $G$}{$\Pi \leftarrow \Pi \cup \{\Omega'\}$ \label{line:case-1}}
 	    \nl \lIf{$\Omega' \cup \{a\}$ PMC of $G$}{$\Pi \leftarrow \Pi \cup \{\Omega' \cup \{a\}\}$\label{line:case-2}}
 	    }
 	\nl \For{$S$ minimal separator of $G$}{
 		\nl \lIf{$S \cup \{a\}$ PMC of $G$}{$\Pi \leftarrow \Pi \cup \{S \cup \{a\}\}$ \label{line:case-3}}
 	}
 	\nl $\widehat{\Pi} \leftarrow$ copy of $\Pi$; \# $\Pi$ and $\widehat{\Pi}$ contain all PMCs which are not $a$-nice (\autoref{def:a-nice-pmc}). \label{line:fill-up} \# \;
 	\nl $\Pi_{a} \leftarrow$ PMC-BFS($G, \widehat{\Pi}$) \label{line:bfs}\;
 	\nl \Return{} $\Pi \cup \Pi_a$\label{line:completion}\;
 	\caption{{\sc \ALGO}}
 	\label{alg:outputlinear}
 \end{algorithm} 
 
The pseudocode of \ALGO\ is provided in~\autoref{alg:outputlinear}. The fourth step of \btenum\ is replaced by the launch of \textsc{PMC-BFS} (line~\ref{line:bfs}), where the starting set $\widehat{\Pi}$ is filled up (line~\ref{line:fill-up}) with the PMCs already listed in the first three steps (respectively lines~\ref{line:case-1},~\ref{line:case-2}, \ref{line:case-3}). Concretely, after line~\ref{line:fill-up}, set $\widehat{\Pi}$ is made up of the PMCs of $G$ which either are PMCs of $G \setminus \{a\}$ or contain $a$, where $a$ is the vertex withdrawn for the recursive call. 

We provide below a property of \textsc{PMC-BFS} which will be useful for the complexity analysis of \ALGO . The process used to compute the out-neighbors of a given PMC completely depends on the guess of two vertices: $u^*$ and $b^*$.
Hence, as announced, each PMC admits at most $n^2$ out-neighbors, considering a rough upper bound.

\begin{lemma}\label{lem:PMC_BFS_complexity}
	For any input $(G,\widehat{\Pi})$, \textsc{PMC-BFS} can be executed in time at most $O(n^3m\vert\Pi_G\vert)$.
\end{lemma}
\begin{proof}
	For a given guess $(u^*,b^*)$ in the neighboring process, a potential neighbor is built in $O(m)$, using graph searches to determine both the $X$-close $(b^*,u^*)$-separator and component $C_G[b^*,S^*]$. Then, we check whether the obtained set is a PMC in $O(nm)$. Finally, we obtain a time complexity of $O(n^3m)$ to compute all neighbors of a given PMC $\Omega'$. The BFS itself consumes $O(n^2\vert \Pi_G \vert)$, given the rough upper bound of $\vert \mathcal{N} \vert$. All in all, the most costly operation in \textsc{PMC-BFS} is the neighborhood computation, achieved overall in $O(n^3m\vert\Pi_G\vert)$.
\end{proof}

\subsection{Shutters of active PMCs}

From now on, we will give ourselves some PMC $\Omega$ which can be considered as a ``target PMC''. Our objective is to prove that, using the neighborhood function between PMCs (\autoref{def:pmc_neighbour}), there is a path in the PMC-graph $(\Pi_G,\mathcal{N})$ from some ``already listed'' PMC in $\widehat{\Pi}$ to our target $\Omega$ which is naturally supposed to be not already listed. As described in the previous subsection, our algorithm is recursive, and, we assume that PMCs of $G \setminus \{a\}$, for some vertex $a$, are already listed (line~\autoref{line:case-1} of \autoref{alg:outputlinear}). PMCs of $G$ containing $a$ are also listed easily in output-linear time (lines~\autoref{line:case-2} and~\autoref{line:case-3}). Consequently, we assume that $\Omega$ does not contain $a$ and is not a PMC of $G \setminus \{a\}$. Equivalently, we assume that $\Omega$ is an $a$-nice PMC (\autoref{def:a-nice-pmc}). As a consequence, $\Omega$ is active with regards to the minimal separator $S \subsetneq \Omega$, such that $S = N(C_G(a,\Omega))$. We fix $T = \Omega \setminus S$. Note that the singleton $\{a\}$ separates two vertices $x,y \in S$ in graph $G[\overline{\Omega} \cup \{x,y\}]$.

Our objective is thus to highlight that $\Omega$ can be attained from another PMC $\Omega'$ which either contains $a$ or is a PMC of $G \setminus \{a\}$, with a number of arcs, polynomial in $n$. This condition will allow us to state that any PMC of $G$ can be attained by \textsc{PMC-BFS}, where the starting set contains all PMCs of $G$ which are not $a$-nice (line~\autoref{line:bfs}).

From now on, we fix $\Omega$ as the notation for the ``target PMC'' described above, with its bipartition $S,T$ which is function of $a \in V(G)$: $S = N(C_G(a,\Omega))$ and $T = \Omega \setminus S$. Observe that $(S,T)$ is a shy decomposition of $\Omega$ (\autoref{def:shy-decomp}) as, otherwise, $\{a\}$ could not separate $x$ from $y$ (\autoref{le:pmc_without_a} and \autoref{le:shy-active}). Moreover, $T \cap C_G(a,\Omega) = \emptyset$.

Our approach consists in computing a \emph{preimage} of $\Omega$, {\em i.e.} a PMC $\Omega' \neq \Omega$ of $G$ such that the oriented arc $(\Omega',\Omega)$ belongs to the neighborhood function (\autoref{def:pmc_neighbour}). Then, we will establish key properties of $\Omega'$ showing that it is in some way going closer to a PMC which is not $a$-nice. To produce the preimage PMC $\Omega'$, we define the crucial concept of \emph{shutter}.

\begin{definition}[Shutter of a PMC]
Let $a \in V(G)$ and $\Omega$ be an $a$-nice PMC of $G$ (\autoref{def:a-nice-pmc}), let $S = N(C(a,\Omega)) \subsetneq \Omega$ and $T = \Omega \setminus S$. An $a$-\emph{shutter}\footnote{Let us say briefly a \emph{shutter} if there is no ambiguity on $a$.} $D \subseteq C(a,\Omega)$ is a set of vertices satisfying:
\begin{itemize}
\item $D$ is an inclusion-wise minimal separator of a pair $(u,v) \in S \times S$ in graph $G[\overline{\Omega} \cup \{u,v\}]$, meaning that it is not possible to have another pair $(u^*,v^*) \in S \times S$ such that some proper subset of $D$ separates $u^*$ from $v^*$ in $G[\overline{\Omega} \cup \{u^*,v^*\}]$.
\item For any pair $(u,v)$ separated by $D$ in $G[\overline{\Omega} \cup \{u,v\}]$, $D$ separates pair $(u,v)$ closely (\autoref{def:sep-close}) in this graph.
\item If $D$ does not contain vertex $a$, {\em i.e.} $a \notin D$, then for each pair $(u,v)$ separated by $D$, one among these two vertices $u,v$ is not only fully adjacent to $D$ (condition of the previous bullet) but is also not adjacent to $C_G(a,\Omega \cup D)$, {\em i.e.} the connected component of $G \setminus (\Omega \cup D)$ containing $a$.
\end{itemize}
\label{def:ashutter}
\end{definition}

Let us explain the role played by shutters and their meaning for the PMC $\Omega$. Looking at the shy decomposition $(S,T)$ of the PMC $\Omega$, our goal is to slightly modify the content of $\Omega$ in $S$ (for removal) and $C(a,\Omega)$ (for add-on) to produce another PMC $\Omega'$. Shutters stand thus as a rather restricted way to add vertices to the PMC $\Omega$ while removing some of $S$ to obtain a different PMC. Moreover, all vertices of the shutter are adjacent to at least one common vertex in $\Omega$. This key fact will help us in proving that $(\Omega',\Omega) \in \mathcal{N}$. Finally, we will show how the succession of such arcs in the PMC-graph gives us the whole connectivity.

The following theorem ensures us the existence of an $a$-shutter for $\Omega$. Moreover, it states that such $a$-shutter $D$ will not contain $a$, except if $D = \{a\}$.

\begin{theorem}[Existence of shutters]
Let $\Omega$ be an $a$-nice PMC. Then, $\Omega$ admits at least one $a$-shutter. 
Moreover, one of such $a$-shutter $D$ is either $D = \{a\}$ or does not contain $a$.
\label{th:ashutter}
\end{theorem}
\begin{proof}
We propose an algorithmic inductive procedure to build such $a$-shutter $D$. We will show that this algorithm finishes and indeed returns a shutter. More precisely, we produce sequences of vertex-sets $(D_i)_{i\ge 0}$, $(D_i')_{i\ge 0}$, and $(C_i)_{i\ge 1}$ so that at the last iteration $k\ge 0$, set $D_k'$ will be our shutter. The number of steps $k$ is linear, {\em i.e.} $k < n$. Set $D_i'$ is always a subset of $D_i$, {\em i.e.} $D_i' \subseteq D_i$.

\textit{Iteration $i=0$.} We fix $D_0 = \{a\}$. Note that $D_0$ is thus the separator of two vertices $x,y \in S$, as $\Omega$ is not a PMC of $G \setminus \{a\}$ (\autoref{le:pmc_without_a}). As it is a singleton, it is naturally inclusion-wise minimal, so we fix $D_0' = D_0$. If all pairs separated by $D_0$ are separated closely, then $D = D_0$ makes our theorem hold: $D_0$ is an $a$-shutter and $D_0 = \{a\}$ (see~\autoref{fig:ashutter} for an example). From now on, we assume that $D_0$ does not satisfy this property: there exists a pair $(u_1,v_1) \in S \times S$ (potentially $(x,y)$) separated by $D_0$ in graph $G[\overline{\Omega} \cup \{u_1,v_1\}]$ but not closely. More concretely, $a$ is an articulation point for the pair $(u_1,v_1)$ but neither $u_1$ nor $v_1$ are adjacent to $a$. In this case, we show via the induction process that $\Omega$ admits an $a$-shutter which does not contain $a$. We begin with base case $i = 1$.

\begin{figure}[h]
	\centering
	\scalebox{0.8}{\input{tikz/ashutter}}
	\caption{An example of graph $G$ (in black) given with an $a$-nice PMC $\Omega$ (in red) with $S = N(C_G(a,\Omega))$ and $T = \Omega \setminus S$. Indeed, $\Omega$ is not a PMC of $G\setminus \{a\}$ as $(x,z) \notin \torso(G\setminus \{a\},\Omega)$. The connected components of $G \setminus \Omega$ included into $C_G(T,S)$ are surrounded by blue dashed areas. The set $D_0 = \{a\}$ separates pair $(x,z)$ in graph $G[\overline{\Omega} \cup \{x,z\}]$, but also pairs $(y,z)$ and $(x,w)$ in their corresponding induced subgraph. Observe that, for instance, pair $(x,y)$ is not separated by $D_0$ as both vertices are adjacent to at least one common blue connected component. The same statement holds for pair $(y,w)$ which is directly adjacent.
	In this example, $D_0$ is an $a$-shutter: all pairs separated by $D_0$ - as mentioned above: $(x,z),(y,z),(x,w)$ - are separated closely, as $x$, $y$ (and even $z$) are adjacent to $a$.}
	\label{fig:ashutter}
\end{figure}

For any induction step $i\ge 1$ of our algorithm, we show the following statements:
\begin{enumerate}[label=(\alph*)]
\item $D_i$ is the $(u_i,v_i)$-close separator in $G[\overline{\Omega} \cup\{u_i,v_i\}]$, for some $(u_i,v_i) \in S \times S$; \label{it:di}
\item $D_i' \subseteq D_i$ is an inclusion-wise minimal separator between pairs in $S \times S$: formally, there is no pair $(u,v) \in S \times S$ admitting some minimal $(u,v)$-separator $D \subsetneq D_i'$ in graph $G[\overline{\Omega} \cup\{u,v\}]$; \label{it:dip}
\item set $D_i$ (and thus also $D_i'$) does not contain $a$; \label{it:contain}
\item $C_i$ is the connected component of $G \setminus (\Omega \cup D_i)$ which contains $a$, {\em i.e.} $C_i = C_G(a,\Omega \cup D_i)$, and such that $D_i$ is included into the neighborhood of $C_i$, {\em i.e.} $D_i \subseteq N_G(C_i)$; \label{it:c_i} 
\item for any $i \ge 2$, $C_{i-1} \subsetneq C_i$. \label{it:extend}
\end{enumerate}

At each step, we show that if $D_i'$ does not satisfy the theorem requirements, then we can define $D_{i+1}$, $D_{i+1}'$ and $C_{i+1}$ satisfying the induction properties. Since $C_i$ augments at each step (condition~\autoref{it:extend}), then the process necessarily stops at some moment and the theorem holds.

\textit{Base case: iteration $i=1$.} We fix $D_1$ as the close $(u_1,v_1)$-separator in graph $G[\overline{\Omega} \cup \{u_1,v_1\}]$. Set $D_1$ does not contain $a$ as $(a,u_1) \notin E(G)$. Hence, graph $G \setminus (\Omega \cup D_1)$ admits a connected component containing $a$: we fix it as $C_1$. If $D_1$ is not inclusion-wise minimal, we fix $D_1'$ any proper subset of $D_1$ which is inclusion-wise minimal and separate pairs of $S \times S$. By definition, bullets~\autoref{it:di},~\autoref{it:dip},~\autoref{it:contain} are satisfied and \autoref{it:extend} does not have to be checked for the case $i = 1$. For~\autoref{it:c_i}, by definition, $C_1$ contains $a$, but it remains to prove that $D_1 \subseteq N_G(C_1)$. Since $D_0 = \{a\}$ was an $(u_1,v_1)$-separator in $G[\overline{\Omega} \cup \{u_1,v_1\}]$, any path from $u_1$ to $v_1$ whose internal vertices belong to $\overline{\Omega}$ passes through vertex $a$. Furthermore, as $D_1$ is a minimal $(u_1,v_1)$-separator in the same graph $G[\overline{\Omega} \cup \{u_1,v_1\}]$, for every $d \in D_1$, there exists a path from $u_1$ to $v_1$ passing through $d$ and whose internal vertices belong to $\overline{\Omega \cup D_1}$. Consequently, this highlights the existence, for each $d \in D_1$, of a path from $a$ to $d$ whose internal vertices belong to $V(G) \setminus (\Omega \cup D_1)$. In brief, $C_1$ is adjacent to $d$ in $G$, for each $d \in D$, so $D_1 \subseteq N_G(C_1)$.

\textit{Induction step.} Assume by induction hypothesis that all properties \autoref{it:di}-\autoref{it:extend} hold for some $i \ge 1$. We also assume that $D_i'$ does not satisfy the theorem requirements, otherwise we are done. There can be two reasons for that: either $D_i'$ separates some pair ``unclosely'', or for some pair separated closely by $D_i'$, the ``close vertex'' of the pair is adjacent to the connected component of $a$. We distinguish these two cases.

\textbf{Case A:} there exists a pair $(u_{i+1},v_{i+1})$ which is separated by $D_i'$ in $G[\overline{\Omega} \cup \{u_{i+1},v_{i+1}\}]$ but not closely. 

We denote by $C_i'$ the connected component of $G \setminus (\Omega \cup D_i')$ which contains $a$. Naturally, $C_i \subseteq C_i'$ as $D_i' \subseteq D_i$: indeed, the removal of vertices from $D_i$ to obtain $D_i'$ could only make the component $C_i$ grow up. Moreover, as $D_i$ was fully adjacent to $C_i$, then its subset $D_i'$ will also be fully adjacent to $C_i' \supseteq C_i$, in brief $D_i' \subseteq N_G(C_i')$.

A crucial observation is here: $u_{i+1}$ and $v_{i+1}$ cannot be both adjacent to the connected component $C_i'$. Otherwise, $D_i$ would not be a $(u_{i+1},v_{i+1})$-separator. Assume w.l.o.g. that either $v_{i+1}$ is adjacent to $C_i'$ in $G$ or none of these two vertices $\{u_{i+1},v_{i+1}\}$ are adjacent to $C_i'$. We fix $D_{i+1}$ equal to the close $(u_{i+1},v_{i+1})$-separator in $G[\overline{\Omega} \cup \{u_{i+1},v_{i+1}\}]$\footnote{If $C_i'$ was adjacent to $u_{i+1}$, then $D_{i+1}$ would have been the close $(v_{i+1},u_{i+1})$-separator instead, and the same reasoning follows by simply inverting $u_{i+1}$ and $v_{i+1}$ in the remainder.}. As $D_i'$ does not separate $(u_{i+1},v_{i+1})$ closely, there exists a vertex $d_{i+1} \in D_i'$ such that $u_{i+1}$ is not adjacent to $d_{i+1}$. $D_{i+1}$ thus satisfies $D_{i+1} \neq D_i'$. Moreover, as they are minimal separators of the same pair $(u_{i+1},v_{i+1})$ in the same graph $G[\overline{\Omega} \cup\{u_{i+1},v_{i+1}\}]$, they cannot be included one into another. Hence, $D_i' \setminus D_{i+1}$ is nonempty and will contain at least $d_{i+1}$. As a consequence, regarding graph $G \setminus (\Omega \cup D_{i+1})$, component $C_i'$ will \emph{grow up} in a larger component containing $d_{i+1}$, as $D_{i+1}$ is closer to $u_{i+1}$ than $D_i'$, and $d_{i+1} \in N_G(C_i')$. Formally, set $C_i' \cup \{d_{i+1}\}$ is included into the connected component of $G \setminus (\Omega \cup D_{i+1})$ containing $a$.

We also fix $D_{i+1}'$ as any inclusion-wise minimal separator for pairs in $S \times S$ which is a subset of $D_{i+1}$. It exists by minimality principle: note that we could have $D_{i+1} = D_{i+1}'$ if $D_{i+1}$ was already inclusion-wise minimal. Finally, we fix $C_{i+1}$ as the connected component of $G \setminus (\Omega \cup D_{i+1})$ containing $a$. 
By definition of $D_{i+1}$ and $D_{i+1}'$, properties \autoref{it:di} and \autoref{it:dip} are satisfied. As $C_i' \cup \{d_{i+1}\} \subseteq C_{i+1}$ and previously $C_i \subseteq C_i'$, the inclusion property~\autoref{it:extend} is satisfied. Given that $C_{i+1}$ is a connected component of $G \setminus (\Omega \cup D_{i+1})$ which contains $a$, then obviously $a \notin D_{i+1}$ and thus $a \notin D_{i+1}'$, which gives~\autoref{it:contain}. It remains to prove that $D_{i+1} \subseteq N_G(C_{i+1})$. 

Observe again that $D_i'$ and $D_{i+1}$ are minimal separators of the same pair $(u_{i+1},v_{i+1})$ in the same graph $G[\overline{\Omega} \cup \{u_{i+1},v_{i+1}\}]$. Since $D_{i+1}$ is the close $(u_{i+1},v_{i+1})$-separator of this graph, then $D_i'$ and $D_{i+1}$ are parallel (\autoref{lem:parallel-close}). As a consequence, for any $x \in D_{i+1}$, there exists a path $P_x$ from $x$ to some $y \in D_i'$ whose internal vertices do not belong to $\Omega \cup D_i' \cup D_{i+1}$ (\autoref{lem:parallel-sep}). Therefore, any path from $a$ to $y\in D_i'$ whose internal vertices belong to $C_i'$ can be extended to $x \in D_{i+1}$: the obtained path has internal vertices belonging to $V \setminus (\Omega \cup D_{i+1})$. Said differently, the connected component of $G \setminus (\Omega \cup D_{i+1})$ containing $a$ admits the whole set $D_{i+1}$ in its neighborhood: $D_{i+1} \subseteq N_G(C_{i+1})$. Finally, \autoref{it:c_i} is also completed.  

\textbf{Case B}: all pairs separated by $D_i'$ are separated closely, but there exists a pair $(u_{i+1},v_{i+1})$ separated by $D_i'$ in $G[\overline{\Omega} \cup \{u_{i+1},v_{i+1}\}]$ such that, among $\{u_{i+1},v_{i+1}\}$, none of these two vertices is at the same time fully adjacent to $D_i'$ and disconnected from $C_i'$. 

Assume w.l.o.g. that $u_{i+1}$ is fully adjacent to $D_i'$: $D_i' \subseteq N_G(u_{i+1})$. In Case B, necessarily, $u_{i+1}$ is adjacent to $C_i'$. Therefore, $v_{i+1}$ is not adjacent to $C_i'$. Consequently, $v_{i+1}$ is not fully adjacent to $D_i'$, otherwise Case B would not be satisfied (because of $v_{i+1}$). We fix $D_{i+1}$ as the close $(v_{i+1},u_{i+1})$-separator in $G[\overline{\Omega} \cup \{u_{i+1},v_{i+1}\}]$. Separators $D_i'$ and $D_{i+1}$ are different as $v_{i+1}$ is not fully adjacent to $D_i'$. As they separate the same pair in the same graph, $D_i' \setminus D_{i+1} \neq \emptyset$. As $C_i'$ is not adjacent to $v_{i+1}$, then $D_{i+1} \cap C_i' = \emptyset$. Furthermore, component $C_i'$ is fully adjacent to $D_i'$ (consequence of~\autoref{it:c_i}), so graph $G \setminus (\Omega \cup D_{i+1})$ contains a connected component which includes $C_i'$ (as $D_{i+1} \cap C_i' = \emptyset$) and its adjacent vertices $D_i' \setminus D_{i+1}$. Let $C_{i+1}$ be this connected component which contains $a$. Finally, we fix $D_{i+1}' \subseteq D_{i+1}$ as an inclusion-wise minimal separator between pairs of $S \times S$. 

Again, by definition, items~\autoref{it:di} and~\autoref{it:dip} are satisfied. The observations above lead us to the same conclusion for item~\autoref{it:extend}: $C_i \subsetneq C_i' \cup (D_i' \setminus D_{i+1}) \subseteq C_{i+1}$. As for Case A, $C_{i+1}$ contains $a$, so $D_{i+1}$ and $D_{i+1}'$ do not, fulfilling~\autoref{it:contain}. To finish, we use the same observation as for Case A: $D_{i+1}$ and $D_i'$ are separators of the same pair in the same graph, and $D_{i+1}$ is closer to $v_{i+1}$. Therefore, for any $x \in D_{i+1}$, there exists a path from $a$ to $x$ passing through some vertex $y \in D_i'$. The internal vertices of this path belong to $V \setminus (\Omega \cup D_{i+1})$. In other words, any vertex of $D_{i+1}$ belongs to the neighborhood of $C_{i+1}$, which completes item~\autoref{it:c_i}.

\begin{figure}[h]
	\centering
	\scalebox{0.9}{\input{tikz/notashutter}}
	\caption{Another example of graph $G$ (in black) given with an $a$-nice PMC $\Omega$ (in red). The set $D_0 = \{a\}$ separates pairs $(x,w)$ in graph $G[\overline{\Omega} \cup \{x,w\}]$. Set $D_0 = D_0'$ is not an $a$-shutter as it separates $(x,w)$ but not closely (Case A). We thus fix $u_1 = x$ and $v_1 = w$ - note that we could have selected $u_1 = w$ and $v_1 = x$ but in the proof this choice is arbitrary. We compute $D_1$ as the close $(u_1,v_1)$-separator in $G[\overline{\Omega} \cup \{u_1,v_1\}]$: this set is made up of the vertices filled with orange. Unfortunately, $D_1$ is not inclusion-wise, as its subset $D_1'$ (surrounded by an orange area) separates pair $(y,z)$. However, $D_1'$ is an inclusion-wise minimal separator and we pursue with this set. Again, $D_1'$ is not an $a$-shutter: it separates $(y,z)$ closely as $y$ is fully adjacent to $D_1'$ but $y$ is also adjacent to $C(a,\Omega \cup D_1')$. Hence, we fall into Case B. Consequently, we compute the close $(z,y)$ separator $D_2$, made up of the vertices filled with yellow. This time, $D_2 = D_2'$ is an $a$-shutter: it separates only pair $(z,y)$, it is fully adjacent to $z$ in the pair, and $z$ is not adjacent to $C(a,\Omega \cup D_2')$.}
	\label{fig:notashutter}
\end{figure}

\textit{Conclusion.} We were able to introduce an inductive process which produces successive separators while the given $D_i'$ does not satisfy the theorem requirements. \autoref{fig:notashutter} provides an example of such shutter construction. Note that the sequence $(C_i)_{i\ge 1}$ is strictly increasing. Since it cannot grow infinitely ($\vert C_i \vert < n$), the process will necessarily stop at some moment $k \ge 0$. This will occur when neither Case A nor Case B will be satisfied. Thus, all pairs separated by $D_k'$ will be separated closely. Moreover, as Case B is not satisfied, for each pair separated by $D_k'$, at least one separated vertex will be not only adjacent to the whole $D_k'$ but also non adjacent to $C_k'$, the connected component of $G \setminus (\Omega \cup D_k')$ which contains $a$. Set $D_k'$ is thus an $a$-shutter (\autoref{def:ashutter}), this completes our proof of existence. Note that either $D_0' = \{a\}$ is returned, or it is $D_k'$ with $k > 0$ and then $a \notin D_k'$ by~\autoref{it:contain}. 
\end{proof}

\subsection{Preimage of an $a$-nice PMC} \label{subsec:preimage}

Every $a$-nice PMC $\Omega$ of $G$ admits an $a$-shutter $D$: either $D = \{a\}$ or $a \notin D$. We show that, in both cases, we can produce another PMC very close to $\Omega$, in terms of adjacency, and which contains $D$. Then, we will prove that such PMC is a preimage of $\Omega$ for the neighboring function $\mathcal{N}$.

We need a preliminary definition which will be useful for this purpose. We define a subset of $\Omega$ containing the \emph{shutter-close vertices}: they are vertices which not only undergo a separation due to $D$ but also are adjacent to it.

\begin{definition}[Shutter-close set]
Consider an $a$-nice PMC $\Omega$ of $G$, $S = N(C_G(a,\Omega))$ and an $a$-shutter $D$. A \emph{shutter-close set} of $\Omega$ relative to shutter $D$ is an inclusion-wise minimal set $S_D \subsetneq S$ satisfying the following properties:
\begin{itemize}
\item For any pair $(v,w)$ of vertices separated by $D$ in $G[\overline{\Omega} \cup\{v,w\}]$, at least one among these two vertices belongs to $S_D$.
\item If $v \in S_D$, then $v$ is fully adjacent to $D$, {\em i.e.} $D \subseteq N_G(v)$.
\end{itemize}
\label{def:shutter-close}
\end{definition}

As $S_D$ is defined as inclusion-wise minimal, we have $S_D \neq S$. Indeed, if it was the case, $D$ would be adjacent to any $s \in S$, by definition of $S_D$. For any fixed $x \in S$, for all separated pairs containing $x$, {\em i.e.} of the form $(x,y)$, we could put $y$ inside $S_D$ instead of $x$ and hence $x$ could be removed: \autoref{def:shutter-close} would still hold as all separated pairs are represented but $S_D \subseteq S \setminus \{x\}$. This fact is the consequence of the following lemma, which implies that $S \setminus S_D$ is nonempty.

\begin{lemma}[Property of shutter-close sets]
Consider an $a$-nice PMC $\Omega$ of $G$, $S = N(C_G(a,\Omega))$ and an $a$-shutter $D$ with shutter-close set $S_D$. There is at least one pair $(x,y) \in S \times S$ separated by $D$ in $G[\overline{\Omega} \cup\{x,y\}]$, such that $x \in S_D$ and $y \notin S_D$.
\label{lem:minimal-shutter-close}
\end{lemma}
\begin{proof}
Assume by way of contradiction that no such pair exists. According to \autoref{def:shutter-close}, every pair separated by $D$ admits at least one of its vertices in $S_D$. Consequently, all pairs separated by $D$ are included into $S_D \times S_D$. Consider any element $x \in S_D$ and let $\varphi_S(x)$ be the set of vertices which are separated from $x$ by $D$: due to our hypothesis, $\varphi_S(x) \subsetneq S_D$. Consider set $S_D' = S_D \setminus \{x\}$: it satisfies $\varphi_S(x) \subseteq S_D'$. As $S_D' \subsetneq S_D$, naturally all vertices of $S_D'$ are fully adjacent to $D$. Moreover, since separated pairs involving $x$ admit a second vertex inside $S_D$, then all pairs separated by $D$ are intersected in $S_D'$. Set $S_D'$ satisfies the two bullets of \autoref{def:shutter-close} and is a proper subset of $S_D$, contradicting the minimality of $S_D$.
\end{proof}

We can see also the shutter-close set as a \emph{minimal colored vertex cover} of some bicolored graph $H(\Omega,D)$ whose vertex set is made up of $S$ and two vertices $(v,w) \in S \times S$ are adjacent iff they are separated by $D$ in $G[\overline{\Omega} \cup\{v,w\}]$. We color in blue the vertices which are fully adjacent to $D$, in red the others. By \autoref{def:ashutter}, we know that the blue vertices form a vertex cover of $H(\Omega,D)$. Nevertheless, this blue vertex cover is not necessarily inclusion-wise minimal, which justifies the definition of shutter-close sets: see~\autoref{fig:preimage-a} for an illustration.

In the remainder, set $S_D$ represents a set of candidates to be withdrawn from $\Omega$ in order to obtain another PMC $\Omega'$ such that $(\Omega',\Omega) \in \mathcal{N}$. We will consider first the case where $D = \{a\}$. 

\subsubsection{When $a$ is a shutter}

Here, we assume that the shutter $D$ satisfies $D = \{a\}$. Considering the PMC superset $\Omega \cup \{a\}$ with the techniques presented in \autoref{sec:superset}, we can obtain a PMC $\Omega_a \neq \Omega$.

\begin{theorem}
Consider an $a$-nice PMC $\Omega$ with $S = N(C_G(a,\Omega))$ which admits $D = \{a\}$ as an $a$-shutter and let $S_D$ be a shutter-close set of $\Omega$ relative to $D$. There exists a PMC $\Omega_a = \{a\} \cup \widehat{S} \cup T$, where $\widehat{S} \subsetneq S$ and $S \setminus \widehat{S} \subseteq S_D$.
Furthermore, $(\Omega_a,\Omega) \in \mathcal{N}$.
\label{th:preimage-a}
\end{theorem}
\begin{proof}
Consider the PMC superset $X = \Omega \cup \{a\}$. Note that the \emph{native} connected components of $X$, {\em i.e.} the connected components of $G \setminus X$, cannot be adjacent to both $a$ and some vertex $t \in T$, because $S \subsetneq X$ separates $a$ from $T$.  We will remove vertices from $X$ with \removal\ using a certain selection of vertices (not necessarily arbitrary in line~\autoref{line:selection-removal} of \autoref{alg:removal}). We will show that we obtain the announced PMC $\Omega_a$. The priority rules on removals will ensure us that no vertex of $T \cup \{a\}$ is withdrawn.

\begin{figure}[h!]
	\centering
	\begin{subfigure}{0.63\textwidth}
	\centering
	\scalebox{0.7}{\input{tikz/ashutter2}}
	\caption{Repetition of \autoref{fig:ashutter} where $D = \{a\}$ is an $a$-shutter of an $a$-nice PMC $\Omega = S \cup T$. Pairs $(x,z),(x,w),(y,z)$ are separated.}
	\label{subfig:ashutter}
	\end{subfigure}~
	\hspace{0.02\textwidth}
	\begin{subfigure}{0.34\textwidth}
	\centering
	\scalebox{1}{\input{tikz/blue-cover}}
	\caption{Illustration of set $S_D$ via graph $H(\Omega,D)$: vertices $x$, $y$ and $z$ are in blue since they are fully adjacent to $D$. A minimal blue vertex cover of $H(\Omega,D)$ is $S_D = \{x,y\}$ for example.}
	\label{subfig:blue-cover}
	\end{subfigure}
	\begin{subfigure}{0.99\textwidth}
	\centering
	\scalebox{0.8}{\input{tikz/preimage-a}}
	\caption{The PMC $\Omega_a$ (surrounded by yellow) obtained, given the $a$-nice PMC $\Omega$ (surrounded by red), its $a$-shutter $D = \{a\}$ and the shutter-close set $S_D = \{x,y\}$. In this example, $S \setminus \widehat{S} = S_D$.}
	\label{subfig:preimage-a}
	\end{subfigure}~
	\caption{An example of shutter-close set $S_D = \{x,y\}$ for the $a$-shutter presented in \autoref{fig:ashutter}, together with the PMC $\Omega_a$ obtained according to \autoref{th:preimage-a}.}
	\label{fig:preimage-a}
\end{figure}

As $\{a\}$ is an $a$-shutter, then by \autoref{lem:minimal-shutter-close}, there exists at least one pair $(x,y) \in S_D \times (S \setminus S_D)$ such that $\{a\}$ is a minimal $(x,y)$-separator in $G[\overline{\Omega} \cup \{x,y\}]$. By definition of $S_D$, $x$ is adjacent to $a$, {\em i.e.} $(a,x) \in E(G)$. We select $y \in S \setminus S_D$ as 
the protected vertex involved in a conflict pair (line~\autoref{line:selection-removal} of \autoref{alg:removal}) and remove first its conflict $x$ from $X$: let us fix $X' = X \setminus \{x\}$. Set $X'$ contains $T \cup \{a\}$ so no native connected component from $G \setminus X$ is a full component of $X'$. Moreover, the connected component of $G \setminus X'$ which contains $x$ does not admit $y \in X'$ in its neighborhood. In brief, $X'$ does not have any full component: it is a shy set. If $X'$ is a PMC, then the theorem holds. Naturally, we assume that conflictual pairs remain in $X'$. We will then remove all other conflicts for $y$, and pursue the execution of \removal\ with non-arbitrary selections, detailed below.

First, let us state some nice properties satisfied by $X'$ in terms of reachability. Concretely, we show that pairs in $\{a\} \times T$, $T \times T$, $S \times T$ and $\{a\} \times S$ are not conflictual for $X'$.
\begin{itemize}
\item $T \subseteq R_G(a,X')$: $a$ is indeed adjacent to $x \in \overline{X'}$ and, as $\Omega$ is a PMC of $G$, there is a path from $x$ to any $t \in T$ whose internal vertices belong to $\overline{\Omega} \cap C_G(T,S) \subseteq \overline{X'}$.
\item For any $t \in T$, $T \subseteq R_G(t,X')$: as $\Omega$ is a PMC, there exists a path from $t$ to any $t' \in T$ whose internal vertices belong to $\overline{\Omega} \cap C_G(T,S) \subseteq \overline{X'}$.
\item For any $s \in S \setminus \{x\}$, $T \subseteq  R_G(s,X')$: again, as $\Omega$ is a PMC of $G$, there is path from $s$ to any $t \in T$ whose internal vertices belong to $\overline{\Omega} \cap C_G(T,S) \subseteq \overline{X'}$.
\item $S \setminus \{x\} \subseteq R_G(a,X')$: the reason is quite different here, $S$ is a minimal separator of $G$ and $C_G(a,S)$ is a full component for $S$ in $G$. Therefore, as $a \in C_G(a,S)$ by definition, for each $s \in S \setminus \{x\}$, there exists a path from $a$ to $s$ in $G$ whose internal vertices only belong to $C_G(a,S) \setminus \{a\}$. This gives also a path from $a$ to $s$ whose internal vertices belong to $\overline{X'}$, as $(C_G(a,S) \setminus \{a\}) \cap X' = \emptyset$.
\end{itemize}
Consequently, the only conflictual pairs of $X'$ belong to $S \times S$, or more precisely to $(S \setminus \{x\}) \times (S \setminus \{x\})$. Note also that pairs in $(S \setminus S_D) \times (S \setminus S_D)$ are not conflictual: if one such pair was, then it would contradict \autoref{def:shutter-close} as for some pair separated by $D$, none of its vertices would belong to $S_D$. Finally, we conclude that conflictual pairs only involve pairs in $S_D \times S_D$ and $S_D \times (S \setminus S_D)$. 

We pursue the removal process, by first removing all conflicts for $y \in S \setminus S_D$ and, second selecting new protected vertices in $S \setminus S_D$ if it is possible, otherwise in $S_D$ if all vertices of $S \setminus S_D$ do not have any conflict. As vertices of $T \cup \{a\}$ are not involved in any conflictual pair, the obtained set $\Omega_a$ contains $T \cup \{a\}$. Remember that the native connected components of $X$ could not be adjacent to the whole set $T \cup \{a\}$: since only native connected components can avoid our resulting set to be a PMC (\autoref{th:brightening}), there is no chance it can occur as $T \cup \{a\} \subsetneq \Omega_a$. As a conclusion, $\Omega_a$ is a PMC. Moreover, as we avoid the removal of vertices in $S \setminus S_D$ by selecting them as protected vertices, we have $X' \setminus \Omega_a \subsetneq S_D$. More generally, as $x \in S_D$, we have $X \setminus \Omega_a = S \setminus \Omega_a \subseteq S_D$. In summary, $\Omega_a$ contains $T \cup \{a\}$ together with some vertices of $S$, denoted by $\widehat{S}$, such that $S \setminus \widehat{S} = X \setminus \Omega_a \subseteq S_D$, proving the first part of the theorem. 

It remains to prove that $(\Omega_a,\Omega) \in \mathcal{N}$. Consider thus the set $\Omega_a$ and fix $u^* = a$ (following~\autoref{def:pmc_neighbour}). We obtain the PMC superset $X^* = \Omega_a \cup N_G(a)$. As all vertices removed from $X$ in the process, {\em i.e.} in $X \setminus \Omega_a$, belong to $S_D$ and are adjacent to $a$, we have $\Omega \subsetneq X^*$ and, even, $\Omega = X^* \cap (S \cup C_G(T,S))$. Fix $b^* \in T$. As $\Omega$ is a PMC of $G$, vertex $b^*$ reaches all vertices of $S$ via internal vertices in $C_G(T,S) \setminus T \subseteq \overline{X^*}$. Hence, $S$ is the $X^*$-close $(b^*,u^*)$-separator of $G$: $S^* = S$. As all vertices of $\Omega$ belong to $S \cup C_G(b^*,S^*)$, the removal of $X^* \setminus (S \cup C_G(b^*,S^*))$ gives us the PMC $\Omega$.
\end{proof}

\subsubsection{When $a$ does not belong to the $a$-shutter}

As a reminder, in \autoref{th:ashutter}, we established that either an $a$-shutter $D$ is equal to $\{a\}$ or does not contain $a$. We focus on this second case: $a \notin D$.

The notion of shutter-close set needs to be slightly adapted to this different scenario. By \autoref{def:ashutter}, if an $a$-shutter does not contain $a$, for any pair $(u,v)$ separated by $D$ in $G[\overline{\Omega} \cup \{u,v\}]$, one vertex of the pair $(u,v)$ is both fully adjacent to $D$ and not adjacent to $C_G(a,\Omega \cup D)$. Thus, we provide an extension of \autoref{def:shutter-close} where the shutter-close vertices must also be nonadjacent to $C_G(a,\Omega \cup D)$.

\begin{definition}[Extended shutter-close set]
Consider an $a$-nice PMC $\Omega$ of $G$, $S = N(C_G(a,\Omega))$ and an $a$-shutter $D$ such that $a\notin D$. An \emph{extended shutter-close set} of $\Omega$ relative to $D$ is an inclusion-wise minimal set $S_D \subsetneq S$ satisfying the following properties:
\begin{itemize}
\item for any pair $(v,w)$ of vertices separated by $D$ in $G[\overline{\Omega} \cup\{v,w\}]$, at least one among these two vertices belongs to $S_D$.
\item If $v \in S_D$, then $v$ is fully adjacent to $D$, {\em i.e.} $D \subseteq N_G(v)$.
\item If $v \in S_D$, then $v$ is not adjacent to $C_G(a,\Omega \cup D)$. 
\end{itemize}
\label{def:ext-shutter-close}
\end{definition}

As for classic shutter-close sets, it can be shown with the same arguments that there is a pair in $S_D \times (S \times S_D)$ which is separated by $D$, where $S_D$ is an extended shutter-close set.

\begin{lemma}[Property of extended shutter-close sets]
Consider an $a$-nice PMC $\Omega$ of $G$, $S = N(C_G(a,\Omega))$ and an $a$-shutter $D$, $a \notin D$, with extended shutter-close set $S_D$. There is at least one pair $(x,y) \in S \times S$ separated by $D$ in $G[\overline{\Omega} \cup\{x,y\}]$, such that $x \in S_D$ and $y \notin S_D$.
\label{lem:ext-minimal-close}
\end{lemma}
\begin{proof}
The proof comes from the fact that $S_D$ is inclusion-wise minimal: if no such pair exists, set $S_D$ could be reduced. The reasoning is identical to the one used for \autoref{lem:minimal-shutter-close}.
\end{proof}

The graph analogy we established for shutter-close sets also holds for the extended version. Consider the same graph $H(\Omega,D)$ but with different colors: now the vertices in blue are the vertices not only fully adjacent to $D$ but also not adjacent to $C_G(a,\Omega \cup D)$. Other vertices are in red. From \autoref{def:ashutter}, we know that the blue vertices form a vertex cover of $H(\Omega,D)$. \autoref{def:ext-shutter-close} is thus quite natural: the extended shutter-close set is an inclusion-wise minimal blue vertex cover of $H(\Omega,D)$.

\begin{theorem}
Consider an $a$-nice PMC $\Omega$, $S = N(C_G(a,\Omega))$, $T = \Omega \setminus S$, and an $a$-shutter $D$ satisfying $a \notin D$. Let $S_D$ be an extended shutter-close set of $\Omega$ relative to $D$. There exists a PMC: 
$$\Omega' = D \cup (\Omega \setminus S_D) \cup S_D',$$
where $S_D'$ is a subset of $S_D$, {\em i.e.} $S_D' \subseteq S_D$.

Therefore, $a \notin \Omega'$. We also have that the separator $S' = N(C_G(a,\Omega'))$ satisfies $C_G(a,S') \subsetneq C_G(a,S)$. Furthermore, $(\Omega',\Omega) \in \mathcal{N}$.
\label{th:preimage-not-a}
\end{theorem}
\begin{proof}
As for \autoref{th:preimage-a}, we will terminate the proof by showing that $(\Omega',\Omega) \in \mathcal{N}$. The first part of the proof consists in showing the existence of $\Omega'$ and the property over $S'$.

We consider the PMC superset $X = \Omega \cup D$. We show how vertex removals in $S_D$ allows us to obtain the announced PMC $\Omega'$, distinguishing several cases. We define $S'$ as the $X$-close $(a,t)$-separator, where $t \in T$. Naturally, $S'$ exists as $S \subsetneq X$ separates $a$ from $t$. However, $S'$ is different from $S$ as it does not contain $S_D \neq \emptyset$: from \autoref{def:ext-shutter-close}, vertices of $S_D$ are not adjacent to $C_G(a,X)$. Furthermore, $S'$ is $X$-shy with regards to $C_G(t,S')$. Indeed, assume by way of contradiction that it is not. There would be a connected component $C \subsetneq C_G(t,S')$ of $G \setminus X$ such that $S' \subseteq N(C)$. This is impossible as $D \subsetneq S'$ and $t$ are separated by $S \subsetneq X$, so $C$ cannot be adjacent to all vertices of $S'$. By \autoref{thm:indirect_PMC_superset_characterisation}, we can identify a PMC which is a subset of $X$ by following conflict removals.

As a last property, we have $C_G(a,S') \subsetneq C_G(a,S)$ because $S' \neq S$ both separate pair $(a,t)$ in $G$, $S \subsetneq X$ and $S'$ is the $X$-close $(a,t)$-separator. Less formally, both separators are included into $X$ but $S'$ is the closest to $a$.

We start the removal process of $X$ by forcing, as it was achieved in the proof of~\autoref{thm:indirect_PMC_superset_characterisation}: the shy separator $S' = N(C_G(t,S'))$ to be maintained, and removing vertices from $X \setminus S' \subseteq C_G(t,S')$. All pairs of vertices in $S' \times S'$ see each other despite $X$ thanks to $C_G(a,X)$. We will fix some priority rules for selecting the protected vertex, the main one being that it belongs in priority to $S'$, instead of $X \setminus S'$.

Let $(s,s')$ be a pair in $S_D \times (S \setminus S_D)$ separated by $D$ in $G[\overline{\Omega} \cup \{s,s'\}]$: from \autoref{lem:ext-minimal-close}, such pair exists. We fix $s' \in S \setminus S_D$ as a protected vertex, and we remove first the vertex $s$: it can be done since $s \notin S'$ and $s$ has a conflict with $s'$. The obtained set is $X_s = X \setminus \{s\}$. Let us pursue with a few connectivity remarks concerning $X_s$.

Each pair $(u,v) \in (\Omega \setminus S_D) \times (\Omega \setminus S_D)$ is a reachable pair despite $X$: these vertices are connected in graph $G[\overline{X} \cup \{u,v\}]$. Indeed, by \autoref{def:ext-shutter-close}, pairs separated by $D$ contain at least one vertex from $S_D$. As a consequence, pairs in $(\Omega \setminus S_D) \times (\Omega \setminus S_D)$ are also reachable despite $X_s \subsetneq X$: they are not conflictual.
Similarly, pairs in $D \times D$ are not conflictual despite $X_s$, since all of them are adjacent to $s$, and $s \notin X_s$. Moreover, pairs in $D \times (\Omega \setminus S_D)$ are not conflictual: let $x \in (\Omega \setminus S_D)$, either $x$ saw $s$ despite $X$, and now it sees entirely $D$ despite $X_s$, or $x$ was separated from $s$ by $D$. In this latter case, as $D$ was a minimal $(s,x)$-separator in $G[\overline{\Omega}\cup \{s,x\}]$, by definition of minimal separators, vertex $x$ reaches every vertex of $D$ despite $X \supsetneq X_s$.



\begin{figure}[t]
	\centering
	\scalebox{0.8}{\input{tikz/preimage-nota}}
	\caption{Example of PMC $\Omega'$ obtained, considering the PMC $\Omega$ already presented in \autoref{fig:notashutter}. Here, $S_D = \{z\}$. }
	\label{fig:preimage-nota}
\end{figure}

In summary, the pairs $D \times D$, $D \times (\Omega \setminus S_D)$ and $(\Omega \setminus S_D) \times (\Omega \setminus S_D)$ are not conflictual for $X_s$. Moreover, set $X_s$ can be partitioned into three sets: $D$, $\Omega \setminus S_D$ and $S_D \setminus \{s\}$. Hence, for all remaining reachability conflicts in the process, at least one vertex of the conflict pair belongs to $S_D$ (and thus not to $S'$). Remember that $S' \cap S_D = \emptyset$. We can thus force that the removed vertex at each step belongs to $S_D$, by taking in priority protected vertices first inside $S'$ (to handle all $S' \times S_D $ conflict pairs), second inside $(\Omega \setminus S_D)$ (to handle conflict pairs $(\Omega \setminus S_D) \times S_D$) and eventually inside $S_D$ (to handle the only remaining conflicts in $S_D \times S_D$). Since by \autoref{thm:indirect_PMC_superset_characterisation} we are sure to obtain a PMC $\Omega'$ via this process (native components of $X$ cannot be adjacent to both $D$ and $T$, while $D \cup T \subseteq \Omega'$), we have $\Omega' \subsetneq X$ satisfying $X \setminus S_D \subseteq \Omega'$. Consequently, $\Omega'$ contains sets $D$ and $\Omega \setminus S_D$. As announced, we have $$D \cup (\Omega \setminus S_D) \subseteq ~\Omega' ~\subsetneq D \cup (\Omega \setminus S_D) \cup S_D.$$ 
Moreover, as $S'$ was taken as the $X$-shy separator for the removal process, $S' \subsetneq \Omega'$. Finally, fixing $S'_D = \Omega' \cap S_D$, we have $\Omega' = D \cup (\Omega \setminus S_D) \cup S'_D$.

Let us show now that $(\Omega',\Omega) \in \mathcal{N}$. The arguments are very similar to the ones used in \autoref{th:preimage-a}. We guess $u^* \in D$, and we add $N_G(u^*)$ to $\Omega'$: $X^* = \Omega' \cup N_G(u^*)$. Note the following key fact: $\Omega \setminus \Omega'$ is included into $N_G(u^*)$ as $\Omega \setminus \Omega' \subseteq S_D$ and by definition, vertices of $S_D$ are adjacent to every vertex of $D$. In brief, $\Omega \subsetneq X^*$. We guess some $b^*$ which belongs to $T \subseteq \Omega \cap \Omega'$. As $u^* \in C_G(a,S)$ and $b^* \in C_G(T,S)$, $S$ is a minimal $(b^*,u^*)$-separator. More precisely, set $S$ is the $X^*$-close $(b^*,u^*)$-separator in $G$: as $\Omega$ is a PMC of $G$, there is a path from $b^*$ to each $s \in S$ whose internal vertices belong to $\overline{X^*}$. So $S^* = S$. Finally, we only keep vertices of $X^* \cap (S^* \cup C_G(b^*,S^*))= \Omega$.
\end{proof}

An example of preimage $\Omega'$, highlighted by~\autoref{th:preimage-not-a}, is shown in~\autoref{fig:preimage-nota} on the same setting than in~\autoref{fig:notashutter}.

\subsubsection{Conclusions on PMC neighborhood}

Finally, we present a statement summarizing our contributions in \autoref{subsec:preimage}: any $a$-nice PMC admits a preimage which is ``closer'' to some PMC which is not $a$-nice.

\begin{theorem}[Preimage of an $a$-nice PMC]
Let $\Omega$ be an $a$-nice PMC of $G$. There exists a PMC $\Omega' \neq \Omega$, satisfying $(\Omega',\Omega) \in \mathcal{N}$, such that $\Omega'$ either contains $a$ or otherwise $C_G(a,\Omega') \subsetneq C_G(a,\Omega)$.
\label{th:preimage}
\end{theorem}
\begin{proof}
From \autoref{th:ashutter}, $\Omega$ admits an $a$-shutter $D$. Either $D = \{a\}$: in this case, from \autoref{th:preimage-a}, there is a PMC $\Omega' = \Omega_a$ which contains $a$ and such that $(\Omega',\Omega) \in \mathcal{N}$. Otherwise, $a \notin D$: from \autoref{th:preimage-not-a}, there exists a PMC $\Omega'$ such that $(\Omega',\Omega) \in \mathcal{N}$ and moreover, $C_G(a,S') \subsetneq C_G(a,S)$, reusing the notation from \autoref{th:preimage-not-a}. In both PMCs, $S$  (resp. $S'$) is the neighborhood of the connected component of $G \setminus \Omega$ (resp. $G \setminus \Omega'$) containing $a$. Hence, $C_G(a,S) = C_G(a,\Omega)$ and $C_G(a,S') = C_G(a,\Omega')$. This completes the proof.
\end{proof}

The direct consequence of this result is a connectivity result on PMC neighborhood $\mathcal{N}$ which allows us to conclude that \textsc{PMC-BFS} will find all $a$-nice PMCs.

\begin{corollary}
Let $\Omega$ be an $a$-nice PMC of $G$. There exists a sequence of PMCs $\Omega^{(r)},\ldots,\Omega^{(2)},\Omega^{(1)}$ with $r \ge 2$, such that:
\begin{itemize}
\item $\Omega^{(1)} = \Omega$,
\item $\Omega^{(r)}$ is not $a$-nice,
\item for each $2\le i\le r$, $(\Omega^{(i)},\Omega^{(i-1)}) \in \mathcal{N}$.
\end{itemize}
\label{co:pmc-connect}
\end{corollary}
\begin{proof}
According to \autoref{th:preimage}, the PMC $\Omega$ admits a preimage $\Omega'$ which either contains $a$ or satisfies $C_G(a,\Omega') \subsetneq C_G(a,\Omega)$. If $\Omega'$ is not $a$-nice (in particular if it contains $a$), then we highlight the (short) sequence $\Omega',\Omega$. Otherwise, if $\Omega'$ is still $a$-nice, then its set $C_G(a,\Omega')$ is reduced compared to $C_G(a,\Omega)$. In this case, we apply again \autoref{th:preimage} on PMC $\Omega'$ and iteratively on all PMCs produced which are $a$-nice. The fact is that set $C_G(a,\Omega)$ decreases by at least one vertex at each step. Hence, this process is finite and since $a$ always belongs to $C_G(a,\Omega)$, we will for sure obtain at some moment $r \ge 2$, a PMC $\Omega^{(r)}$ which is not $a$-nice.
\end{proof}

Finally, let us introduce you to \autoref{fig:solution-path} and \autoref{fig:solution-path2} as two complete examples of how a target $a$-nice PMC $\Omega$ can be reached starting from a non $a$-nice one. On successive PMCs, the connected component containing $a$ ({\em i.e.} set $C_G(a,\Omega)$) is highlighted in green: we can see it melting step by step. In \autoref{fig:solution-path}, the reader can notice the reasoning-based reading path from top left (target PMC $\Omega$) to bottom right (some not $a$-nice PMC) and the algorithm-based reading path - representing a possible traversal for \textsc{PMC-BFS} - from the bottom right to the top left. The reasoning-based reading path follows the preimage search until a non $a$-nice PMC is found. Conversely, the algorithm-based reading path provide an overview of the functioning of \textsc{PMC-BFS} which leads to discover the target PMC $\Omega$. \autoref{fig:solution-path2} provides an alternative smaller path to reach the target PMC $\Omega$.

\begin{figure}
	\centering
	\scalebox{0.98}{\input{tikz/solutionpath}}
	\caption{Successive preimages of $\Omega$ until reaching a PMC $\Omega^{(5)}$ which is not $a$-nice. Observe that all intermediary PMCs are also $a$-nice, since $\{a\}$ always separates $w$ from another vertex.}
	\label{fig:solution-path}
\end{figure}

\begin{figure}
	\centering
	\scalebox{1}{\input{tikz/solutionpath2}}
	\caption{An alternative PMC-path compared to \autoref{fig:solution-path}: how the choice of the $a$-shutter can influence the path from an $a$-nice PMC to our target $\Omega$. Indeed, the preimage of $\Omega$ is $\Omega^*$ here (different from $\Omega^{(1)}$), it contains the $a$-shutter $\{b\}$ (instead of $\{z_1,z_2\}$) of $\Omega$ in $G$.}
	\label{fig:solution-path2}
\end{figure}


\subsection{Analysis of the algorithm} \label{subsec:analysis}

We are now ready to provide the analysis of our enumeration algorithm \ALGO\ (\autoref{alg:outputlinear}). We show that no PMC of $G$ is forgotten in the enumeration process. Also, we prove that it is executed in overall output-linear time.

\mainResult*
\begin{proof}
	Remember that \ALGO\ follows the principle of \textsc{BT-Enum} with a modification of the fourth step. The goal of this fourth step is to list all the $a$-nice PMCs. By \autoref{co:pmc-connect}, \textsc{PMC-BFS} (\autoref{alg:pmc_bfs}) finds all the $a$-nice PMCs if the starting queue of the BFS contains all PMCs which are not $a$-nice, already computed in the three first steps. According to~\autoref{lem:PMC_BFS_complexity}, its time complexity is $O(n^3m\vert \Pi_G \vert)$.
	
Since we use the same three first steps from \textsc{BT-Enum} (lines~\ref{line:case-1},~\ref{line:case-2} and~\ref{line:case-3} of~\autoref{alg:outputlinear}), all PMCs which are not $a$-nice (which contain $a$ or are PMCs of $G \setminus \{a\}$) are indeed found before executing \textsc{PMC-BFS}. We denote by $T(G)$ the time complexity of \ALGO\ for $G$. Combining the running time needed for each step, we obtain:	
	\begin{align*}
T(G) &=  T(G \setminus \{a\}) + O(nm\vert \Pi_{G \setminus \{a\}} \vert) + O(nm\vert \Delta_G \vert) + O(n^3m \vert \Pi_G \vert))\\
& =   T(G\setminus \{a\}) + O(n^3m\vert \Pi_G \vert).
\end{align*}
	Indeed, there are less PMCs in $G \setminus \{a\}$ than in $G$ (\autoref{le:more-pmcs}), and $\Delta_G \le n\vert \Pi_G \vert$ from~\autoref{sec:intro}. Hence, there exists some positive constant $K \ge 1$ such that $T(G) \le  T(G\setminus \{a\}) + Kn^3m\vert \Pi_G \vert$.

We show by induction on $n$ that $T(G) \le Kn^4m\vert \Pi_G \vert$. The base case ($G$ is a clique, line~\ref{line:base} of~\autoref{alg:outputlinear}) takes $O(1)$ time and satisfies the previous inequality (if no, then $K$ could adjusted to a larger value but still constant). Assume, by induction hypothesis, that any subgraph $H$ of $G$ satisfies $T(H) \leq K\cdot \vert V(H) \vert^{4} \cdot \vert E(H) \vert \cdot \vert \Pi_H \vert$. As $G \setminus \{a\}$ is a subgraph of $G$ and admits less PMCs, we get $T(G \setminus \{a\}) \leq K(n - 1)^{4}m \vert \Pi_G \vert$. Consequently, \[T(G) \le K(n-1)^{4}m \vert \Pi_G \vert + Kn^3m \vert \Pi_G \vert = Km\left((n-1)^{4} + n^3\right) \vert \Pi_G \vert \le Kn^{4}m\vert \Pi_G \vert.\]

	This concludes the proof since all the PMCs of the input graph are found by \autoref{alg:outputlinear} and its execution time is $O(n^4m \vert \Pi_G \vert)$.
\end{proof}

\newpage

\section{Hardness results}

We provide in this section polynomial reductions from $3$-SAT to two extension problems: \textsc{Minimal $(a,b)$-separator Extension} and \textsc{PMC Extension}. Finally, we show that in, our reduced instance for \textsc{PMC Extension},the constraint on forbidden vertices can be withdrawn. This allows us to conclude the NP-hardness of \textsc{PMC Subset}, which is \textsc{PMC Extension} with the restriction $F = \emptyset$. 

\subsection{NP-hardness of minimal separator extension} \label{subsec:hardness-sep}

We focus on the problem consisting in extending a minimal $(a,b)$-separator. To the best of our knowledge, the computational complexity of this problem is unknown. An example of yes-instance is given in~\autoref{fig:instance_absepProblem}.

\defdecproblem{\textsc{minimal $(a,b)$-separator extension}}{A graph $G = (V,E)$, $a,b \in V$ with $(a,b) \notin E$, a set $M \subseteq V$ of \textit{mandatory} vertices, a set $F \subseteq V$ of \textit{forbidden} ones such that $M \cap F = \emptyset$.}{Is there a minimal $(a,b)$-separator $S \subseteq V \setminus \{a,b\}$ such that $M \subseteq S$ and $S \cap F = \emptyset$ ?}

\begin{figure}[ht]
\centering
\scalebox{0.7}{\input{tikz/minabSepExt_instance}}
\caption{An instance of \textsc{Minimal $(a,b)$-separator Extension} with a graph $G$, vertices $a,b$, sets $M = \{m_1,m_2\}$ and $F = \{f_1,f_2\}$. A solution is given in blue.}
\label{fig:instance_absepProblem}
\end{figure}

We propose a polynomial reduction from 3-SAT to \textsc{Minimal \absep\ extension}. This hardness result will be used as an intermediary step to handle the computational complexity of \textsc{PMC Extension} in~\autoref{subsec:hard-pmc-ext}. Given a CNF 3-SAT formula $\Phi$, we build a corresponding graph $G(\Phi)$ which is an instance of \textsc{Minimal \absep\ extension}.

\begin{theorem}\label{the:min_absep_ext_npc}
\textsc{Minimal \absep\ extension} is NP-Complete.
\end{theorem}
\begin{proof}
Obviously, the problem belongs to NP since checking whether a set $S$ is a minimal $(a,b)$-separator, contains $M$ and does not intersect $F$, can be done in polynomial time. We will now reduce from 3-SAT. We work on the specific case $F = \emptyset$. Let $x_1,\ldots,x_n$ be the Boolean variables, $x_i$ and $\notxi{x}{i}$ the possible expressions for a literal involving variable $x_i$. Let $\Phi = \bigwedge\limits_{j = 1}^m C_j$ be a CNF 3-SAT formula with the clauses $C_j$, made up of at most three literals. We denote by $\overline{C_j}$ the negation of $C_j$.

\textbf{Construction.} We transform the formula $\Phi$ into a graph $G(\Phi)$. We propose a $6$-vertex gadget to represent variables. It models the selection of a Boolean value for variable $x_i$. Each gadget is connected to $a$ and $b$.
\begin{figure}[ht]
\centering
\begin{subfigure}{0.4\textwidth}
\centering
\scalebox{0.7}{\input{tikz/gadget_xi}}
\caption{Variable gadget for $x_i$ with its connection to $a$ and $b$.
\label{subfig:gadget_xi}}
\end{subfigure}
\begin{subfigure}{0.58\textwidth}
\scalebox{0.7}{\input{tikz/clause_ci}}
\caption{How clause vertices are connected to variable gadgets: example with $C_j = x_i \lor \notxi{x}{k}$.}
\label{subfig:clause_ci}
\end{subfigure}
\caption{The variable gadgets and their links with clause vertices.}
\label{fig:gadgetsep}
\end{figure}
In \autoref{fig:gadgetsep}, the diamond green vertices represent the mandatory set $M$. Variable gadgets are thus made up of two mandatory vertices $a_i$ and $b_i$, and four other which can be used for extension: $\alpha_i,\notxi{\alpha}{i},\beta_i,\notxi{\beta}{i}$. Observe that each variable gadget, attached with vertices $a$ and $b$, form a hypercube of dimension $3$.

We complete variable gadgets by adding clause vertices. Each clause $C_j$ is represented by a vertex in $M$ adjacent to $a$, and for every literals $\ell_i \in \overline{C_j}$, we add an edge from $C_j$ to $\beta_i$ if $\ell_i = x_i$, otherwise to $\notxi{\beta}{i}$ if $\ell_i = \notxi{x}{i}$. This graph can naturally be built in polynomial time since for a formula with $n$ variables and $m$ clauses, we need $6n + m +2$ vertices and at most $12n + 4m$ edges.

To achieve the reduction, we show the following claim, which highlights a bijection between the truth assignments $\sigma$ of $\Phi$ and the minimal $(a,b)$-separators of $G(\Phi)$ extending $M$.

\textit{Claim}: The minimal $(a,b)$-separators of $G(\Phi)$ containing $M$ are exactly the sets $$S(\sigma) = M \cup \{T_i(\Phi,\sigma)\}_{1\le i\le n},$$  for all truth assignments $\sigma$ of $\Phi$, where $T_i(\Phi,\sigma) = \{\alpha_i,\beta_i\}$ if $\sigma(x_i)=\mbox{True}$, otherwise $T_i(\Phi,\sigma) = \{\notxi{\alpha}{i},\notxi{\beta}{i}\}$. See \autoref{fig:reducsep} for an example of $S(\sigma)$.

\textbf{From 3-SAT to minimal $(a,b)$-separators.} Let $\sigma$ be a truth assignment of $\Phi$, where $\sigma(x_i)$ takes value True or False. By definition of $\sigma$, for every clause, there is at least one literal that makes this clause true. We show that $S(\sigma)$, as defined above, is a minimal $(a,b)$-separator.

\begin{figure}[t]
\centering
\scalebox{0.65}{\input{tikz/reduc_3sat}}
\caption{Graph $G(\Phi)$ for the 3-SAT formula $\Phi$: $(x_1 \lor x_2 \lor x_3) \land (\notxi{x}{1} \lor x_2 \lor x_5) \land (\notxi{x}{3} \lor \notxi{x}{4} \lor \notxi{x}{5})$. Here, $M$ can be extended into $S(\sigma)$: vertices of $S(\sigma) \setminus M$ are drawn in blue. The truth assignment $\sigma$ is: $x_1,x_2$ are fixed True, $x_3,x_4,x_5$ are fixed False.}
\label{fig:reducsep}
\end{figure}

In any variable gadget, it is not possible to highlight an $(a,b)$-path which does not traverse $S$. Moreover, clause vertices belong to $M \subsetneq S$. Hence, $S(\sigma)$ is an $(a,b)$-separator. Let us prove its minimality: we must have $N(C(a,S(\sigma))) = N(C(b,S(\sigma))) = S(\sigma)$. As respectively one vertex among $\{\alpha_i,\notxi{\alpha}{i}\}$ and $\{\beta_i,\notxi{\beta}{i}\}$ does not belong to $S$, then all vertices of $S(\sigma)$ contained in variable gadgets can reach $a$ and $b$ despite $S(\sigma)$. Concerning clause vertices (which belong to $S(\sigma)$), they are adjacent to $a$ but their adjacency towards $C(b,S(\sigma))$ is ensured by the fact that $\sigma$ is a truth assignment. Indeed, for each $C_j$, there is at least one literal $\ell_i$ from $\overline{C_j}$ which is False: so, one vertex $\beta_i$ or $\notxi{\beta}{i}$ adjacent to $C_j$ is not in the separator $S(\sigma)$. Hence, any vertex $C_j$ is adjacent to  $C(b,S(\sigma))$. This completes the proof that $S(\sigma)$ is a minimal $(a,b)$-separator.

\textbf{From minimal $(a,b)$-separators to 3-SAT.} Consider a minimal \absep\ $S$ such that $M \subsetneq S$. Fix some arbitrary $1\le i\le n$ and $A_i = \{\alpha_i,\notxi{\alpha}{i},\beta_i,\notxi{\beta}{i}\}$. On one hand, we cannot have $\vert S \cap A_i \vert \le 1$, otherwise there is an $(a,b)$-path in $G \setminus S$. On the other hand, $\vert S \cap A_i \vert \ge 3$ would contradict the minimality of $S$. W.l.o.g., let $S \cap A_i = \{\alpha_i,\beta_i,\notxi{\beta}{i}\}$: there is no simple $(a,b)$-path in $G$ whose intersection with $S$ is the single vertex $\alpha_i$, contradicting \autoref{lem:ab-sep-paths}. Hence, $\vert S \cap A_i \vert = 2$. Cases $S \cap A_i = \{\alpha_i,\notxi{\beta}{i}\}$ or $\{\beta_i,\notxi{\alpha}{i}\}$ are not possible, otherwise $S$ is not an $(a,b)$-separator. Similarly, cases $\{\alpha_i,\notxi{\alpha}{i}\}$ and $\{\beta_i,\notxi{\beta}{i}\}$ contradict the minimality of $S$ (either $a_i$ or $b_i$ do not satisfy~\autoref{lem:ab-sep-paths}). All in all, the only possible choice for each $S \cap A_i$ is either $\{\alpha_i,\beta_i\}$ or $\{\notxi{\alpha}{i},\notxi{\beta}{i}\}$, confirming the form of sets $T_i(\Phi,\sigma)$. In brief, we just proved that any minimal $(a,b)$-separator is formed as $M \cup \{T_i(\Phi,\sigma)\}_{1\le i\le n}$, where $\sigma$ is a variable assignment. Let us show that $\sigma$ is indeed a truth assignment.

For each clause vertex $C_j$, there must be a $\beta$-vertex, {\em i.e.} $\beta_i$ or $\notxi{\beta}{i}$, adjacent to $C_j$ which does not belong to $S$: otherwise, \autoref{lem:ab-sep-paths} would be contradicted. In summary, if $\sigma$ is not a truth assignment, $S$ cannot be a minimal $(a,b)$-separator. For sure, $\sigma$ is a truth assignment, fulfilling the claim.

As a conclusion, there exists a minimal $(a,b)$-separator extending $M$ in $G(\Phi)$ iff the 3-SAT formula $\Phi$ admits a truth assignment.
\end{proof}

\subsection{NP-hardness of PMC extension}- \label{subsec:hard-pmc-ext}

Our proof that \textsc{PMC Extension} is NP-hard is strongly related to the construction of $G(\Phi)$ in the previous subsection. We proceed by showing that identifying a PMC extending $M$ in $G(\Phi)$ is equivalent to finding a minimal $(a,b)$-separator extending $M$ in $G(\Phi)$. First, we establish some properties of PMCs in graph $G(\Phi)$.

\begin{lemma}\label{lem:abnotsamecomp}
In graph $G(\Phi)$, we have $N[a] \cup N[b] = V(G(\Phi))$. For any PMC $\Omega$ of $G(\Phi)$ such that $\Omega \cap \{a,b\} = \emptyset$, there are exactly two connected components in $G \setminus \Omega$, which are $C(a,\Omega)$ and $C(b,\Omega)$.
\end{lemma}
\begin{proof}
From graph construction, $N[a] \cup N[b] = V(G(\Phi))$ and $N[a] \cap N[b] = \emptyset$. Assume that $a$ and $b$ belong to the same connected component $C$ of $G \setminus \Omega$, then $C$ is necessarily a full component for $\Omega$ since $\Omega \subseteq N(a) \cup N(b)$, contradicting~\autoref{thm:pmc_charac}. So, $G \setminus \Omega$ admits at least two connected components, containing respectively $a$ and $b$. There cannot be a third one $C'$, since $C' \subseteq N(a) \cup N(b)$.
\end{proof}

\begin{lemma}\label{lem:reduc_graph_sep_property}
For any \absep\ $S$ of $G(\Phi)$ extending $M$ and any vertex $x \in C(a,S)$ (resp. $\in C(b,S)$), there is $v \in S$ such that $N(v) \cap C(a,S) = \{x\}$ (resp. $N(v) \cap C(b,S) = \{x\}$).
\end{lemma}
\begin{proof}
Any $(a,b)$-separator $S$ of $G(\Phi)$ is a superset of some minimal $(a,b)$-separator of $G(\Phi)$, which is, from the claim of~\autoref{the:min_absep_ext_npc}, some $S(\sigma)$. Hence, for some truth assignment $\sigma$, $S(\sigma) \subseteq S$. Consider any $x \in C(a,S)$: it can thus be either $a$ or some vertex in $\{\alpha_i,\notxi{\alpha}{i}\}_{1\le i\le n}$. If $x = a$, then pick $v$ as any vertex $\alpha_i$ or $\notxi{\alpha}{i}$ belonging to $S$. The neighborhood of $v$ in $C(a,S)$ contains only $a$. Else, if $x = \alpha_i$ w.l.o.g. (case $x = \notxi{\alpha}{i}$ works similarly by symmetry), then it means that $\notxi{\beta}{i}$ belongs necessarily to $S$, hence with $v = \notxi{\beta}{i}$, we have $N(v) \cap C(a,S) = \{\alpha_i\}$.
We can achieve the same reasoning on side $C(b,S)$, the only difference being the adjacency between vertices in $\{\beta_i,\notxi{\beta}{i}\}_{1\le i\le n}$ with clause vertices, but since the latter belongs to $M \subsetneq S$, it has no impact.
\end{proof}

We pursue with a key lemma which presents the form of PMCs in $G(\Phi)$ which both extends $M$ and do not contain $a$ and $b$. In fact, they are very close to the separators $S(\sigma)$ introduced in~\autoref{the:min_absep_ext_npc}. 
In the remainder, for any $1\le j\le m$, we denote by $B(C_j)$ the set of $\beta$-vertices adjacent to the clause vertex $C_j$, formally $B(C_j) = N(C_j) \cap \{\beta_i,\notxi{\beta}{i}\}_{1\le i\le n}$.

\begin{lemma}\label{lem:pmc_ti=3}
If $\Omega$ is a PMC of $G(\Phi)$ extending $M$ and such that $a,b \notin \Omega$, then the following statements hold:
\begin{itemize}
\item there is exactly one index $1\le i\le n$ such that $\vert T_i \vert = 3$ with $T_i = \{\alpha_i, \notxi{\alpha}{i}, \beta_i, \notxi{\beta}{i}\} \cap \Omega$;
\item for all other indices $1\le p\neq i\le n$, we have $\vert T_p \vert = 2$ with $T_p = \{\alpha_p,\beta_p\}$ or $T_p = \{\notxi{\alpha}{p},\notxi{\beta}{p}\}$.
\item for any $1\le j\le m$, $B(C_j) \setminus \Omega \neq \emptyset$.
\end{itemize}
\end{lemma}
\begin{proof}
We proceed by case disjunction. From \autoref{lem:abnotsamecomp}, vertices $a$ and $b$ cannot belong to the same connected component. For any $1\le i\le n$, if $\vert T_i \vert = 0$ or $1$, there would be an $(a,b)$-path in $G \setminus \Omega$, which is impossible. Case $\vert T_i \vert = 4$ cannot occur either because there would be neither a direct edge nor a path with internal vertices in $\overline{\Omega}$ between $\alpha_i$ and $\beta_i$. Indeed, we are blocked not only by vertices in $T_i$ but also by clause vertices. 

In the following, we show that we cannot have $\vert T_i \vert = 2$ for all $1\le i\le n$. First, we exclude the case where $\Omega$ contains, for some index $i$, pairs $\alpha_i \notxi{\beta}{i}$ or $\notxi{\alpha}{i} \beta_i$ because it leads to a path from $a$ to $b$ in $G \setminus \Omega$ (\autoref{lem:abnotsamecomp}). Then, there is only three configurations for each $i$ to consider :
\begin{enumerate}[label=(\arabic*)]
    \itemsep0em 
    \item $T_i = \alpha_i \notxi{\alpha}{i}$,
    \item $T_i = \beta_i \notxi{\beta}{i}$
    \item $T_i = \alpha_i \beta_i$ or $\notxi{\alpha}{i} \notxi{\beta}{i}$
\end{enumerate}
If all $x_i$ are in the case $(1)$, then we have $C(b,\Omega)$ has a full component for $\Omega$. Same for the case $(2)$ and $C(a,\Omega)$. If all $x_i$ are in case $(3)$, then $C(a,\Omega)$ is a full component for $\Omega$. If we mix only the cases $(1)$ and $(2)$, then we have the $b_i$ of an index $i$ in case $(1)$ that cannot see an $a_{i'}$ of an index ${i'}$ in case $(2)$ despite $\Omega$. If we blend the case $(1)$ and the case $(3)$, then we have either no full component if there is a $C_j$ such that $B(C_j) \setminus \Omega = \emptyset$ but vertex $C_j$ will not reach any $b_{i'}$ in the case $(1)$. If we mix case $(2)$ and the case $(3)$, then $C(a,S)$ is a full component for $\Omega$.

In brief, we cannot have $\vert T_i \vert = 2$ for all indices $1\le i\le n$. We finally show that we cannot have more than one $x_i$ with $\vert T_i \vert = 3$. With $\vert T_i \vert = 3$, there is a vertex $u\in T_i$ not adjacent to $C(a,\Omega)$ and another vertex $v$ not adjacent to $C(b,\Omega)$. For any pair of $x_i, x_{i'}$ such that $\vert T_i \vert = \vert T_{i'} \vert = 3$, the vertex $u$ (resp. $v$) of $T_i$ cannot reach the vertex $v'$ (resp $u'$) of $T_{i'}$.

At last, if there is exactly one $x_i$ with $\vert T_i \vert = 3$, then all other $x_p$, with $p\neq i$ have $\vert T_p \vert = 2$. Furthermore, variable gadgets with $\vert T_p \vert = 2$ cannot be in the case $(1)$ or $(2)$ because they would create a conflict pair with the gadget $T_i$. Indeed, $T_i$ contains exactly one vertex not adjacent to $C(a,\Omega)$ and one vertex not adjacent to $C(b,\Omega)$. Another gadget $T_p$ in case $(1)$ or $(2)$ contains either a vertex not adjacent to $C(a,\Omega)$ or a vertex not adjacent to $C(b,\Omega)$.
As a conclusion, all $\vert T_p \vert = 2$ are in the case $(3)$. Furthermore, we have also that $\forall j, B(C_j) \setminus \Omega \neq \emptyset$ because otherwise, gadgets in case $(3)$ would only admit $C(a, \Omega)$ as a full component. It follows that the vertex of the gadget with $\vert T_i \vert = 3$ which is disconnected to $C(a,\Omega)$ would not reach one clause vertex with covered $\beta$-vertices.
\end{proof}


These different elements on the structures of PMCs extending $M$ and avoiding $a,b$ reveal the form already found in~\autoref{the:min_absep_ext_npc}: these PMCs seem strongly related to the truth assignment of the SAT formula $\Phi$. We confirm this intuition, showing that \textsc{PMC Extension} and \textsc{Minimal $(a,b)$-separator Extension} are equivalent problems on this setting.

\begin{lemma}\label{lem:PMC_eq_ab_sep}
	In $G(\Phi)$, there exists a PMC $\Omega$ extending $M$ and such that $a,b \notin \Omega$ iff there exist a minimal $(a,b)$-separator $S$ extending $M$.
\end{lemma}
\begin{proof}
\textbf{Only if direction}. Given a minimal $(a,b)$-separator $S$ of $G(\Phi)$ extending $M$, $G(\Phi) \setminus S$ admits two full components $C(a,S)$ and $C(b,S)$. We can add an extra vertex $x$ of any variable gadget into $S$: we say w.l.o.g. that $x \in C(a,S)$. By \autoref{lem:reduc_graph_sep_property}, there exists $v \in S$ such that $N(v) \cap C(a,S) = \{x\}$, so for set $X = S \cup \{a\}$, $C(a,X)$ is not a full component for $X$. Moreover, $C(b,X)$ is not a full component either for $X$, since $x \in C(a,S)$. Set $X$ is shy. Moreover, as $C(a,S)$ was a full component for $S$, vertex $x$ is enlightened for $X$. As $C(b,S)$ was also a full component for $S$, vertices of $S$ see each other via this connected component. Consequently, $X$ has the reachability property and is a PMC.

\textbf{If direction}. Given a PMC $\Omega$ satisfying $M \subsetneq \Omega$ and $\Omega \cap \{a,b\} = \emptyset$, we know by \autoref{lem:pmc_ti=3} that we have two vertices picked up in every gadgets except for one where there are three. By \autoref{lem:abnotsamecomp}, we know that $G(\Phi) \setminus \Omega$ has exactly two connected components: $C(a,\Omega)$ and $C(b,\Omega)$. Now consider the variable gadget with $\vert T_i \vert = 3$. We distinguish two cases:
	\begin{itemize}
		\item $T_i$ contains all $\beta$-vertices: in this setting, the separator $S = N(C(a,\Omega))$ contains every vertex of $M$. It follows, by \autoref{le:neighbor-sep}, that $S$ is a minimal $(a,b)$-separator extending $M$.
		\item $T_i$ contains all $\alpha$-vertices: let us fix $S = N(C(b,\Omega))$. If $S$ contains every vertex of $M$, then again $S$ is a minimal $(a,b)$-separator extending $M$. Assume by way of contradiction that it is not the case. Component $C(b,\Omega)$ is adjacent to all vertices in $\{a_i,b_i\}_{1\le i\le n}$. So, $C(b,\Omega)$ is not adjacent to some clause vertex $C_j$.
Then, in the gadget where $\vert T_i \vert = 3$, since we added an $\alpha$-vertex, vertex $b_i$ is not adjacent to $C(a,\Omega)$. Finally, as $C_j$ is not adjacent to $C(b,\Omega)$ and $(b_i,C_j) \notin E(G(\Phi))$, then $(b_i,C_j)$ is a conflict pair for $\Omega$, which yields a contradiction.
	\end{itemize}
In brief, the existence of a PMC $\Omega$ extending $M$ and avoiding $\{a,b\}$ allows us to identify a minimal $(a,b)$-separator extending $M$ in $G(\Phi)$.
\end{proof}



The hardness of \textsc{PMC Extension} thus comes naturally.

\hardness*
\begin{proof}
By \autoref{the:min_absep_ext_npc}, it is NP-hard to find a minimal $(a,b)$-separator extending $M$ in graphs $G(\Phi)$. By \autoref{lem:PMC_eq_ab_sep}, finding such a separator is equivalent to solving \textsc{PMC Extension} on graphs $G(\Phi)$ with $M$ fixed as in~\autoref{subsec:hardness-sep} and $F = \{a,b\}$. 
\end{proof}

\subsection{NP-hardness of PMC subset}

In this subsection, we provide a more refined hardness result regarding \textsc{PMC Extension}. We prove that \textsc{PMC Extension} is NP-hard, even restricted to $F =\emptyset$. This subproblem is called \textsc{PMC Subset} and defined below.

\defdecproblem{\textsc{PMC Subset}}{A graph $G = (V,E)$, a set $M \subseteq V$ of \textit{mandatory} vertices.}{Is there a PMC $\Omega$ of $G$ satisfying $M \subseteq \Omega$?}

The avoidance of set $F = \{a,b\}$, used as an essential argument in the previous subsection, can in fact, with additional effort, be withdrawn to draw the equivalence between PMCs and minimal $(a,b)$-separators extending $M$.

\begin{theorem}
	In $G(\Phi)$, there is no PMC $\Omega$ extending $M$ containing both $a$ and $b$. Moreover, if $a$ or $b$ belongs to $\Omega$, then $\Omega$ contains a minimal $(a,b)$-separator extending $M$.
\end{theorem}
\begin{proof}
	We assume that formula $\Phi$ is non trivial and contains at least one clause and two variables, hence $G(\Phi)$ admits at least one clause vertex and two variable gadgets. In this setting, if both $a$ and $b$ belong to $\Omega$, then, any mandatory vertex $a_i$ (or $b_i$) can only reach an $a_p$ (or $b_p$), $p\neq i$, by either $a$, $b$ or a clause vertex. Since they all are in $\Omega$, it follows that $\Omega$ is not a PMC: the reachability property is not satisfied. Then both $a$ and $b$ cannot belong to the same PMC.
	
	Let us suppose now that exactly one vertex among $a$ or $b$ belongs to $\Omega$. Starting from set $Y = M \cup \{a\}$ or $Y =M \cup \{b\}$, graph $G(\Phi) \setminus Y$ contains a component full for $Y$. It is thus necessary to add vertices from variable gadgets if we want to obtain a PMC $\Omega$ fulfilling $Y \subsetneq \Omega$. Our objective is to show that $a$ and $b$ cannot see each other despite $\Omega$, formally $a \notin R_G^+(b,\Omega)$.
	
	We assume w.l.o.g. that $b \in \Omega$: the same symmetric reasoning works for case $a \in \Omega$. Then, $S = N(C(a,\Omega)$ is a minimal separator of $G$. Assume by contradiction that $b \in S$ and let us denote by $v \in \Omega$ a vertex separated from $a$ by $S$. From~\autoref{lem:ab-sep-paths}, there exists an $(a,v)$-path $P$ intersecting $S$ only in $b$. Path $P$ necessarily traverses a variable gadget to go from $a$ to $b$ without intersecting $S$, since clause vertices belong to $S$ (they are adjacent to $a$). Hence, some ``column'' of a variable gadget (of the form $\notxi{\alpha}{i},\beta_i$ or $\alpha_i,\notxi{\beta}{i}$) does not belong to $S$, so to $\Omega$. As this column is adjacent to $a_i,b_i \in M$, then, $a_i,b_i \in S$. If $v \notin \{\beta_i,\notxi{\beta}{i}$, we have a contradiction: it is not possible to build an $(a,v)$-path intersecting $S$ only in $a_i$, this path will necessarily intersect either $b$ or a clause vertex, which are also in $S$. Else, if $v \in \{\beta_i,\notxi{\beta}{i}$, then it is not possible to build an $(a,v)$-path interesecting $S$ only in $a_p$, with $p \neq i$, without passing through either $b$ or a clause vertex. Consequently, $S$ is not minimal, contradicting $b \in S$.
Thus, $\Omega$ contains a minimal $(a,b)$-separator $S$. 

We pursue with the hypothesis $b \in \Omega$ since it does not influence the validity of this explanation. We must necessarily have at least two vertices per gadget to obtain an $(a,b)$-separator $S = N(C(a,\Omega))$. Moreover, for each variable gadget, we cannot pick only the $\beta$-vertices, as it would create a conflict pair $(b,a_i)$. Same observation for $\alpha$-vertices: it would create a conflict pair $(b_i,b_p)$, with any $p \neq i$. In summary, for any variable gadget $1\le i \le n$, two of the vertices in $T_i = \Omega \cap \{\alpha_i,\beta_i,\notxi{\alpha}{i},\notxi{\beta}{i}\}$ form a diagonal pair, for instance $\alpha_i,\beta_i$.  In addition with it, if we have $\vert T_i \vert = 3$ for some index $i$, then a vertex of $T_i$ will for sure not see $b$. So, for any $1\le i\le n$, $\vert T_i \vert = 2$ and $T_i$ is made up with diagonal pairs. In this situation, $C(a,\Omega)$ is adjacent to all vertices of $\Omega$ except $b$, in particular $M \subsetneq N(C(a,\Omega)) = S$. As a conclusion, $\Omega$ contains a minimal $(a,b)$-separator $S$ extending $M$.
\end{proof}

\begin{corollary}
	\textsc{PMC Subset} is NP-complete.
\end{corollary}

\bibliographystyle{plain}
\bibliography{biblio}

\begin{appendix}
\end{appendix}

\end{document}

%% file: preliminaries.tex
In this section, we introduce some general notation and concepts related to connectivity in graphs. We also define the notion of \emph{Potential Maximal Clique} (PMC) and provide some observations.

\subsection{General graph notations}

All graphs mentioned in this paper are simple, finite, undirected, unweighted and connected. Any graph $G = (V, E)$ is given with $V$ a set of $n$ vertices and $E \subseteq V \times V$ a set of $m$ edges. As we will sometimes handle several graphs in parallel, function $V(\cdot)$ will provide the vertex set of a graph, and $E(\cdot)$ its edge set: here, $V(G) = V$ and $E(G) = E$. Given a set of vertices $U \subseteq V$, the notation $\overline{U}$ refers to the complementary of $U$ in $V$: $\overline{U} = V \setminus U$. We say two vertices $u$ and $v$ are adjacent in $G$ if $(u,v) \in E$. For any vertex $u \in V$, we denote respectively by $N(u) = \{v : (u,v) \in E\}$ the \textit{open neighborhood}\footnote{We will use notation $N_G(u)$ when $G$ is not absolutely clear from the context.} of $u$ and $N[u] = N(u) \cup \{u\}$ its \textit{close neighborhood}. These notions are naturally generalized to sets: for any $U \subseteq V$, $N(U) = \{v: (u,v) \in E, u \in U, v \notin U\}$ and $N[U] = U \cup N(U)$.

A \emph{path} $P$ from $a$ to $b$ of $G$ is a sequence of vertices $(v_0,v_1,v_2,\ldots,v_{k-1},v_k)$ such that $a = v_0$, $b = v_k$, and $(v_i,v_{i+1}) \in E$ for any $0\le i < k$. Concisely, we say that $P$ is an $(a,b)$-path. A \emph{simple path} is a path without any repetition of vertices. From now on, the term \emph{path} will automatically refer to the notion of simple path.
A graph $G$ is \emph{connected} if for any pair of vertices $a,b \in V$, $a \neq b$, there is an $(a,b)$-path in $G$.

A \emph{subgraph} $G'$ of $G$ is a graph $G' = (V', E')$ made up of a subset of vertices/edges of $G$, formally $V' \subseteq V$ and $E' \subseteq E \cap (V' \times V')$. For any $V' \subseteq V$, we denote $G[V'] = (V', E \cap (V' \times V'))$ the \emph{induced subgraph} of $V'$. Its edge set will be written briefly as $E[V'] = E \cap (V' \times V')$. Given a vertex set $U \subseteq V$, let $G \setminus U$ denote the graph $G$ deprived of vertices in $U$, formally $G \setminus U = G[\overline{U}]$. More generally, given two sets $A$ and $B$ such that $A \subseteq B$, then $B \setminus A$ denote the set $B$ deprived of elements of $A$.

Finally, we distinguish between the notion of \emph{subset} which includes equality (symbol $\subseteq$) and \emph{proper subset} which refers necessarily to a smaller set (symbol $\subsetneq$).

\subsection{Connectivity and separators}

\subsubsection{Connectivity}

We introduce now standard notations of connectivity, beginning with connected components.

\begin{definition}[Connected components]
Consider a graph $G$. Let $U$ be a set of vertices $U \subsetneq V(G)$ and let $a \notin U$. We denote respectively by:
\begin{itemize}
\item $C_G(a,U)$ the \emph{connected component} of $G \setminus U$ which contains $a$,
\item $C_G[a,U]$ the corresponding \emph{close connected component}: \[C_G[a,U] = C_G(a,U) \cup \left(U \cap N_G(C_G(a,U))).\right.\]
\end{itemize}
When $G$ is clear from the context, we will sometimes remove it from subscript and use simply $C(a,U)$ and $C[a,U]$.
\label{def:connected-comp}
\end{definition}

We pursue with the concept of \emph{full component}, which is a common notion related to both minimal separators and PMCs. A full component is simply a connected component which admits the whole removed set in its neighborhood.

\begin{definition}[Full components~\cite{BoTo02}]
Consider a graph $G$ and a vertex set $U \subseteq V(G)$. We say a connected component $C$ of $G \setminus U$ is a \emph{full component} for $U$ if the neighborhood of the component in $G$ is exactly $U$, formally $N(C) = U$.
\label{def:full-comp}
\end{definition}

Said differently, if some connected component $C_G(a,U)$ is a full component for $U$, then $C_G[a,U] = C_G(a,U) \cup U$. We will see in the remainder that a minimal separator $S$ admits full components, but this is not the case for a PMC.

\begin{definition}[Shy set]
Consider a graph $G$ and a vertex set $U \subseteq V(G)$. We say that $U$ is a \emph{shy set} if there is no connected component of $G \setminus U$ which is a full component for $U$. Formally, for any connected component $C$ of $G \setminus U$, $N_G(C) \neq U$.
\label{def:shy-set}
\end{definition}

Observe that if some set $U$ is shy, then naturally any of its superset is also shy, otherwise a full component for the superset would imply necessarily a full component for $U$ in $G \setminus U$. Trivially, $V(G)$ is always a shy set.

\begin{observation}
Given a shy set $U \subseteq V(G)$ of $G$, any of its superset $W \supseteq U$ is also shy.
\label{obs:shy-set}
\end{observation}

Another common notion widely used in the context of treewidth is the \emph{torso}.

\begin{definition}[Torso~\cite{MaOsRa10,RoSe91}]
Consider a graph $G$ and some set $U \subseteq V(G)$. The graph $\torso(G, U)$ has vertex set $U$ and
its edge set is made up of all edges $(u,v) \in U \times U$ satisfying either $(u,v) \in E(G)$ or there is a $(u,v)$-path of $G$ whose internal vertices are not in $U$.
\end{definition}

Torsos characterize, given a vertex set $U$, its connectivity outside $U$, {\em i.e.} either through direct edges between two vertices or through paths outside $U$. For tree decompositions, torsos are used to characterize the connectivity between vertices of the same bag. In this paper, the torso will be a tool manipulated to not only define PMCs but also to introduce their properties.

We extend the notion of neighborhood into a larger (natural) one that we call \emph{reachability}. Again, if we consider some vertex set $U \subseteq V(G)$, reachability defines, given some vertex $u \in U$, the set of vertices which can be reached from $u$ using only intermediary vertices outside $U$.

\begin{definition}[Reachability]
Let $U$ be a set of vertices $U \subseteq V(G)$ and a vertex $u \in U$. We define as $R_G(u,U)$ the set of vertices of $U$ which can be reached from $u$ either directly (with an edge) or via intermediary vertices which are outside $U$. 

Two equivalent and more formal definitions follow. $v \in R_G(u,U)$ if and only if:
\begin{itemize}
\item either $(u,v) \in E(G)$ or there is a path in $G$ starting from $u$ and arriving at $v$ such that all its internal vertices belong to $\overline{U} = V(G) \setminus U$.
\item $(u,v) \in E(\emph{\torso}(G,U))$.
\end{itemize}
\label{def:reachable-set}
\end{definition}

Note that in this definition, set $R_G(u,U)$ is a subset of $U$ and can be seen as the neighborhood of $u$ in $\torso(G,U)$. Nevertheless, we will sometimes need a broader notion of reachability which contains the vertices outside $U$ also. 

\begin{definition}[Global reachability]
Let $U$ be a set of vertices $U \subseteq V(G)$ and some vertex $a \in V(G)$. We fix $R_G^+(a,U)$ the set of vertices of $G$ which can be reached from $a$ either via some direct edge or through some path whose internal vertices belong to $\overline{U}$.
\label{def:global-reachable}
\end{definition}

\begin{figure}[h]
	\centering
	\begin{subfigure}{0.55\textwidth}
	\centering
	\scalebox{0.8}{\input{tikz/torso}}
	\caption{Graph $G$, set $U = \{u,v,x,y\}$. We have $a \in R_G^+(v,U)$ and $x \in R_G(v,U) \subsetneq R_G^+(v,U)$ via $C_G(a,U)$. Moreover, $u \in R_G(v,U)$ via direct edge and $y \in R_G(v,U)$ via the connected component on top of the figure.}
	\end{subfigure} ~
	\begin{subfigure}{0.4\textwidth}
	\centering
	\scalebox{0.8}{\input{tikz/torso2}}
	\caption{Graph $\torso(G,U)$. Vertices $u$ and $x$ are connected neither by a direct edge nor via a path whose all internal vertices belong to $\overline{U}$.}
	\end{subfigure}
	\caption{Illustration of the notions of torso and reachability.}
	\label{fig:torso}
\end{figure}

Note that there are two changes compared to \autoref{def:reachable-set}: first vertex $a$ is not necessarily in $U$, second set $R_G^+(a,U)$ is not a subset of $U$: it also contains vertices of $\overline{U}$ reachable from $a$ via paths outside $U$. In the case $u \in U$, then $R_G(u,U) = R_G^+(u,U) \cap U$. All these aspects are illustrated in~\autoref{fig:torso}. We can now introduce some vocabulary related to reachability.

\begin{definition}[Reachability property and conflicts]
Let $U \subseteq V(G)$ be a set of vertices and $u,v \in U$. We say that:
\begin{itemize}
\item Vertices $u$ and $v$ are \emph{visible despite} $U$ if $v \in R_G(u,U)$ - equivalently $u \in R_G(v,U)$ as $G$ is undirected. Another formulation we will use for this statement is either (i) $u$ sees $v$ despite $U$ or (ii) $u$ and $v$ \emph{see each other} despite $U$.
\item Conversely, vertices $u$ and $v$ are \emph{in conflict} for $U$ if $v \notin R_G(u,U)$. Another possible formulation used is that $(u,v)$ is a \emph{conflict pair} for $U$.
\item Vertex $u$ is \emph{enlightened} for $U$ if it is not involved in any conflict pair: $R_G(u,U) = U$.
\item Set $U$ satisfies the \emph{reachability property} if, for any pair $x,y \in U$, $x$ and $y$ are visible despite $U$.
\end{itemize} 
\label{def:reachable}
\end{definition}

Naturally, by considering a subset $U' \subseteq U$ where $U$ has the reachability property, we preserve the property as the graph $G \setminus U'$ is less restricted than $G \setminus U$ in terms of connectivity. More concretely, two reachable vertices $x,y \in U'$ in $G \setminus U$ are still reachable in $G \setminus U'$. Trivially, the empty set and singletons always have the reachability property.

\begin{observation}
Given a set $U$ that satisfies the reachability property, any subset of $U$ also satisfies it.
\label{obs:reachable}
\end{observation}

\subsubsection{Minimal separators}

As we will see it further, PMCs of a graph are closely related to its minimal separators. We begin with a clear presentation on the different shades of separators.

\begin{definition}[Minimal separators and set $\Delta_G$]
Consider a graph $G$, two vertices $a,b \in V(G)$ and some set $S \subseteq V(G) \setminus \{a,b\}$. We say that:
\begin{itemize}
\item $S$ is an $(a,b)$-\emph{separator} if vertices $a$ and $b$ belong to two different connected components in $G \setminus S$.
\item $S$ is a \emph{minimal} $(a,b)$-separator if it is an $(a,b)$-separator and no proper subset of $S$ is an $(a,b)$-separator.
\item $S$ is a \emph{minimal separator} of $G$ if we can identify two vertices $a,b \in V(G)$ such that $S$ is a minimal $(a,b)$-separator.
\item $S$ is an \emph{inclusion-wise minimal separator} if none of its proper subset is a minimal separator.  
\end{itemize}
We denote by $\Delta_G$ the collection of minimal separators of a graph $G$.
\label{def:separators}
\end{definition}

As a first remark, two vertices $a$ and $b$ can be vertex-separated if they are not adjacent, since an $(a,b)$-separator cannot contain any of them. Observe also that two inclusion-wise minimal separators cannot be subsets of each other, by definition. However, we can have $S' \subsetneq S$ such that both $S'$ and $S$ are minimal separators: in this case, $S'$ is a minimal separator for some pair of vertices - say $(a',b')$ - and $S$ is a minimal separator of another pair $(a,b) \neq (a',b')$. Set $S$ also separates $(a',b')$ but not minimally.

Naturally, the set $\Delta_G$ can be of exponential size compared to $n = \vert V(G) \vert$: as an example, consider some \emph{melon graph} which consists in $p(n) = \frac{n-2}{2}$ disjoint paths of length $3$ between two vertices $a$ and $b$. There are exactly $2^{p(n)}$ minimal $(a,b)$-separators in such graph. Minimal separators of a graph $G$ can be enumerated in polynomial delay~\cite{BeBoCo99,Ta10}, while the list of inclusion-wise minimal separators of $G$ cannot be produced in output-polynomial time~\cite{BoDeKuLi24} unless $P = NP$.

\begin{lemma}[Enumeration of minimal separators~\cite{BeBoCo99}]
All minimal separators of $G$ can be enumerated with a delay of $O(n^3)$, providing a global running time of $O(n^3\vert \Delta_G \vert)$.
\label{le:enum-sep}
\end{lemma}

Plenty of folklore statements exist on minimal separators. We will recall some of them which will be useful for our results.

\begin{lemma}[Minimal separators and full components~\cite{BoTo02}]
For any graph $G$ with $a,b \in V(G)$ and some set $S \subseteq V(G) \setminus \{a,b\}$, $S$ is a minimal $(a,b)$-separator if and only if (iff) $a$ and $b$ belong to two different full components for $S$.

As a consequence, for each $s \in S$, there is a path from $a$ to $b$ intersecting $S$ only in $s$.
\label{lem:ab-sep-paths}
\end{lemma}

Another folklore notion, present in many papers dealing with minimal separators, is the relationship between two minimal separators of the same pair: either they are \emph{crossing} or \emph{parallel}. We need some property on parallel separators for the remainder.

\begin{definition}[Parallel separators]
Two minimal $(a,b)$-separators $S$ and $S'$ of a graph $G$ are \emph{parallel}, if either $C_G(a,S) \subseteq C_G(a,S')$ or $C_G(a,S') \subseteq C_G(a,S)$.
\label{def:parallel-sep}
\end{definition}

It is equivalent to state that two parallel $(a,b)$-separators satisfy respectively either $C_G(b,S') \subseteq C_G(b,S)$ or $C_G(b,S) \subseteq C_G(b,S')$.
The following lemma is a direct consequence of \autoref{lem:ab-sep-paths}.

\begin{lemma}
Let $S$ and $S'$ be two parallel minimal $(a,b)$-separators, $S \neq S'$, and let us say w.l.o.g. that $C_G(a,S) \subsetneq C_G(a,S')$.

For any vertex $x \in S$, there exists a path $P_x$ from $x$ to some $y \in S'$ whose internal vertices do not belong to $S \cup S' \cup \{a,b\}$.
\label{lem:parallel-sep}
\end{lemma}
\begin{proof}
We distinguish two cases. If $x \in S \cap S'$, then trivially there exists a one-vertex path $P_x = \{x\}$ without internal vertices and whose arrival (identical to departure, $y = x$) belongs to $S'$. Thus, we focus on the case $x \in S \setminus S'$.

According to \autoref{lem:ab-sep-paths}, there exists a path from $a$ to $b$ passing through $x$, hence we can consider its subpath from $x$ to $b$. We denote by $P$ this subpath. From \autoref{lem:ab-sep-paths} again, $C(a,S)$ is full component for $S$, so $x \in N(C(a,S))$. But $C(a,S) \subseteq C(a,S')$ and $x \in S'$, so $x \in C(a,S')$. Since $x \in C(a,S')$ and $S'$ is an $(a,b)$-separator, path $P$ necessarily intersects $S'$. Let $y$ be the first vertex of $P$ intersecting $S'$: we fix $P_x$ as the subpath of $P$ starting at $x$ and arriving at $y$. As the only vertex of $P$ intersecting $S$ is $x$ and since $y$ is the first vertex of $S'$ in $P$, the internal vertices of $P_x$ do not belong to $S \cup S'$. Moreover, as $P_x$ is a subpath of some $(a,b)$-path, its internal vertices cannot contain $a$ or $b$.
\end{proof}

\subsubsection{Close separators}

A well-known concept on minimal separators is recalled here: the unicity of a minimal $(a,b)$-separator contained in $N_G(a)$, for two vertices $a,b$ not adjacent. Such unique separator was called the \emph{close separator} by Takata in~\cite{Ta10} and comes from a more general definition. An example is illustrated in \autoref{fig:closesep}.

\begin{lemma}[Close separator~\cite{Ta10}]\label{lem:close-sep}
Let $A \subsetneq V$ be a connected set (\emph{i.e.} such that $G[A]$ is connected) where $a \in A$, $b \notin N[A]$. The \textit{close separator} $S$ of $(A,b)$ is the unique minimal $(a,b)$-separator included in $N(A)$, denoted $\sepag{G}{A}{b}$. Formally, $$\sepag{G}{A}{b} = N\left(C(b,N(A))\right).$$
The close separator $\sepag{G}{A}{b}$ can be computed in linear time $O(m)$. When $G$ is clear from the context, we might use notation $\sepa{A}{b}$ instead of $\sepag{G}{A}{b}$.
\end{lemma}

\begin{figure}[h]
\centering
\scalebox{0.9}{\input{tikz/a_b_close_sep}}
\caption{Schema of a close separator $\sepa{A}{b}$ in blue. The red area highlights $N(A)$.}
\label{fig:closesep}
\end{figure}

Given a pair $(a,b)$ of nonadjacent vertices, one close separator is particularly interesting: the one for which $A = \{a\}$.

\begin{definition}[Close $(a,b)$-separator~\cite{Ta10}]
Given a graph $G$ with two vertices $a,b$ such that $(a,b) \notin E(G)$, the \emph{close} $(a,b)$-\emph{separator} is exactly $\sepag{G}{\{a\}}{b}$.
\end{definition}

The close $(a,b)$-separator is thus the unique minimal $(a,b)$-separator included into $N(a)$. We need a property dealing with the close $(a,b)$-separator and its relationship with other minimal $(a,b)$-separators.

\begin{lemma}[Close $(a,b)$-separators and parallelism]
Let $(a,b) \notin E(G)$, $S^*$ the close $(a,b)$-separator, and $S$ be any other minimal $(a,b)$-separator: $S \neq S^*$. Then, $S^*$ and $S$ are parallel and $C(a,S^*) \subsetneq C(a,S)$.
\label{lem:parallel-close}
\end{lemma}
\begin{proof}
Our first goal is to show that any vertex of $S^*$ belongs to either $S$ or $C(a,S)$. Assume by way of contradiction that some $v \in S^*$ belongs to $C(b,S)$. As $S^* \subseteq N(a)$, then $(a,v) \in E(G)$: this yields a contradiction as $v \in C(b,S)$ and $a$ must be separated from $C(b,S)$ by $S$.

Consequently, $S^* \subseteq S \cup C(a,S)$. Since $C(b,S)$ induces a connected subgraph of $G$ and $S^* \cap C(b,S) = \emptyset$, then $C(b,S) \subseteq C(b,S^*)$. As $S^*$ and $S$ are both minimal $(a,b)$-separators, one cannot be included into the other. As $S \setminus S^*$ is nonempty, it contains a vertex of $C(b,S^*)$. Hence, $C(b,S) \subsetneq C(b,S^*)$ and $C(a,S^*) \subsetneq C(a,S)$ as announced.
\end{proof}

\begin{definition}
Let $G$ be a graph with two vertices $a,b$ such that $(a,b) \notin E(G)$, and $S \subseteq V(G) \setminus \{a,b\}$. We say that $S$ \emph{separates closely} the pair $(a,b)$ if $S$ is either the $(a,b)$-close separator or the $(b,a)$-close separator in $G$.
\label{def:sep-close}
\end{definition}

Directly inspired by the close separators defined by Takata in~\cite{Ta10}, we define here a generalization of these separators.

\begin{definition}[$X$-Close $(a,b)$-separator]
    Consider a graph $G$, a pair $a,b \in V(G)$ and a set $X \subseteq V(G)$ such that $a \in X$, and $X \setminus \{a\}$ is an $(a,b)$-separator. Set $S \subseteq X \setminus \{a\}$ is the $X$-close $(a,b)$-separator if:
    \begin{center}
    $S = \sepag{G}{A}{b}$ with $A = C_G(a,X \setminus \{a\})$.
    \end{center}
    By \autoref{lem:close-sep}, the $X$-close $(a,b)$-separator is a minimal $(a,b)$-separator, is unique and can be computed in $O(m)$ as $A$ can be deducted in $O(m)$. Moreover, $S \subsetneq N(A) = R_G(a,X)\setminus \{a\}$.
    \label{def:x-close}
\end{definition}

Concretely, the $X$-close $(a,b)$-separator is the unique minimal $(a,b)$-separator included into $X$ which is ``as close as possible'' from $a$, or said differently with the smallest connected component containing $a$ in $G \setminus S$. The close $(a,b)$-separator - without any mention of $X$ - is the minimal $(a,b)$-separator of the entire graph as close as possible from $a$. Indeed, it is  also the $X^*$-close $(a,b)$-separator, with $X^* = V(G)$. In other words, the $X$-close $(a,b)$-separator can be seen as the extension of the close $(a,b)$-separator to a more restricted area of the graph. 

As for \autoref{lem:parallel-close}, the $X$-close $(a,b)$-separator is parallel to any minimal $(a,b)$-separator included into $X$. 

\begin{lemma}[$X$-Close $(a,b)$-separators and parallelism]
Let $X \subseteq V(G)$, $(a,b) \notin E(G)$, $a \in X$, and $S^*$ be the $X$-close $(a,b)$-separator. Consider any other minimal $(a,b)$-separator $S$ included into $X$: $S \subseteq X$ and $S \neq S^*$. Then, $S^*$ and $S$ are parallel and $C_G(a,S^*) \subsetneq C_G(a,S)$.
\label{lem:parallel-close-x}
\end{lemma}
\begin{proof}
The arguments are very similar to the ones used for \autoref{lem:parallel-close}. Assume by way of contradiction that some $v \in S^*$ belongs to $C_G(b,S)$. As $S^*$ is $X$-close for $(a,b)$, we have $S^* \subseteq R_G(a,X)$. Vertex $a$ must be separated from $v \in C_G(b,S)$ by $S$, but $v \in R_G(a,X)$: $v$ is reachable from $a$ despite $X$. This contradicts the fact that $S \subseteq X$ separates pair $(a,v)$. Therefore, $S^* \subseteq S \cup C_G(a,S)$ and we conclude as in the proof of \autoref{lem:parallel-close}.
\end{proof}

More generally, we say that an $X$-close separator is the $X$-close $(a,b)$-separator for some pair $(a,b)$. In the remainder, $X$-close separators stand as a key tool to characterize PMC supersets (\autoref{subsec:charac-pmc}).

\subsection{A story of potential maximal cliques}

Initially, PMCs were defined as a key tool for designing a divide-and-conquer algorithm computing the treewidth of an input graph~\cite{BoTo98}. Their original definition comes with the notion of \emph{triangulation}. 
A graph $G$ is \textit{chordal} if every cycle of length at least $4$ has a chord. A triangulation represents a way to add edges to some graph $G$ in order to make it chordal. Treewidth and triangulations are very close concepts: computing the treewidth of a graph amounts to finding a triangulation with the smallest clique size.

\begin{definition}[Triangulation \cite{He05}]
Let $G = (V,E)$ be any graph. We say that a graph $H$ is a \emph{triangulation} of $G$ if both $H$ is chordal and is obtained by adding edges to $G$, {\em i.e.} $H = (V, E \cup F)$, where $E \cap F = \emptyset$.

A \emph{minimal triangulation} $H$ of $G$ is a triangulation of $G$ for which $H \setminus \{e\}$ is not chordal, for any $e \in F$.
\label{def:triangulation}
\end{definition}

PMCs are exactly the maximal cliques of the minimal triangulations of $G$.

\begin{definition}[Potential Maximal Clique (PMC)]
A set $\Omega \subseteq V(G)$ is a \emph{potential maximal clique} (PMC) of $G$ if $\Omega$ is an inclusion-wise maximal clique of some minimal triangulation of $G$.
\label{def:pmc}
\end{definition}

\begin{theorem}[\cite{BoTo02}\label{thm:pmc_charac}]
	Let $G$ be a graph and $\Omega \subseteq V(G)$. Set $\Omega$ is a PMC iff it satisfies the two following properties:
	\begin{itemize}
		\item $\Omega$ is a shy set.
		\item $\Omega$ has the reachability property.
	\end{itemize}
\end{theorem}

From now on, this latter characterization will be the main method used to identify PMCs. Thanks to it, we can check whether some set is a PMC or not in polynomial time.

\begin{lemma}[PMC identification \cite{BoTo02}]
	Given a set $\Omega \subseteq V(G)$, we can decide in $O(nm)$ whether $\Omega$ is a PMC of $G$ or not.
	\label{le:pmc-check}
\end{lemma}

In the literature, many properties of PMCs were identified. Most of them concern the case where $G$ is not a complete graph. Indeed, if $G$ is a clique, then it admits a unique PMC $V(G)$, as any other vertex set is not shy. For this reason, we focus in this paper only on graphs $G$ which are not complete, to skip this trivial but specific case. On other graphs, PMCs and minimal separators admit strong connections.

A first property of PMCs is that none of their proper subset is a PMC.

\begin{lemma}[\cite{BoTo02}]
	Any PMC $\Omega'$ of $G$ cannot be strictly included into another PMC $\Omega$ of $G$. 
\end{lemma}
This comes from the fact that, as PMCs have the reachability property, any proper subset admits at least one full component - the one which contains a removed vertex. Observe that this statement also holds for minimal $(a,b)$-separators with fixed $(a,b)$, but not for minimal separators in a broad sense. 

Second, an interesting fact is that the neighborhood of any connected component of $G \setminus \Omega$, for any PMC $\Omega$, is a minimal separator of $G$.

\begin{lemma}[\cite{BoTo02}]
	For any PMC $\Omega$ of $G$, the minimal separators of $G$ which are included into $\Omega$ are exactly the sets $N_G(C)$, for any connected component $C$ of $G \setminus \Omega$. Those sets are called \emph{neighbor separators}.
	\label{le:neighbor-sep}
\end{lemma}

Many efforts in the study of PMCs are dedicated to the structural features of these neighbor separators $S \subseteq \Omega$. We pursue with a property of sets $\Omega \setminus S$, which are PMCs without one of its neighbor separator.

\begin{lemma}[\cite{BoTo02}]\label{lem:PMC_in_one_FC_of_S}
	Let $\Omega$ be a PMC of $G$ and $S$ a minimal separator of $G$ contained in $\Omega$. The set $\Omega \setminus S$ is included into one connected component of $G \setminus S$ which is full for $S$.
\end{lemma}

One aspect of neighbor separators is crucial to achieve the analysis of both \textsc{BT-Enum}~\cite{BoTo02} and our algorithm \ALGO . It concerns the impact of \emph{triangulating} these separators $S \subsetneq \Omega$, {\em i.e.} adding all possible edges in $S \times S$.

\begin{definition}[Active separator \cite{BoTo02}]\label{def:active-sep}
Let $\Omega$ be a PMC of $G$ and $S \subseteq \Omega$ a minimal separator of $G$. We say $S$ is \textit{active} for $\Omega$ if $\Omega$ is not a clique after triangulating all neighbor separators $S' \subsetneq \Omega$ different from $S$, i.e. $S' \neq S$.
PMCs which admit at least one active separator are called \emph{active PMCs}.
\end{definition}

A surprising point is that the notion of active separator is in some way related to the property of being (or not) a PMC in subgraphs of $G$. We define $a$-\emph{nice PMCs}, reusing an expression from~\cite{FoVi08}.

\begin{definition}[$a$-nice PMC]
Given a graph $G$ and some vertex $a \in V(G)$, we say a PMC $\Omega$ of $G$ is $a$-\emph{nice} iff $a \notin \Omega$ and $\Omega$ is not a PMC of graph $G \setminus \{a\}$.
\label{def:a-nice-pmc}
\end{definition}

Authors of~\cite{BoTo02} established that if $\Omega$ is a PMC of $G$ but not of $G \setminus \{a\}$, then necessarily the singleton $\{a\}$ in $G$ is a one-vertex separator which blocks the reachability between at least two vertices of $\Omega$. Consequently, the separator $S = N(C_G(a,\Omega))$ is active: as $\{a\}$ is sufficient to separate some pair $u,v \in \Omega$ in $G[\overline{\Omega} \cup \{u,v\}]$, then for any connected component $C$ of $G \setminus \Omega$ different from $C_G(a,\Omega)$, we cannot have $\{u,v\} \subseteq N(C)$. So, triangulating all separators $N(C)$ with $C \neq C_G(a,\Omega)$ will not make $\Omega$ become a clique, since edge $(u,v)$ will miss. This is why $a$-nice PMCs are also active PMCs. This notion of $a$-nice PMCs is described in~\autoref{fig:a-nice}.


\begin{figure}[h]
\centering
\begin{subfigure}{0.33\textwidth}
\centering
\scalebox{0.8}{\input{tikz/a-nice}}
\caption{A PMC $\Omega$ (blue vertices) of $G$.}
\label{subfig:a-nice}
\end{subfigure}~
\begin{subfigure}{0.32\textwidth}
\centering
\scalebox{0.8}{\input{tikz/a-nice-2}}
\caption{$\Omega$ is not a PMC in $G \setminus \{a\}$.}
\label{subfig:a-nice-2}
\end{subfigure}
\begin{subfigure}{0.32\textwidth}
\centering
\scalebox{0.8}{\input{tikz/a-nice-3}}
\caption{But $\Omega$ is a PMC in $G \setminus \{b\}$.}
\label{subfig:a-nice-3}
\end{subfigure}
\caption{An example of PMC $\Omega$ which is $a$-nice but not $b$-nice. Indeed, $\Omega$ does not satisfy the reachability property in $G \setminus \{a\}$: there is a conflict pair (pink dashed line).}
\label{fig:a-nice}
\end{figure}

Now we can briefly present the algorithm \btenum\ proposed in~\cite{BoTo02} to enumerate all the PMCs in output-quadratic time. Lemmas following this paragraph will justify why \btenum\ actually finds all PMCs of the input graph. This algorithm uses recursivity, based upon the following key idea: knowing the PMCs of graph $G' = G \setminus \{a\}$ for some $a \in V(G)$, we are able to compute the PMCs of a graph $G$. To achieve this transfer, \btenum\ performs four successive enumeration steps: 
\begin{enumerate}
\item The PMCs of $G$ can be PMCs of $G'$: we consider the list of PMCs returned recursively for $G'$ and check whether they are PMCs of $G$ in polynomial time (\autoref{le:pmc-check}).
\item The PMCs of $G$ can be possibly written as $\Omega' \cup \{a\}$, where $\Omega'$ is a PMC of $G'$: we consider the list of PMCs returned recursively for $G'$ and check whether all $\Omega' \cup \{a\}$ are PMCs or not.
\item The PMCs of $G$ can also be written as $S \cup \{a\}$, where $S$ is a minimal separator of $G$: we check whether such sets are PMCs.
\item There can be PMCs of $G$ which are not recognized in the three first cases, but they can be expressed as $S \cup (T \setminus \{a\})$, where $S$ and $T$ are two minimal separators of $G$. Originally, these PMCs were expressed as $S \cup (C \cap T)$ in~\cite{BoTo02}, where $C$ is a full component for $S$. But Korhonen~\cite{Ko20} recently improved this characterization.
\end{enumerate}
In fact, the collection of PMCs found in the fourth step contains all $a$-nice PMCs. Furthermore, the authors showed that the three first steps of \btenum\ compute all the PMCs which are either PMCs of $G \setminus \{a\}$ (step 1) or containing $a$ (steps 2 and 3). In the remainder, our own enumeration algorithm \ALGO\ will utterly modify this fourth step to reduce the enumeration running time.

The following series of lemmas proves why \btenum\ enumerates all PMCs of $G$. These results will also be used for our contribution. \autoref{le:pmc_with_a} justifies that steps 2 and 3 are sufficient to enumerate the PMCs which contain $a$. \autoref{le:pmc_without_a} proves that the fourth step enumerates $a$-nice PMCs, as announced.

\begin{lemma}[\cite{BoTo02}]
	Let $\Omega'$ be a PMC of $G \setminus \{a\}$, for any $a \in V(G)$. Then either $\Omega' \cup \{a\}$ or $\Omega'$ is a PMC of $G$. It follows that $\vert \Pi_G \vert \ge \vert \Pi_{G \setminus \{a\}} \vert$.
	\label{le:more-pmcs}
\end{lemma}

\begin{lemma}[\cite{BoTo02}]\label{le:pmc_with_a}
	Let $\Omega$ be a PMC of $G$ such that $a \in \Omega$. Then $\Omega' = \Omega \setminus \{a\}$ is either a PMC of $G \setminus \{a\}$ or a minimal separator of $G$.
\end{lemma}

\begin{lemma}[\cite{BoTo02}, Theorem 2 of \cite{korhonen_finding_2022}]\label{le:pmc_without_a}
	Let $\Omega$ be a PMC of $G$ such that $a \notin \Omega$. Let $S = N(C_G(a,\Omega))$ be the minimal separator of $\Omega$ neighbor of $C_G(a,\Omega)$.

	If $\Omega$ is an $a$-nice PMC, then there is a minimal separator $T$ of $G$ such that $\Omega = S \cup T \setminus \{a\}$ with $a \in T$. Moreover, singleton $\{a\}$ is a $(x,y)$-separator in $G[\overline{\Omega} \cup \{x,y\}]$, for some $x,y \in S$.
\end{lemma}

In summary, \btenum\ proceeds as follows: it determines recursively the list of PMCs of some subgraph $G \setminus \{a\}$ of the input graph $G$, then it executes the four cases provided above and return the list of sets considered which were confirmed as PMCs. The base case of the recursion is the enumeration of PMCs in the single-vertex graph, which is trivial. The output-quadratic running time of \btenum\ comes from the fourth case, which consists in enumerating all pairs of minimal separators of $G$. The first three cases cost a time linear in the number of PMCs, in particular the third one which uses the polynomial-delay enumeration of minimal separators (\autoref{le:enum-sep}) and exploits the inequality $\vert \Delta_G \vert \le n \vert \Pi_G \vert$. This gives eventually a total running time of $O(n^2m\vert \Delta_G \vert^2)$. As recalled above, our best known lower bound for the number of PMCs is $\vert \Pi_G \vert \ge \frac{\vert \Delta_G \vert}{n}$. Thus \btenum\ guarantees output-quadratic time $O(n^3m\vert \Pi_G \vert^2)$.

\subsection{PMC enumeration and extension} \label{subsec:problems}

A common technique to obtain polynomial-delay enumeration algorithms is called \emph{Flashlight Search} \cite{ReTa75}. In this context, consider that we are tackling some enumeration problems where the objective is to list all inclusion-wise maximal sets of elements satisfying a certain property. Considering an arbitrary ordering $e_1,e_2,...,e_n$ of the elements that could belong to solutions, this method consists in creating a binary tree whose paths from root to the leaves provide us with the sets to be enumerated. Internal nodes of the binary tree represent a \emph{partial solution}, in other words, a subset of a solution to be enumerated. For each internal node at depth $i$, we check whether it is possible to \emph{extend} the current partial solution with element $e_i$. We potentially create two branches below: one leading to a node representing the partial solution extended with $e_i$, the other leading to the same partial solution without $e_i$. We stop extending a branch when no solution hides beneath. At each depth, to prune the tree, the \emph{extension version} of the enumeration problem needs to be answered. Below is the extension version of PMC enumeration.

\defdecproblem{\textsc{PMC Extension}}{A graph $G = (V,E)$, a subset $M \subseteq V$ of \textit{mandatory} vertices, a subset $F \subseteq V$ of \textit{forbidden} vertices such that $M \cap F = \emptyset$.}{Is there, a set $\Omega$ of vertices with $M \subseteq \Omega$, $\Omega \cap F = \emptyset$ and such that $\Omega$ is a PMC?}

Let us focus now on the possible application of Flashlight Search for PMC enumeration. The root of the tree is labeled $\emptyset$. Consider any ordering $v_1,\ldots,v_n$ of $V(G)$. Given a node of the tree, if the answer of the extension problem considering the add-on of $v_i$ (or resp. its exclusion) into the current mandatory (resp. forbidden) set is yes, then one can see that there will be at least one PMC satisfying our current restriction $(M, F)$, and it is worth continuing the exploration. It follows that if the answer is no, then there is no PMC below, and we can cut the branch. 

\begin{figure}[ht]
\centering
\scalebox{0.7}{\input{tikz/Flashlight_Search}}
\caption{Example of a Flashlight Search applied to $v_2 - v_1 - v_3$ to enumerate its two PMCs.}
\label{fig:flashlight_search}
\end{figure}

\autoref{fig:flashlight_search} provides an example of exploration using Flashlight Search, considering the very simple input graph: path $(v_2,v_1,v_3)$. Starting from the root of the binary tree, we ask the extension problem whether there is a PMC including $v_1$ ($M = \{v_1\}$ and $F = \emptyset$). If the answer is yes, we create a branch labeled $v_1$. We do the same for the exclusion of $v_1$. Deeper in the tree, at a depth $i$, we may arrive in a situation where we can still add a vertex $v_i$ to the mandatory (resp. forbidden) set, but the extension problem answers no. In this case, there will be no PMC containing (resp. excluding) $v_i$, and the branch is given up. This occurs on the right child of the root as there is no PMC excluding $v_1$ in our input graph.

Flashlight Search is very effective when there exists a polynomial time algorithm solving the extension problem. Then, a polynomial-delay polynomial-space algorithm for the enumeration problem comes out. Takata \cite{Ta10} used this framework in a smart way to enumerate all minimal $(a,b)$-separators of a graph, extending the component $C_G(a,S)$ instead of the separator $S$ itself.

Unfortunately, as for various extension problems, we show in the remainder that \textsc{PMC Extension} is NP-hard. Nevertheless, for our main contribution, we will tackle the problem of deciding whether a set contains a PMC. This problem, denoted as \textsc{PMC Superset}, is the subproblem of \textsc{PMC Extension} with $M = \emptyset$. Following \autoref{obs:shy-set}, if $X$ is not shy - {\em i.e.} admits a full component - then none of its subsets can be shy. Therefore, we can directly assume that if our input set $X$ is not shy, then the answer to the decision problem \textsc{PMC Superset} is necessarily no for $X$. The definition of \textsc{PMC Superset} follows.

\defdecproblem{\textsc{PMC Superset}}{A graph $G = (V,E)$ and a shy set $X \subseteq V$.}{Is there a PMC $\Omega \subseteq X$ of $G$?}


%% file: tikz/torso.tex
\begin{tikzpicture}

\node[noeud, fill = red] (a) at (3,2) {};
\node[noeud, fill = red] (b) at (5,2) {};
\node[noeud, fill = red] (c) at (5.5,4) {};
\node[noeud, fill = red] (d) at (2.5,4) {};

\node[noeud] (v) at (4,5) {};

\node[noeud] (u1) at (7,2.5) {};
\node[noeud] (u2) at (7,3.5) {};
\node[noeud] (u3) at (7.5,3) {};

\node[noeud] (w1) at (1,2) {};
\node[noeud] (w2) at (1,3) {};

\draw (a) node[below left,scale=1.5, color = red] {$u$};
\draw (b) node[below right,scale=1.5, color = red] {$v$};
\draw (c) node[above right,scale=1.5, color = red] {$x$};
\draw (d) node[above left,scale=1.5, color = red] {$y$};
\draw (u3) node[below,scale=1.5] {$a$};

\draw[rounded corners,line width=0.5mm, color = green] (6.6,2.2) rectangle (7.9,3.8);
\draw (9.1,3) node[scale=1.5, color = green] {$C_G(a,U)$};

\draw [line width = 1.4pt] (b) -- (a);

\draw [line width = 1.4pt] (c) -- (v);
\draw [line width = 1.4pt] (d) -- (v);
\draw [line width = 1.4pt] (b) -- (v);

\draw [line width = 1.4pt] (u1) -- (u2);
\draw [line width = 1.4pt] (u2) -- (u3);
\draw [line width = 1.4pt] (u1) -- (u3);
\draw [line width = 1.4pt] (u1) -- (b);
\draw [line width = 1.4pt] (u2) -- (c);

\draw [line width = 1.4pt] (d) -- (w1);
\draw [line width = 1.4pt] (d) -- (w2);
\draw [line width = 1.4pt] (a) -- (w1);
\draw [line width = 1.4pt] (w1) -- (w2);
\end{tikzpicture}

%% file: tikz/torso2.tex
\begin{tikzpicture}

\node[noeud, fill = red] (a) at (3,2) {};
\node[noeud, fill = red] (b) at (5,2) {};
\node[noeud, fill = red] (c) at (5.5,4) {};
\node[noeud, fill = red] (d) at (2.5,4) {};

\draw (a) node[below left,scale=1.5, color = red] {$u$};
\draw (b) node[below right,scale=1.5, color = red] {$v$};
\draw (c) node[above right,scale=1.5, color = red] {$x$};
\draw (d) node[above left,scale=1.5, color = red] {$y$};

\draw [line width = 1.5pt, color = red] (b) -- (a);
\draw [line width = 1.5pt, color = red] (b) -- (c);
\draw [line width = 1.5pt, color = red] (c) -- (d);
\draw [line width = 1.5pt, color = red] (d) -- (b);
\draw [line width = 1.5pt, color = red] (a) -- (d);
\end{tikzpicture}

%% file: tikz/a_b_close_sep.tex
\begin{tikzpicture}
    \node[noeud] (a) at (0,0) {};
    \node[noeud] (1) at (0,-1) {};
    \node[noeud] (2) at (-1,0) {};
    \node[noeud] (3) at (-1,-1) {};

    \draw (a) node[above left,scale=1.2] {$a$};

    \node[noeud] (4) at (-2,-0.5) {};
    \node[noeud] (5) at (-3,-0.5) {};

    \node[noeud] (6) at (1,0) {};
    \node[noeud] (7) at (1,-1) {};

    \node[noeud] (8) at (2,-0.5) {};
    \node[noeud] (9) at (3,-0.5) {};
    \node[noeud] (b) at (4,-0.5) {};

    \draw (b) node[below right,scale=1.2] {$b$};

    \draw[rounded corners,line width=0.8mm, color = green] (-1.3,-1.3) rectangle (0.3,0.6);
    \draw (-0.5,1) node[scale=1.2, color = green] {$A$};

    \draw[rounded corners,line width=0.8mm, color = red] (-2.3,-0.8) rectangle (-1.7,-0.2);
    
    \draw[rounded corners,line width=0.8mm, color = red] (0.725,-1.275) rectangle (1.275,0.275);
    
    \draw[rounded corners,line width=0.8mm, color = green] (1.7,-1.2) rectangle (4.6,0.3);
    \draw (3,0.8) node[scale=1.2, color = green] {$C(b,N(A))$};

    \draw[rounded corners,line width=1mm, color = blue] (0.6,-1.4) rectangle (1.4,0.4);
    \draw (1.5,-2) node[scale=1.2, color = blue] {$\sepa{A}{b} \subseteq N(A)$};
    
    \draw (a) -- (1);
    \draw (a) -- (2);
    \draw (a) -- (6);
    \draw (1) -- (3);
    \draw (1) -- (7);
    \draw (2) -- (3);
    \draw (2) -- (4);
    \draw (3) -- (4);
    \draw (4) -- (5);
    \draw (6) -- (7);
    \draw (6) -- (8);
    \draw (7) -- (8);
    \draw (8) -- (9);
    \draw (9) -- (b);
    \draw (1) -- (2);
    \draw (3) -- (6);
    \draw[out = -10, in = 220] (7) to (9);

\end{tikzpicture}

%% file: tikz/a-nice.tex
\begin{tikzpicture}

\node[noeud] (a) at (1.8,4) {};
\node[noeud] (b) at (4.2,4) {};
\node[noeud] (c) at (3,1.8) {};

\node[noeud, fill = blue] (x) at (3,3) {};
\node[noeud, fill = blue] (r) at (3,4.5) {};
\node[noeud, fill = blue] (s) at (2,2.2) {};
\node[noeud, fill = blue] (t) at (4,2.2) {};

%
\draw (a) node[above left,scale=1.2] {$a$};
\draw (x) node[above left,scale=1.2, color = blue] {$\Omega$};
\draw (b) node[above right,scale=1.2] {$b$};
%

\draw [line width = 1.4pt] (x) -- (r);
\draw [line width = 1.4pt] (x) -- (s);
\draw [line width = 1.4pt] (x) -- (t);
\draw [line width = 1.4pt] (r) -- (t);

\draw [line width = 1.4pt] (a) -- (r);
\draw [line width = 1.4pt] (a) -- (s);

\draw [line width = 1.4pt] (b) -- (r);
\draw [line width = 1.4pt] (b) -- (t);

\draw [line width = 1.4pt] (c) -- (s);
\draw [line width = 1.4pt] (c) -- (t);
%
\end{tikzpicture}

%% file: tikz/a-nice-2.tex
\begin{tikzpicture}

\node[noeud] (b) at (4.2,4) {};
\node[noeud] (c) at (3,1.8) {};

\node[noeud, fill = blue] (x) at (3,3) {};
\node[noeud, fill = blue] (r) at (3,4.5) {};
\node[noeud, fill = blue] (s) at (2,2.2) {};
\node[noeud, fill = blue] (t) at (4,2.2) {};

%
\draw (x) node[above left,scale=1.2, color = blue] {$\Omega$};
\draw (b) node[above right,scale=1.2] {$b$};
%

\draw [out = 190, in = 110, line width = 1.8pt, dashed, color = white!60!red] (r) to (s);

\draw [line width = 1.4pt] (x) -- (r);
\draw [line width = 1.4pt] (x) -- (s);
\draw [line width = 1.4pt] (x) -- (t);
\draw [line width = 1.4pt] (r) -- (t);


\draw [line width = 1.4pt] (b) -- (r);
\draw [line width = 1.4pt] (b) -- (t);

\draw [line width = 1.4pt] (c) -- (s);
\draw [line width = 1.4pt] (c) -- (t);
%
\end{tikzpicture}

%% file: tikz/a-nice-3.tex
\begin{tikzpicture}

\node[noeud] (a) at (1.8,4) {};
\node[noeud] (c) at (3,1.8) {};

\node[noeud, fill = blue] (x) at (3,3) {};
\node[noeud, fill = blue] (r) at (3,4.5) {};
\node[noeud, fill = blue] (s) at (2,2.2) {};
\node[noeud, fill = blue] (t) at (4,2.2) {};

%
\draw (a) node[above left,scale=1.2] {$a$};
\draw (x) node[above left,scale=1.2, color = blue] {$\Omega$};
%

\draw [line width = 1.4pt] (x) -- (r);
\draw [line width = 1.4pt] (x) -- (s);
\draw [line width = 1.4pt] (x) -- (t);
\draw [line width = 1.4pt] (r) -- (t);

\draw [line width = 1.4pt] (a) -- (r);
\draw [line width = 1.4pt] (a) -- (s);


\draw [line width = 1.4pt] (c) -- (s);
\draw [line width = 1.4pt] (c) -- (t);
%
\end{tikzpicture}

%% file: tikz/Flashlight_Search.tex
\begin{tikzpicture}
\coordinate (root) at (0,0);
\coordinate (y1) at (-6,-2);
\coordinate (y1y2) at (-8,-4);
\coordinate (y1y2y3) at (-10,-6);
\coordinate (y1y2n3) at (-6,-6); 

\coordinate (y1n2) at (-2,-4);
\coordinate (y1n2y3) at (-4,-6); 
\coordinate (y1n2n3) at (0,-6); 

\coordinate (n1) at (4,-2);

\node [draw, circle,  minimum size=0.2cm, ultra thick, fill=white](root) at (root) {$\emptyset$};

\node [draw, circle,  minimum height=0.2cm, minimum width=0.2cm, ultra thick, fill=white] (y1) at (y1) {};
\node [draw, circle,  minimum height=0.2cm, minimum width=0.2cm, ultra thick, fill=white] (y1y2) at (y1y2) {};
\node [draw, cross out, red,  minimum height=0.2cm, minimum width=0.2cm, ultra thick, fill=white, scale=1.5] (y1y2y3) at (y1y2y3) {};
\node [draw, circle, double, green!60!black,  minimum height=0.2cm, minimum width=0.2cm, ultra thick, fill=white] (y1y2n3) at (y1y2n3) {};

\node [draw, circle,  minimum height=0.2cm, minimum width=0.2cm, ultra thick, fill=white] (y1n2) at (y1n2) {};
\node [draw, circle, double, green!60!black,  minimum height=0.2cm, minimum width=0.2cm, ultra thick, fill=white] (y1n2y3) at (y1n2y3) {};
\node [draw, cross out, red, minimum height=0.2cm, minimum width=0.2cm, ultra thick, fill=white, scale=1.5] (y1n2n3) at (y1n2n3) {};

\node [draw, cross out, red, minimum height=0.2cm, minimum width=0.2cm, ultra thick, fill=white, scale=1.5] (n1) at (n1) {};

\draw [line width=2pt] (root) -- (y1) node[midway, sloped, above, scale=2] {$v_1$};

\draw [line width=2pt] (y1) -- (y1y2) node[midway, above, scale=2, sloped] {$v_2$};
\draw [line width=2pt, , dashed] (y1y2) -- (y1y2y3) node[midway, above, scale=2, sloped] {$v_3$};
\draw [line width=2pt] (y1y2) -- (y1y2n3) node[midway, above, scale=2, sloped] {$\overline{v_3}$};

\draw [line width=2pt] (y1) -- (y1n2) node[midway, above, scale=2, sloped] {$\overline{v_2}$};
\draw [line width=2pt] (y1n2) -- (y1n2y3) node[midway, above, scale=2, sloped] {$v_3$};
\draw [line width=2pt, dashed] (y1n2) -- (y1n2n3) node[midway, above, scale=2, sloped] {$\overline{v_3}$};

\draw [line width=2pt, dashed] (root) -- (n1) node[midway, above, scale=2, sloped] {$\overline{v_1}$};

\node [scale=1.25, red!60!black,align=center] at (-10,-7) {No solution\\with $\{v_1,v_2,v_3\}$};
\node [scale=1.25, red!60!black,align=center] at (4,-3) {No solution\\without $\{v_1\}$};
\node [scale=1.25, red!60!black,align=center] at (0,-7) {No solution\\with $\{v_1\}$\\without $\{v_2,v_3\}$};
\node [scale=1.25, green!60!black] at (-6,-7) {$\{v_1,v_2\}$};
\node [scale=1.25, green!60!black] at (-4,-7) {$\{v_1,v_3\}$};

\end{tikzpicture}

%% file: tikz/negative-shy-set.tex
\begin{tikzpicture}

\node[noeud, fill = red] (a) at (1,5) {};
\node[noeud, fill = red] (b) at (5,5) {};

\node[noeud] (xa) at (1.5,3) {};
\node[noeud] (xb) at (4.5,3) {};

\node[noeud, fill = red] (x1) at (3,6) {};
\node[noeud, fill = red] (x2) at (3,5) {};
\node[noeud, fill = red] (x3) at (3,4) {};
\node[noeud, fill = red] (x4) at (3,3) {};

\draw (xa) node[below left,scale=1.5] {$a'$};
\draw (xb) node[below right,scale=1.5] {$b'$};
\draw (a) node[above left,scale=1.5, color = red] {$a$};
\draw (b) node[above right,scale=1.5, color = red] {$b$};

\draw[rounded corners,line width=0.5mm, color = red] (2.6,2.5) rectangle (3.4,6.5);
\draw (3.8,6.5) node[scale=1.5, color = red] {$S$};
\draw (3,1.5) node[scale=1.5, color = red] {$X = S \cup \{a,b\}$};

\draw [line width = 1.4pt] (x1) -- (a);
\draw [line width = 1.4pt] (x2) -- (a);
\draw [line width = 1.4pt] (x3) -- (a);
\draw [line width = 1.4pt] (x4) -- (a);

\draw [line width = 1.4pt] (x1) -- (xa);
\draw [line width = 1.4pt] (x2) -- (xa);
\draw [line width = 1.4pt] (x3) -- (xa);
\draw [line width = 1.4pt] (x4) -- (xa);
\draw [line width = 1.4pt] (a) -- (xa);

\draw [line width = 1.4pt] (x1) -- (b);
\draw [line width = 1.4pt] (x2) -- (b);
\draw [line width = 1.4pt] (x3) -- (b);
\draw [line width = 1.4pt] (x4) -- (b);

\draw [line width = 1.4pt] (x1) -- (xb);
\draw [line width = 1.4pt] (x2) -- (xb);
\draw [line width = 1.4pt] (x3) -- (xb);
\draw [line width = 1.4pt] (x4) -- (xb);
\draw [line width = 1.4pt] (b) -- (xb);
\end{tikzpicture}

%% file: tikz/positive-shy-set.tex
\begin{tikzpicture}

\node[noeud, fill = red] (a) at (1,5) {};
\node[noeud, fill = red] (b) at (5,5) {};

\node[noeud] (xa) at (1.5,3) {};
\node[noeud] (xb) at (4.5,3) {};

\node[noeud, fill = red] (x1) at (3,6) {};
\node[noeud, fill = red] (x2) at (3,5) {};
\node[noeud, fill = red] (x3) at (3,4) {};
\node[noeud, fill = red] (x4) at (3,3) {};

\draw (xa) node[below left,scale=1.5] {$a'$};
\draw (xb) node[below right,scale=1.5] {$b'$};
\draw (a) node[above left,scale=1.5, color = red] {$a$};
\draw (b) node[above right,scale=1.5, color = red] {$b$};

\draw[rounded corners,line width=0.5mm, color = red] (2.6,2.5) rectangle (3.4,6.5);
\draw (3.8,6.5) node[scale=1.5, color = red] {$S$};
\draw (3,1.5) node[scale=1.5, color = red] {$X = S \cup \{a,b\}$};

\draw [line width = 1.4pt] (x1) -- (a);
\draw [line width = 1.4pt] (x2) -- (a);
\draw [line width = 1.4pt] (x3) -- (a);
\draw [line width = 1.4pt] (x4) -- (a);

\draw [line width = 1.4pt] (x2) -- (xa);
\draw [line width = 1.4pt] (x3) -- (xa);
\draw [line width = 1.4pt] (x4) -- (xa);
\draw [line width = 1.4pt] (a) -- (xa);

\draw [line width = 1.4pt] (x1) -- (b);
\draw [line width = 1.4pt] (x2) -- (b);
\draw [line width = 1.4pt] (x3) -- (b);
\draw [line width = 1.4pt] (x4) -- (b);

\draw [line width = 1.4pt] (x1) -- (xb);
\draw [line width = 1.4pt] (x2) -- (xb);
\draw [line width = 1.4pt] (x3) -- (xb);
\draw [line width = 1.4pt] (x4) -- (xb);
\draw [line width = 1.4pt] (b) -- (xb);
\end{tikzpicture}

%% file: tikz/removal.tex
\begin{tikzpicture}

\draw[rounded corners,line width=0.5mm, color = red, fill = white!80!red] (1.5,1.5) rectangle (4.5,5.5);
\draw (3,6) node[scale=1.2, color = red] {$U=X$};

\node[noeud, fill = red] (x3) at (2.5,2) {};
\node[noeud, fill = red] (x4) at (3.5,2) {};
\node[noeud, fill = red] (x2) at (2,3.5) {};
\node[noeud, fill = red] (x5) at (4,3.5) {};
\node[noeud, fill = red] (x1) at (2.5,5) {};
\node[noeud, fill = red] (x6) at (3.5,5) {};

\node[noeud] (c11) at (1,2.8) {};
\node[noeud](c12) at (1,2.3) {};

\node[noeud] (c2) at (5,1) {};

\node[noeud] (c3) at (5.5,3.7) {};

\draw (x1) node[above left] {$x_1$};
\draw (x6) node[above right] {$x_6$};
\draw (x2) node[above right] {$x_2$};
\draw (x3) node[below left] {$x_3$};
\draw (x4) node[below left] {$x_4$};
\draw (x5) node[above left] {$x_5$};

\draw [line width = 1.4pt] (x1) -- (x2);
\draw [line width = 1.4pt] (x1) -- (x6);
\draw [line width = 1.4pt] (x2) -- (x4);
\draw [line width = 1.4pt] (x2) -- (x5);

\draw [line width = 1.4pt] (x2) -- (c11);
\draw [line width = 1.4pt] (x2) -- (c12);
\draw [line width = 1.4pt] (x3) -- (c12);
\draw [line width = 1.4pt] (c11) -- (c12);

\draw [line width = 1.4pt] (c3) -- (x6);
\draw [line width = 1.4pt] (c3) -- (x5);
\draw [line width = 1.4pt] (c3) -- (x4);

\draw [line width = 1.4pt] (c2) -- (x5);
\draw [line width = 1.4pt] (c2) -- (x4);
\draw [out=180,in=300,line width = 1.4pt] (c2) to (x3);
\end{tikzpicture}

%% file: tikz/removal2.tex
\begin{tikzpicture}

\draw[rounded corners,line width=0.5mm, color = red, fill = white!80!red] (1.5,2.3) -- (1.5,5.5) -- (4.5,5.5) -- (4.5,1.5) -- (2.9,1.5) -- (2.9,2.3) -- (1.5,2.3);
\draw (3,6) node[scale=1.2, color = red] {$U = X \setminus \{x_3\}$};

\node[noeud, fill = white!60!blue] (x3) at (2.5,2) {};
\node[noeud, fill = red] (x4) at (3.5,2) {};
\node[noeud, fill = red] (x2) at (2,3.5) {};
\node[noeud, fill = red] (x5) at (4,3.5) {};
\node[noeud, fill = red] (x1) at (2.5,5) {};
\node[noeud, fill = red] (x6) at (3.5,5) {};

\node[noeud, fill = white!60!blue] (c11) at (1,2.8) {};
\node[noeud, fill = white!60!blue](c12) at (1,2.3) {};

\node[noeud, fill = white!60!blue] (c2) at (5,1) {};

\node[noeud] (c3) at (5.5,3.7) {};

\draw (x1) node[above left] {$x_1$};
\draw (x6) node[above right] {$x_6$};
\draw (x2) node[above right] {$x_2$};
\draw (x4) node[below left] {$x_4$};
\draw (x5) node[above left] {$x_5$};
\draw (2,1.7) node[scale=1.2, color = white!30!blue] {$C_0$};
\draw (c3) node[below right,scale=0.8] {native};

\draw [line width = 1.4pt] (x1) -- (x2);
\draw [line width = 1.4pt] (x1) -- (x6);
\draw [line width = 1.4pt] (x2) -- (x4);
\draw [line width = 1.4pt] (x2) -- (x5);

\draw [line width = 1.4pt] (x2) -- (c11);
\draw [line width = 1.4pt] (x2) -- (c12);
\draw [line width = 1.4pt] (x3) -- (c12);
\draw [line width = 1.4pt] (c11) -- (c12);

\draw [line width = 1.4pt] (c3) -- (x6);
\draw [line width = 1.4pt] (c3) -- (x5);
\draw [line width = 1.4pt] (c3) -- (x4);

\draw [line width = 1.4pt] (c2) -- (x5);
\draw [line width = 1.4pt] (c2) -- (x4);
\draw [out=180,in=300,line width = 1.4pt] (c2) to (x3);
\end{tikzpicture}

%% file: tikz/removal3.tex
\begin{tikzpicture}

\draw[rounded corners,line width=0.5mm, color = red, fill = white!80!red] (1.5,3.1) rectangle (4.5,5.5);
\draw (3,6) node[scale=1.2, color = red] {$U = X \setminus \{x_3,x_4\}$};

\node[noeud, fill = blue] (x3) at (2.5,2) {};
\node[noeud, fill = blue] (x4) at (3.5,2) {};
\node[noeud, fill = red] (x2) at (2,3.5) {};
\node[noeud, fill = red] (x5) at (4,3.5) {};
\node[noeud, fill = red] (x1) at (2.5,5) {};
\node[noeud, fill = red] (x6) at (3.5,5) {};

\node[noeud, fill = blue] (c11) at (1,2.8) {};
\node[noeud, fill = blue](c12) at (1,2.3) {};

\node[noeud, fill = blue] (c2) at (5,1) {};

\node[noeud, fill = blue] (c3) at (5.5,3.7) {};

\draw (x1) node[above left] {$x_1$};
\draw (x6) node[above right] {$x_6$};
\draw (x2) node[above right] {$x_2$};
\draw (x5) node[above left] {$x_5$};
\draw (2,1.7) node[scale = 1.2, color = blue] {$C_1$};

\draw [line width = 1.4pt] (x1) -- (x2);
\draw [line width = 1.4pt] (x1) -- (x6);
\draw [line width = 1.4pt] (x2) -- (x4);
\draw [line width = 1.4pt] (x2) -- (x5);

\draw [line width = 1.4pt] (x2) -- (c11);
\draw [line width = 1.4pt] (x2) -- (c12);
\draw [line width = 1.4pt] (x3) -- (c12);
\draw [line width = 1.4pt] (c11) -- (c12);

\draw [line width = 1.4pt] (c3) -- (x6);
\draw [line width = 1.4pt] (c3) -- (x5);
\draw [line width = 1.4pt] (c3) -- (x4);

\draw [line width = 1.4pt] (c2) -- (x5);
\draw [line width = 1.4pt] (c2) -- (x4);
\draw [out=180,in=300,line width = 1.4pt] (c2) to (x3);
\end{tikzpicture}

%% file: tikz/shy-sep.tex
\begin{tikzpicture}

\node[noeud] (x) at (5,5.5) {};
\node[noeud] (y) at (5,4) {};
\node[noeud] (z) at (5,2.5) {};
\node[noeud] (w) at (5,1) {};

\node[noeud] (a) at (2.5,4) {};
\node[noeud] (b) at (2.5,2.2) {};

\node[noeud] (t1) at (7,3) {};
\node[noeud] (t2) at (9,3) {};
\node[noeud] (t3) at (11,3) {};

\node (c) at (7.6,4.75) {};
\node[noeud] (c1) at (7.6,5) {};
\node[noeud] (c2) at (7.2,4.5) {};
\node[noeud] (c3) at (8,4.5) {};

\node[noeud] (d2) at (9.2,6.5) {};

\node[noeud] (u1) at (7,0.5) {};
\node[noeud] (u2) at (8,0.5) {};
\node[noeud] (u3) at (9,0.5) {};


\draw (t1) node[below right,scale=1.5] {$u$};
\draw (t3) node[below right,scale=1.5] {$v$};

\draw[rounded corners,line width=1mm, color = red] (4.2,6.3) rectangle (5.8,0.2);
\draw (5,6.8) node[scale=1.5, color = red] {$S$};
\draw[rounded corners,line width=1mm, color = yellow] (4.4,6.1) rectangle (5.6,0.4);

\draw[rounded corners,line width=1mm, color = red] (6.4,2.2) rectangle (11.8,3.7);
\draw (12.3,3) node[scale=1.5, color = red] {$T$};
\draw[rounded corners,line width=1mm, color = yellow] (6.6,2.4) rectangle (9.7,3.5);
\draw (10.1,2.6) node[scale=1.5, color = yellow!80!black] {$S'$};

\draw[rounded corners,line width=0.3mm, color = white!50!green] (6.1,-0.1) rectangle (12.8,7.6);
\draw (12.2,0.2) node[scale=1.2, color = green] {$C$};

\draw[line width=0.3mm, color = blue, dashed] (c) circle (0.75);
\draw (8.4,5.35) node[scale=1.2, color = blue] {$B$};

\draw [line width = 1.4pt] (x) -- (a);
\draw [line width = 1.4pt] (y) -- (a);
\draw [line width = 1.4pt] (z) -- (a);
\draw [line width = 1.4pt] (b) -- (a);
\draw [line width = 1.4pt] (b) -- (w);

\draw [line width = 1.4pt] (t2) -- (t3);

\draw [line width = 1.4pt] (c1) -- (c2);
\draw [line width = 1.4pt] (c2) -- (c3);
\draw [line width = 1.4pt] (c1) -- (c3);

\draw [line width = 1.4pt] (x) -- (c1);
\draw [line width = 1.4pt] (y) -- (c1);
\draw [out = 200, in = 70,line width = 1.4pt] (c2) to (z);
\draw [out = 210, in = 70, line width = 1.4pt] (c2) to (w);
\draw [line width = 1.4pt] (t1) -- (c2);
\draw [line width = 1.4pt] (t2) -- (c3);

\draw [out=50,in=160,line width = 1.4pt] (x) to (d2);
\draw [out=50,in=185,line width = 1.4pt] (y) to (d2);
\draw [line width = 1.4pt] (d2) -- (t3);

\draw [line width = 1.4pt] (u1) -- (u2);
\draw [line width = 1.4pt] (u2) -- (u3);

\draw [line width = 1.4pt] (z) -- (u1);
\draw [line width = 1.4pt] (w) -- (u1);
\draw [line width = 1.4pt] (t1) -- (u1);
\draw [line width = 1.4pt] (t3) -- (u3);

\end{tikzpicture}

%% file: tikz/shy-sep-2.tex
\begin{tikzpicture}

\node[noeud] (x) at (5,5.5) {};
\node[noeud] (y) at (5,4) {};
\node[noeud] (z) at (5,2.5) {};
\node[noeud] (w) at (5,1) {};

\node[noeud] (a) at (2.5,4) {};
\node[noeud] (b) at (2.5,2.2) {};

\node[noeud] (t1) at (7,3) {};
\node[noeud] (t2) at (9,3) {};
\node[noeud] (t3) at (11,3) {};

\node (c) at (7.6,4.75) {};
\node[noeud] (c1) at (7.6,5) {};
\node[noeud] (c2) at (7.2,4.5) {};
\node[noeud] (c3) at (8,4.5) {};

\node[noeud] (d2) at (9.2,6.5) {};

\node[noeud] (u1) at (7,0.5) {};
\node[noeud] (u2) at (8,0.5) {};
\node[noeud] (u3) at (9,0.5) {};


\draw (t1) node[below right,scale=1.5] {$u$};
\draw (t3) node[below right,scale=1.5] {$v$};

\draw[rounded corners,line width=1mm, color = red] (4.2,6.3) -- (5.8,6.3) -- (5.8,3.5) -- (9.7,3.5) -- (9.7,2.4) -- (5.8,2.4) -- (5.8,0.2) -- (4.2,0.2) -- (4.2,6.3);
\draw (5,6.8) node[scale=1.5, color = red] {$S'$};

\draw[rounded corners,line width=1mm, color = red] (10.3,2.4) rectangle (11.8,3.5);
\draw (12.3,3) node[scale=1.5, color = red] {$T'$};

\draw[rounded corners,line width=0.3mm, color = white!50!green] (6.1,-0.1) -- (6.1,1.0) -- (10.0,1.0) -- (10.0,4.0) -- (8.8,4.0) -- (8.8,7.6) -- (12.8,7.6) -- (12.8,-0.1) -- (6.1,-0.1);
\draw (12.2,0.2) node[scale=1.2, color = green] {$C'$};

\draw[line width=0.3mm, color = blue, dashed] (c) circle (0.75);
\draw (8.4,5.35) node[scale=1.2, color = blue] {$B$};

\draw [line width = 1.4pt] (x) -- (a);
\draw [line width = 1.4pt] (y) -- (a);
\draw [line width = 1.4pt] (z) -- (a);
\draw [line width = 1.4pt] (b) -- (a);
\draw [line width = 1.4pt] (b) -- (w);

\draw [line width = 1.4pt] (t2) -- (t3);

\draw [line width = 1.4pt] (c1) -- (c2);
\draw [line width = 1.4pt] (c2) -- (c3);
\draw [line width = 1.4pt] (c1) -- (c3);

\draw [line width = 1.4pt] (x) -- (c1);
\draw [line width = 1.4pt] (y) -- (c1);
\draw [out = 200, in = 70,line width = 1.4pt] (c2) to (z);
\draw [out = 210, in = 70, line width = 1.4pt] (c2) to (w);
\draw [line width = 1.4pt] (t1) -- (c2);
\draw [line width = 1.4pt] (t2) -- (c3);

\draw [out=50,in=160,line width = 1.4pt] (x) to (d2);
\draw [out=50,in=185,line width = 1.4pt] (y) to (d2);
\draw [line width = 1.4pt] (d2) -- (t3);

\draw [line width = 1.4pt] (u1) -- (u2);
\draw [line width = 1.4pt] (u2) -- (u3);

\draw [line width = 1.4pt] (z) -- (u1);
\draw [line width = 1.4pt] (w) -- (u1);
\draw [line width = 1.4pt] (t1) -- (u1);
\draw [line width = 1.4pt] (t3) -- (u3);

\end{tikzpicture}

%% file: tikz/pmc-super-shy.tex
\begin{tikzpicture}

\draw[rounded corners,line width=1mm, color = purple, fill = white!90!purple] (2,5) -- (2,4) -- (4,-0.5) -- (7.2,-0.5) -- (8.8,1.9) -- (12.8,1.9) -- (12.8,3.9) -- (6.8,3.9) -- (6.8,7.3) -- (4.1,7.3) -- (2,5);
\draw (3.0,6.7) node[scale=1.5, color = purple] {$X$};

\node[noeud] (x) at (5,5.5) {};
\node[noeud, line width = 1mm, color = orange,fill = white!70!orange] (y) at (5,4) {};
\node[noeud] (z) at (5,2.5) {};
\node[noeud] (w) at (5,1) {};

\node[noeud] (a) at (2.5,4) {};
\node[noeud] (b) at (2,2.0) {};

\node[noeud] (t1) at (7,3) {};
\node[noeud, line width = 1mm, color = orange,fill = white!70!orange] (t2) at (9,3) {};
\node[noeud, fill = blue] (t3) at (11,3) {};

\node (c) at (8,4.75) {};
\node[noeud] (c1) at (8,5) {};
\node[noeud] (c2) at (7.6,4.4) {};
\node[noeud] (c3) at (8.4,4.4) {};
\draw[line width=1mm, color = purple, fill = white!90!purple] (c1) circle (0.45);
\node[noeud] (c1) at (8,5) {};

\node (d) at (9.85,6.75) {};
\node[noeud] (d1) at (10.5,7) {};
\node[noeud, fill = blue] (d2) at (9.2,6.5) {};
\draw[line width=1mm, color = purple, fill = white!90!purple] (d1) circle (0.5);
\node[noeud, line width = 1mm, color = orange,fill = white!70!orange] (d1) at (10.5,7) {};

\node[noeud, line width = 1mm, color = orange,fill = white!70!orange] (u1) at (7,0) {};
\node[noeud, fill = blue] (u2) at (8,0) {};
\node[noeud, fill = blue] (u3) at (9,0) {};

\node[noeud] (v1) at (7.5,1.3) {};

\draw (b) node[above left,scale=1.5] {$b$};
\draw (t3) node[below right,scale=1.5, color = blue] {$a$};

\draw[rounded corners,line width=0.5mm, color = red] (4.2,6.3) rectangle (5.8,0.2);
\draw (5,6.8) node[scale=1.5, color = red] {$S'$};

\draw[rounded corners,line width=0.5mm, color = red] (6.4,2.2) rectangle (11.8,3.7);
\draw (12.3,3) node[scale=1.5, color = red] {$T'$};

\draw[rounded corners,line width=0.3mm, color = white!50!green] (6.1,-0.8) rectangle (13.0,8.1);
\draw (11.7,-0.5) node[scale=1.2, color = green] {$C'$};
%
%
\draw[line width=0.3mm, color = blue, dashed] (7.2,4.1) rectangle (8.8,5.7);
\draw[line width=0.3mm, color = blue, dashed] (v1) circle (0.3);
\draw[rounded corners,line width=0.3mm, color = blue, dashed] (6.7,-0.3) rectangle (9.3,0.3);
\draw[line width=0.3mm, color = blue, dashed] (d) circle (1.3);

\draw (8.5,3) node[scale=1.5, color = orange] {$S$};

\draw [line width = 1.4pt] (x) -- (a);
\draw [line width = 1.4pt] (y) -- (a);
\draw [line width = 1.4pt] (z) -- (a);
\draw [line width = 1.4pt] (b) -- (a);
\draw [line width = 1.4pt] (b) -- (w);

\draw [out=60,in=300,line width = 1.4pt] (w) to (y);
\draw [out=-30,in=-120,line width = 1.4pt] (z) to (t2);

\draw [line width = 1.4pt] (t2) -- (t3);

\draw [line width = 1.4pt] (c1) -- (c2);
\draw [line width = 1.4pt] (c2) -- (c3);
\draw [line width = 1.4pt] (c1) -- (c3);

\draw [line width = 1.4pt] (x) -- (c1);
\draw [line width = 1.4pt] (y) -- (c2);
\draw [line width = 1.4pt] (t1) -- (c2);
\draw [line width = 1.4pt] (t2) -- (c3);

\draw [line width = 1.4pt] (d1) -- (d2);

\draw [out=50,in=160,line width = 1.4pt] (x) to (d1);
\draw [out=50,in=185,line width = 1.4pt] (y) to (d2);
\draw [line width = 1.4pt] (d2) -- (t3);

\draw [line width = 1.4pt] (u1) -- (u2);
\draw [line width = 1.4pt] (u2) -- (u3);

\draw [line width = 1.4pt] (z) -- (u1);
\draw [line width = 1.4pt] (w) -- (u1);
\draw [line width = 1.4pt] (t1) -- (u1);
\draw [line width = 1.4pt] (t3) -- (u3);

\draw [out=0,in=200,line width = 1.4pt] (w) to (v1);
\draw [out=30,in=-110,line width = 1.4pt] (v1) to (t2);
\end{tikzpicture}

%% file: tikz/neighbor-yes.tex
\begin{tikzpicture}

\draw[rounded corners,line width=1mm, color = white!80!orange, fill = white!80!orange] (4.4,3.0) rectangle (5.4,0.5);

\draw[rounded corners,line width=1mm, color = white!80!orange, fill = white!80!orange] (6.4,2.4) rectangle (10.6,3.4);
\draw[rounded corners,line width=0.5mm, color = blue] (6.3,2.3) rectangle (10.7,3.5);

\draw[rounded corners,line width=1mm, color = white!80!orange, fill = white!80!orange] (2.6,3.6) rectangle (3.4,4.4);
\draw[rounded corners,line width=0.5mm, color = blue] (2.5,3.5) rectangle (3.5,4.5);

\draw (3.9,0.2) node[scale=1.5, color =  white!60!orange] {$\Omega'$};

\draw[rounded corners,line width=0.5mm, color = blue] (4.3,0.4) rectangle (5.5,1.5);
\draw[rounded corners,line width=0.5mm, color = blue] (6.6,-0.4) rectangle (7.4,0.4);

\draw (8.7,3.9) node[scale=1.3, color =  blue] {$\Omega$};

\node[noeud] (x) at (5,5.5) {};
\node[noeud] (y) at (5,4) {};
\node[noeud, line width = 1mm, color = orange, fill = black] (z) at (5,2.5) {};
\node[noeud] (w) at (5,1) {};

\node[noeud, line width = 1mm, color = orange, fill = white!60!green] (a) at (3,4) {};
\node[noeud] (b) at (3,2.0) {};

\node[noeud, color = black, fill = black] (t1) at (7,3) {};
\node[noeud, line width = 1mm, color = orange, fill = white!60!green] (t2) at (8.5,3) {};
\node[noeud] (t3) at (10,3) {};

%
\node[noeud] (d) at (9,5.5) {};

\node[noeud, line width = 1mm, color = orange, fill = white!60!green] (u1) at (7,0) {};
\node[noeud] (u2) at (8,0) {};
\node[noeud] (u3) at (9,0) {};

\node[noeud] (v1) at (7.5,1.3) {};

\draw (z) node[below left,scale=1.2, color = orange] {$u^*$};
\draw (t3) node[below right,scale=1.2, color = orange] {$b^*$};

\draw [line width = 1.4pt] (x) -- (a);
\draw [line width = 1.4pt] (y) -- (a);
\draw [line width = 1.4pt] (z) -- (a);
\draw [line width = 1.4pt] (b) -- (a);
\draw [line width = 1.4pt] (b) -- (w);

\draw [out=60,in=300,line width = 1.4pt] (w) to (y);
\draw [out=-30,in=-120,line width = 1.4pt] (z) to (t2);

\draw [line width = 1.4pt] (t2) -- (t3);

\draw [line width = 1.4pt] (x) -- (t1);
\draw [line width = 1.4pt] (y) -- (t1);
\draw [line width = 1.4pt] (x) -- (t2);
\draw [line width = 1.4pt] (y) -- (t2);

\draw [line width = 1.4pt] (x) to (d);
\draw [line width = 1.4pt] (y) to (d);
\draw [line width = 1.4pt] (d) -- (t3);

\draw [line width = 1.4pt] (u1) -- (u2);
\draw [line width = 1.4pt] (u2) -- (u3);

\draw [line width = 1.4pt] (z) -- (u1);
\draw [line width = 1.4pt] (w) -- (u1);
\draw [line width = 1.4pt] (t1) -- (u1);
\draw [line width = 1.4pt] (t3) -- (u3);

\draw [out=0,in=200,line width = 1.4pt] (w) to (v1);
\draw [out=30,in=-110,line width = 1.4pt] (v1) to (t2);
\end{tikzpicture}

%% file: tikz/neighbor-no.tex
\begin{tikzpicture}

\draw[rounded corners,line width=1mm, color = white!80!orange, fill = white!80!orange] (4.6,3.0) rectangle (5.5,0.5);
\draw[rounded corners,line width=0.5mm, color = blue] (4.5,0.4) rectangle (5.6,6.0);

\draw[rounded corners,line width=1mm, color = white!80!orange, fill = white!80!orange] (6.4,2.4) rectangle (10.6,3.4);
\draw[rounded corners,line width=0.5mm, color = blue] (6.3,2.3) rectangle (10.7,3.5);

\draw[rounded corners,line width=1mm, color = white!80!orange, fill = white!80!orange] (2.4,3.4) rectangle (3.4,4.4);

\draw (3.9,0.2) node[scale=1.5, color =  white!60!orange] {$\Omega'$};

\draw (8.7,3.9) node[scale=1.3, color =  blue] {$\tilde{\Omega}$};

\node[noeud, line width = 1mm, color = orange, fill = white!60!green] (x) at (5,5.5) {};
\node[noeud, line width = 1mm, color = orange, fill = white!60!green] (y) at (5,4) {};
\node[noeud, line width = 1mm, color = orange, fill = white!60!green] (z) at (5,2.5) {};
\node[noeud, color = black, fill = white!60!green] (w) at (5,1) {};

\node[noeud, line width = 1mm, color = orange, fill = black] (a) at (3,4) {};
\node[noeud, line width = 1mm, color = orange, fill = black] (b) at (3,2.0) {};

\node[noeud, color = black, fill = black] (t1) at (7,3) {};
\node[noeud] (t2) at (8.5,3) {};
\node[noeud] (t3) at (10,3) {};

%
\node[noeud] (d) at (9,5.5) {};

\node[noeud] (u1) at (7,0) {};
\node[noeud] (u2) at (8,0) {};
\node[noeud] (u3) at (9,0) {};

\node[noeud] (v1) at (7.5,1.3) {};

\draw (t1) node[below right,scale=1.2] {$x$};
\draw (t2) node[below right,scale=1.2] {$y$};
\draw (a) node[below left,scale=1.2, color = orange] {$u^*$};
\draw (t3) node[below right,scale=1.2, color = orange] {$b^*$};

\draw [line width = 1.4pt] (x) -- (a);
\draw [line width = 1.4pt] (y) -- (a);
\draw [line width = 1.4pt] (z) -- (a);
\draw [line width = 1.4pt] (b) -- (a);
\draw [line width = 1.4pt] (b) -- (w);

\draw [out=60,in=300,line width = 1.4pt] (w) to (y);
\draw [out=-30,in=-120,line width = 1.4pt] (z) to (t2);

\draw [line width = 1.4pt] (t2) -- (t3);

\draw [line width = 1.4pt] (x) -- (t1);
\draw [line width = 1.4pt] (y) -- (t1);
\draw [line width = 1.4pt] (x) -- (t2);
\draw [line width = 1.4pt] (y) -- (t2);

\draw [line width = 1.4pt] (x) to (d);
\draw [line width = 1.4pt] (y) to (d);
\draw [line width = 1.4pt] (d) -- (t3);

\draw [line width = 1.4pt] (u1) -- (u2);
\draw [line width = 1.4pt] (u2) -- (u3);

\draw [line width = 1.4pt] (z) -- (u1);
\draw [line width = 1.4pt] (w) -- (u1);
\draw [line width = 1.4pt] (t1) -- (u1);
\draw [line width = 1.4pt] (t3) -- (u3);

\draw [out=0,in=200,line width = 1.4pt] (w) to (v1);
\draw [out=30,in=-110,line width = 1.4pt] (v1) to (t2);
\end{tikzpicture}

%% file: tikz/ashutter.tex
\begin{tikzpicture}

\node[noeud] (x) at (5,5.5) {};
\node[noeud] (y) at (5,4) {};
\node[noeud] (z) at (5,2.5) {};
\node[noeud] (w) at (5,1) {};

\node[noeud, fill = purple] (a) at (2.5,4) {};
\node[noeud] (b) at (2,2.0) {};

\node[noeud] (t1) at (7,3) {};
\node[noeud] (t2) at (9,3) {};
\node[noeud] (t3) at (11,3) {};

\node (c) at (8,4.75) {};
\node[noeud] (c1) at (8,5) {};
\node[noeud] (c2) at (7.6,4.5) {};
\node[noeud] (c3) at (8.4,4.5) {};

\node (d) at (9.85,6.75) {};
\node[noeud] (d1) at (10.5,7) {};
\node[noeud] (d2) at (9.2,6.5) {};

\node[noeud] (u1) at (7,0) {};
\node[noeud] (u2) at (8,0) {};
\node[noeud] (u3) at (9,0) {};

\node[noeud] (v1) at (7.5,1.3) {};

\draw (x) node[above left,scale=1.5] {$x$};
\draw (y) node[above left,scale=1.5] {$y$};
\draw (z) node[below left,scale=1.5] {$z$};
\draw (w) node[below left,scale=1.5] {$w$};
\draw (a) node[above left,scale=1.5] {$a$};
\draw (t1) node[below right,scale=1.5] {$t_1$};
\draw (t2) node[below right,scale=1.5] {$t_2$};
\draw (t3) node[below right,scale=1.5] {$t_3$};

\draw[rounded corners,line width=1mm, color = red] (4.2,6.3) rectangle (5.8,0.2);
\draw (5,6.8) node[scale=1.5, color = red] {$S$};

\draw[rounded corners,line width=1mm, color = red] (6.4,2.2) rectangle (11.8,3.7);
\draw (12.3,3) node[scale=1.5, color = red] {$T$};

\draw[rounded corners,line width=0.3mm, color = white!50!green] (6.1,-0.8) rectangle (12.8,8.1);
\draw (11.7,-0.5) node[scale=1.2, color = green] {$C_G(T,S)$};

\draw[rounded corners,line width=0.3mm, color = white!50!green] (3.9,6.5) rectangle (0.5,0);
\draw (1.6,0.5) node[scale=1.2, color = green] {$C_G(a,S)$};

\draw[line width=0.3mm, color = blue, dashed] (c) circle (0.75);
\draw[line width=0.3mm, color = blue, dashed] (v1) circle (0.3);
\draw[rounded corners,line width=0.3mm, color = blue, dashed] (6.7,-0.3) rectangle (9.3,0.3);
\draw[line width=0.3mm, color = blue, dashed] (d) circle (1.2);

\draw[rounded corners,line width=0.5mm, color = purple] (1.8,3.4) rectangle (3.1,4.7);
\draw (2.5,5.0) node[scale=1, color = purple] {$D_0 = \{a\}$};

\draw [line width = 1.4pt] (x) -- (a);
\draw [line width = 1.4pt] (y) -- (a);
\draw [line width = 1.4pt] (z) -- (a);
\draw [line width = 1.4pt] (b) -- (a);
\draw [line width = 1.4pt] (b) -- (w);

\draw [out=60,in=300,line width = 1.4pt] (w) to (y);
\draw [out=-30,in=-120,line width = 1.4pt] (z) to (t2);

\draw [line width = 1.4pt] (t2) -- (t3);

\draw [line width = 1.4pt] (c1) -- (c2);
\draw [line width = 1.4pt] (c2) -- (c3);
\draw [line width = 1.4pt] (c1) -- (c3);

\draw [line width = 1.4pt] (x) -- (c1);
\draw [line width = 1.4pt] (y) -- (c2);
\draw [line width = 1.4pt] (t1) -- (c2);
\draw [line width = 1.4pt] (t2) -- (c3);

\draw [line width = 1.4pt] (d1) -- (d2);

\draw [out=50,in=160,line width = 1.4pt] (x) to (d1);
\draw [out=50,in=185,line width = 1.4pt] (y) to (d2);
\draw [line width = 1.4pt] (d2) -- (t3);

\draw [line width = 1.4pt] (u1) -- (u2);
\draw [line width = 1.4pt] (u2) -- (u3);

\draw [line width = 1.4pt] (z) -- (u1);
\draw [line width = 1.4pt] (w) -- (u1);
\draw [line width = 1.4pt] (t1) -- (u1);
\draw [line width = 1.4pt] (t3) -- (u3);

\draw [out=0,in=200,line width = 1.4pt] (w) to (v1);
\draw [out=30,in=-110,line width = 1.4pt] (v1) to (t2);
\end{tikzpicture}

%% file: tikz/notashutter.tex
\begin{tikzpicture}

\node[noeud] (x) at (5,5.5) {};
\node[noeud] (y) at (5,4) {};
\node[noeud] (z) at (5,2.5) {};
\node[noeud] (w) at (5,1) {};

\node[noeud, fill = purple] (a) at (0.5,4) {};
\node[noeud] (b) at (0,2.0) {};

\node[noeud] (t1) at (7,3) {};
\node[noeud] (t2) at (9,3) {};
\node[noeud] (t3) at (11,3) {};

\node (c) at (8,4.75) {};
\node[noeud] (c1) at (8,5) {};
\node[noeud] (c2) at (7.6,4.5) {};
\node[noeud] (c3) at (8.4,4.5) {};

\node[noeud] (d) at (10,6.2) {};

\node[noeud] (u1) at (7,0) {};
\node[noeud] (u2) at (8,0) {};
\node[noeud] (u3) at (9,0) {};

\node[noeud] (v1) at (7.0,1.6) {};

\node[noeud, fill = orange] (x1) at (3,5.7) {};
\node[noeud, fill = orange] (x2) at (3,5.1) {};
\node[noeud, fill = orange] (x3) at (1.5,3) {};

\node[noeud] (p) at (3.5,3.4) {};

\node[noeud, fill = yellow] (z1) at (3.2,2.7) {};
\node[noeud, fill = yellow] (z2) at (3.2,2.1) {};

\draw (x) node[above left,scale=1.5] {$x$};
\draw (y) node[below left,scale=1.5] {$y$};
\draw (z) node[below left,scale=1.5] {$z$};
\draw (w) node[below left,scale=1.5] {$w$};
\draw (a) node[above left,scale=1.5] {$a$};
\draw (t1) node[below right,scale=1.5] {$t_1$};
\draw (t2) node[below right,scale=1.5] {$t_2$};
\draw (t3) node[below right,scale=1.5] {$t_3$};

\draw[rounded corners,line width=1mm, color = red] (4.2,6.3) rectangle (5.8,0.2);
\draw (5,6.8) node[scale=1.5, color = red] {$S$};

\draw[rounded corners,line width=1mm, color = red] (6.4,2.2) rectangle (11.8,3.7);
\draw (12.3,3) node[scale=1.5, color = red] {$T$};

\draw[rounded corners,line width=0.3mm, color = white!50!green] (6.1,-0.8) rectangle (12.8,7.6);
\draw (11.7,-0.5) node[scale=1.2, color = green] {$C_G(T,S)$};

\draw[rounded corners,line width=0.3mm, color = white!50!green] (3.9,7.2) rectangle (-1.5,0);
\draw (0.6,0.5) node[scale=1.2, color = green] {$C_G(a,S)$};

\draw[line width=0.3mm, color = blue, dashed] (c) circle (0.75);
\draw[line width=0.3mm, color = blue, dashed] (v1) circle (0.3);
\draw[rounded corners,line width=0.3mm, color = blue, dashed] (6.7,-0.3) rectangle (9.3,0.3);
\draw[line width=0.3mm, color = blue, dashed] (d) circle (0.3);

\draw[rounded corners,line width=0.5mm, color = purple] (-0.2,3.4) rectangle (1.1,4.7);
\draw (0.5,5.0) node[scale=1.1, color = purple] {$D_0' = \{a\}$};

\draw[rounded corners,line width=0.5mm, color = orange] (2.7,4.8) rectangle (3.3,5.4);
\draw (2.9,4.4) node[scale=1.0, color = orange] {$D_1'$};

\draw[rounded corners,line width=0.5mm, color = black!10!yellow] (2.9,1.8) rectangle (3.5,3.0);
\draw (3.2,1.4) node[scale=1.1, color = black!20!yellow] {$D_2'$};
\draw [line width = 1.4pt] (b) -- (a);
\draw [out=-40,in=190,line width = 1.4pt] (b) to (w);

\draw [out=60,in=300,line width = 1.4pt] (w) to (y);

\draw [line width = 1.4pt] (c1) -- (c2);
\draw [line width = 1.4pt] (c2) -- (c3);
\draw [line width = 1.4pt] (c1) -- (c3);

\draw [line width = 1.4pt] (x) -- (c1);
\draw [out=60,in=200,line width = 1.4pt] (z) to (c2);
\draw [line width = 1.4pt] (t1) -- (c2);
\draw [line width = 1.4pt] (t2) -- (c3);

\draw [line width = 1.4pt] (x) -- (d);
\draw [line width = 1.4pt] (d) -- (t3);

\draw [line width = 1.4pt] (y) -- (t1);
\draw [line width = 1.4pt] (y) -- (t2);
\draw [out=0,in=160,line width = 1.4pt] (y) to (t3);
%
%
\draw [line width = 1.4pt] (u1) -- (u2);
\draw [line width = 1.4pt] (u2) -- (u3);
%
\draw [line width = 1.4pt] (w) -- (u1);
\draw [line width = 1.4pt] (z) -- (u2);
\draw [line width = 1.4pt] (t2) -- (u2);
\draw [line width = 1.4pt] (t3) -- (u3);
\draw [line width = 1.4pt] (w) -- (v1);
\draw [line width = 1.4pt] (v1) -- (t1);
\draw [out=0, in=230,line width = 1.4pt] (v1) to (t3);

\draw [line width = 1.4pt] (x) -- (x1);
\draw [line width = 1.4pt] (x) -- (x2);
\draw [line width = 1.4pt] (x) -- (x3);

\draw [line width = 1.4pt] (y) -- (x1);
\draw [line width = 1.4pt] (y) -- (x2);

\draw [line width = 1.4pt] (y) -- (a);
\draw [line width = 1.4pt] (x1) -- (a);
\draw [line width = 1.4pt] (x2) -- (a);
\draw [line width = 1.4pt] (x3) -- (a);

\draw [line width = 1.4pt] (z) -- (z1);
\draw [line width = 1.4pt] (z) -- (z2);
\draw [line width = 1.4pt] (z1) -- (z2);

\draw [line width = 1.4pt] (p) -- (z1);
\draw [line width = 1.4pt] (p) -- (x2);

\draw [out=140,in=230,line width = 1.4pt] (z2) to (x2);
\end{tikzpicture}

%% file: tikz/ashutter2.tex
\begin{tikzpicture}

\node[noeud] (x) at (5,5.5) {};
\node[noeud] (y) at (5,4) {};
\node[noeud] (z) at (5,2.5) {};
\node[noeud] (w) at (5,1) {};

\node[noeud, fill = purple] (a) at (2.5,4) {};
\node[noeud] (b) at (2,2.0) {};

\node[noeud] (t1) at (7,3) {};
\node[noeud] (t2) at (9,3) {};
\node[noeud] (t3) at (11,3) {};

\node (c) at (8,4.75) {};
\node[noeud] (c1) at (8,5) {};
\node[noeud] (c2) at (7.6,4.5) {};
\node[noeud] (c3) at (8.4,4.5) {};

\node (d) at (9.85,6.75) {};
\node[noeud] (d1) at (10.5,7) {};
\node[noeud] (d2) at (9.2,6.5) {};

\node[noeud] (u1) at (7,0) {};
\node[noeud] (u2) at (8,0) {};
\node[noeud] (u3) at (9,0) {};

\node[noeud] (v1) at (7.5,1.3) {};

\draw (x) node[above left,scale=1.5] {$x$};
\draw (y) node[above left,scale=1.5] {$y$};
\draw (z) node[below left,scale=1.5] {$z$};
\draw (w) node[below left,scale=1.5] {$w$};
\draw (a) node[above left,scale=1.5] {$a$};
\draw (t1) node[below right,scale=1.5] {$t_1$};
\draw (t2) node[below right,scale=1.5] {$t_2$};
\draw (t3) node[below right,scale=1.5] {$t_3$};

\draw[rounded corners,line width=1mm, color = red] (4.2,6.3) rectangle (5.8,0.2);
\draw (5,6.8) node[scale=1.5, color = red] {$S$};

\draw[rounded corners,line width=1mm, color = red] (6.4,2.2) rectangle (11.8,3.7);
\draw (12.3,3) node[scale=1.5, color = red] {$T$};


\draw[rounded corners,line width=0.3mm, color = white!50!green] (3.9,6.5) rectangle (0.5,0);
\draw (1.6,0.5) node[scale=1.2, color = green] {$C_G(a,S)$};


\draw[rounded corners,line width=0.5mm, color = purple] (1.8,3.4) rectangle (3.1,4.7);
\draw (2.5,5.0) node[scale=1, color = purple] {$D$};

\draw [line width = 1.4pt] (x) -- (a);
\draw [line width = 1.4pt] (y) -- (a);
\draw [line width = 1.4pt] (z) -- (a);
\draw [line width = 1.4pt] (b) -- (a);
\draw [line width = 1.4pt] (b) -- (w);

\draw [out=60,in=300,line width = 1.4pt] (w) to (y);
\draw [out=-30,in=-120,line width = 1.4pt] (z) to (t2);

\draw [line width = 1.4pt] (t2) -- (t3);

\draw [line width = 1.4pt] (c1) -- (c2);
\draw [line width = 1.4pt] (c2) -- (c3);
\draw [line width = 1.4pt] (c1) -- (c3);

\draw [line width = 1.4pt] (x) -- (c1);
\draw [line width = 1.4pt] (y) -- (c2);
\draw [line width = 1.4pt] (t1) -- (c2);
\draw [line width = 1.4pt] (t2) -- (c3);

\draw [line width = 1.4pt] (d1) -- (d2);

\draw [out=50,in=160,line width = 1.4pt] (x) to (d1);
\draw [out=50,in=185,line width = 1.4pt] (y) to (d2);
\draw [line width = 1.4pt] (d2) -- (t3);

\draw [line width = 1.4pt] (u1) -- (u2);
\draw [line width = 1.4pt] (u2) -- (u3);

\draw [line width = 1.4pt] (z) -- (u1);
\draw [line width = 1.4pt] (w) -- (u1);
\draw [line width = 1.4pt] (t1) -- (u1);
\draw [line width = 1.4pt] (t3) -- (u3);

\draw [out=0,in=200,line width = 1.4pt] (w) to (v1);
\draw [out=30,in=-110,line width = 1.4pt] (v1) to (t2);
\end{tikzpicture}

%% file: tikz/blue-cover.tex
\begin{tikzpicture}

\node[noeud, fill = blue] (x) at (2,2) {};
\node[noeud, fill = blue] (y) at (2,4) {};
\node[noeud, fill = blue] (z) at (4,4) {};
\node[noeud, fill = red] (w) at (4,2) {};

\draw (x) node[below left,scale=1] {$x$};
\draw (y) node[above left,scale=1] {$y$};
\draw (z) node[above right,scale=1] {$z$};
\draw (w) node[below right,scale=1] {$w$};

\draw [line width = 1.4pt] (x) -- (w);
\draw [line width = 1.4pt] (x) -- (z);
\draw [line width = 1.4pt] (y) -- (z);
\end{tikzpicture}

%% file: tikz/preimage-a.tex
\begin{tikzpicture}

\node[noeud] (x) at (5,5.5) {};
\node[noeud] (y) at (5,4) {};
\node[noeud] (z) at (5,2.5) {};
\node[noeud] (w) at (5,1) {};

\node[noeud] (a) at (2.5,4) {};
\node[noeud] (b) at (2,2.0) {};

\node[noeud] (t1) at (7,3) {};
\node[noeud] (t2) at (9,3) {};
\node[noeud] (t3) at (11,3) {};

\node (c) at (8,4.75) {};
\node[noeud] (c1) at (8,5) {};
\node[noeud] (c2) at (7.6,4.5) {};
\node[noeud] (c3) at (8.4,4.5) {};

\node (d) at (9.85,6.75) {};
\node[noeud] (d1) at (10.5,7) {};
\node[noeud] (d2) at (9.2,6.5) {};

\node[noeud] (u1) at (7,0) {};
\node[noeud] (u2) at (8,0) {};
\node[noeud] (u3) at (9,0) {};

\node[noeud] (v1) at (7.5,1.3) {};

\draw (x) node[above left,scale=1.5] {$x$};
\draw (y) node[above left,scale=1.5] {$y$};
\draw (z) node[below left,scale=1.5] {$z$};
\draw (w) node[below left,scale=1.5] {$w$};
\draw (a) node[above left,scale=1.5] {$a$};
\draw (t1) node[below right,scale=1.5] {$t_1$};
\draw (t2) node[below right,scale=1.5] {$t_2$};
\draw (t3) node[below right,scale=1.5] {$t_3$};

\draw[rounded corners,line width=1mm, color = red] (4.2,6.3) rectangle (5.8,0.2);
\draw[rounded corners,line width=1mm, color = yellow] (4.3,3.2) rectangle (5.7,0.3);

\draw[rounded corners,line width=1mm, color = red] (6.4,2.2) rectangle (11.8,3.7);
\draw[rounded corners,line width=1mm, color = yellow] (6.5,2.3) rectangle (11.7,3.6);
\draw (5,6.8) node[scale=1.5, color = red] {$\Omega$};
\draw (2.5,5.2) node[scale=1.5, color = black!20!yellow] {$\Omega_a$};




\draw[rounded corners,line width=1mm, color = yellow] (1.8,3.4) rectangle (3.1,4.7);

\draw [line width = 1.4pt] (x) -- (a);
\draw [line width = 1.4pt] (y) -- (a);
\draw [line width = 1.4pt] (z) -- (a);
\draw [line width = 1.4pt] (b) -- (a);
\draw [line width = 1.4pt] (b) -- (w);

\draw [out=60,in=300,line width = 1.4pt] (w) to (y);
\draw [out=-30,in=-120,line width = 1.4pt] (z) to (t2);

\draw [line width = 1.4pt] (t2) -- (t3);

\draw [line width = 1.4pt] (c1) -- (c2);
\draw [line width = 1.4pt] (c2) -- (c3);
\draw [line width = 1.4pt] (c1) -- (c3);

\draw [line width = 1.4pt] (x) -- (c1);
\draw [line width = 1.4pt] (y) -- (c2);
\draw [line width = 1.4pt] (t1) -- (c2);
\draw [line width = 1.4pt] (t2) -- (c3);

\draw [line width = 1.4pt] (d1) -- (d2);

\draw [out=50,in=160,line width = 1.4pt] (x) to (d1);
\draw [out=50,in=185,line width = 1.4pt] (y) to (d2);
\draw [line width = 1.4pt] (d2) -- (t3);

\draw [line width = 1.4pt] (u1) -- (u2);
\draw [line width = 1.4pt] (u2) -- (u3);

\draw [line width = 1.4pt] (z) -- (u1);
\draw [line width = 1.4pt] (w) -- (u1);
\draw [line width = 1.4pt] (t1) -- (u1);
\draw [line width = 1.4pt] (t3) -- (u3);

\draw [out=0,in=200,line width = 1.4pt] (w) to (v1);
\draw [out=30,in=-110,line width = 1.4pt] (v1) to (t2);
\end{tikzpicture}

%% file: tikz/preimage-nota.tex
\begin{tikzpicture}

\node[noeud] (x) at (5,5.5) {};
\node[noeud] (y) at (5,4) {};
\node[noeud] (z) at (5,2.5) {};
\node[noeud] (w) at (5,1) {};

\node[noeud] (a) at (0.5,4) {};
\node[noeud] (b) at (0,2.0) {};

\node[noeud] (t1) at (7,3) {};
\node[noeud] (t2) at (9,3) {};
\node[noeud] (t3) at (11,3) {};

\node (c) at (8,4.75) {};
\node[noeud] (c1) at (8,5) {};
\node[noeud] (c2) at (7.6,4.5) {};
\node[noeud] (c3) at (8.4,4.5) {};

\node[noeud] (d) at (10,6.2) {};

\node[noeud] (u1) at (7,0) {};
\node[noeud] (u2) at (8,0) {};
\node[noeud] (u3) at (9,0) {};

\node[noeud] (v1) at (7.0,1.6) {};

\node[noeud] (x1) at (3,5.7) {};
\node[noeud] (x2) at (3,5.1) {};
\node[noeud] (x3) at (1.5,3) {};

\node[noeud] (p) at (3.5,3.4) {};

\node[noeud] (z1) at (3.2,2.7) {};
\node[noeud] (z2) at (3.2,2.1) {};

\draw (x) node[above left,scale=1.5] {$x$};
\draw (y) node[below left,scale=1.5] {$y$};
\draw (z) node[below left,scale=1.5] {$z$};
\draw (w) node[below left,scale=1.5] {$w$};
\draw (a) node[above left,scale=1.5] {$a$};
\draw (t1) node[below right,scale=1.5] {$t_1$};
\draw (t2) node[below right,scale=1.5] {$t_2$};
\draw (t3) node[below right,scale=1.5] {$t_3$};

\draw[rounded corners,line width=1mm, color = red] (4.2,6.3) rectangle (5.8,0.2);
\draw[rounded corners,line width=1mm, color = yellow] (4.3,1.7) rectangle (5.7,0.3);
\draw[rounded corners,line width=1mm, color = yellow] (4.3,6.2) rectangle (5.7,3.3);
\draw (5,6.8) node[scale=1.5, color = red] {$\Omega$};

\draw[rounded corners,line width=1mm, color = red] (6.4,2.2) rectangle (11.8,3.7);
\draw[rounded corners,line width=1mm, color = yellow] (6.5,2.3) rectangle (11.7,3.6);




%

\draw[rounded corners,line width=1mm, color = black!10!yellow] (2.8,1.7) rectangle (3.6,3.1);
\draw (2.5,1.3) node[scale=1.5, color = black!20!yellow] {$\Omega'$};
\draw [line width = 1.4pt] (b) -- (a);
\draw [out=-40,in=190,line width = 1.4pt] (b) to (w);

\draw [out=60,in=300,line width = 1.4pt] (w) to (y);

\draw [line width = 1.4pt] (c1) -- (c2);
\draw [line width = 1.4pt] (c2) -- (c3);
\draw [line width = 1.4pt] (c1) -- (c3);

\draw [line width = 1.4pt] (x) -- (c1);
\draw [out=60,in=200,line width = 1.4pt] (z) to (c2);
\draw [line width = 1.4pt] (t1) -- (c2);
\draw [line width = 1.4pt] (t2) -- (c3);

\draw [line width = 1.4pt] (x) -- (d);
\draw [line width = 1.4pt] (d) -- (t3);

\draw [line width = 1.4pt] (y) -- (t1);
\draw [line width = 1.4pt] (y) -- (t2);
\draw [out=0,in=160,line width = 1.4pt] (y) to (t3);
%
%
\draw [line width = 1.4pt] (u1) -- (u2);
\draw [line width = 1.4pt] (u2) -- (u3);
%
\draw [line width = 1.4pt] (w) -- (u1);
\draw [line width = 1.4pt] (z) -- (u2);
\draw [line width = 1.4pt] (t2) -- (u2);
\draw [line width = 1.4pt] (t3) -- (u3);
\draw [line width = 1.4pt] (w) -- (v1);
\draw [line width = 1.4pt] (v1) -- (t1);
\draw [out=0, in=230,line width = 1.4pt] (v1) to (t3);

\draw [line width = 1.4pt] (x) -- (x1);
\draw [line width = 1.4pt] (x) -- (x2);
\draw [line width = 1.4pt] (x) -- (x3);

\draw [line width = 1.4pt] (y) -- (x1);
\draw [line width = 1.4pt] (y) -- (x2);

\draw [line width = 1.4pt] (y) -- (a);
\draw [line width = 1.4pt] (x1) -- (a);
\draw [line width = 1.4pt] (x2) -- (a);
\draw [line width = 1.4pt] (x3) -- (a);

\draw [line width = 1.4pt] (z) -- (z1);
\draw [line width = 1.4pt] (z) -- (z2);
\draw [line width = 1.4pt] (z1) -- (z2);

\draw [line width = 1.4pt] (p) -- (z1);
\draw [line width = 1.4pt] (p) -- (x2);

\draw [out=140,in=230,line width = 1.4pt] (z2) to (x2);
\end{tikzpicture}

%% file: tikz/solutionpath.tex
\begin{tikzpicture}
\node (p0) at (2,30) {\scalebox{0.6}{\input{tikz/preimage-0}}};
\node (p1) at (10,30) {\scalebox{0.6}{\input{tikz/preimage-1}}};

\node (p2) at (10,22) {\scalebox{0.6}{\input{tikz/preimage-2}}};
\node (p3) at (2,22) {\scalebox{0.6}{\input{tikz/preimage-3}}};

\node (p4) at (2,15) {\scalebox{0.6}{\input{tikz/preimage-4}}};
\node (p5) at (10,15) {\scalebox{0.6}{\input{tikz/preimage-5}}};

\node (a5) at (8,12.5) {};
\node (b4) at (4,12.5) {};
\node (a4) at (1,17.5) {};
\node (b3) at (1,19.5) {};

\node (a3) at (4,19.5) {};
\node (b2) at (9,19.5) {};

\node (a2) at (7.5,24.5) {};
\node (b1) at (7.5,27.5) {};

\node (a1) at (7,27.5) {};
\node (b0) at (2,27.5) {};

\node[scale = 1.2,color = blue] (o) at (0.4,32.6) {Target PMC reached !};
\node[scale = 1.2,color = green] (o5) at (12.5,12.3) {$a \in \Omega^{(5)}$};

\draw[->,line width = 2pt,out=200, in=-20] (a5)to node[text centered,text width=30mm,yshift=-3mm]{$(\Omega^{(5)},\Omega^{(4)}) \in \mathcal{N}$} (b4);
\draw[->,line width = 2pt] (a4)to node[midway,text centered,text width=30mm,xshift=-15mm]{$(\Omega^{(4)},\Omega^{(3)}) \in \mathcal{N}$} (b3);

\draw[->,line width = 2pt,out=-20, in=200] (a3)to node[text centered,text width=120mm,xshift=17mm, yshift=-8mm]{$(\Omega^{(3)},\Omega^{(2)}) \in \mathcal{N}$: singleton $\{x_2\}$ is an $a$-shutter which separates pairs $(y,u)$ and $(y,z_2)$ closely, $u$ and $z_2$ are fully adjacent to $\{x_2\}$ and not adjacent to $C(a,\Omega^{(2)} \cup \{x_2\})$.} (b2);

\draw[->,line width = 2pt] (a2)to node[text centered,text width=80mm,xshift=34mm]{$(\Omega^{(2)},\Omega^{(1)}) \in \mathcal{N}$: singleton $\{u\}$ is an $a$-shutter which separates pair $(y,z_1)$ closely, $z_1$ is fully adjacent to $\{u\}$ and not adjacent to $C(a,\Omega^{(1)} \cup \{u\})$.} (b1);

\draw[->,line width = 2pt,out=200, in=-20] (a1)to node[text centered,text width=80mm,xshift=-20mm, yshift=-13mm]{$(\Omega^{(1)},\Omega) \in \mathcal{N}$: singleton $\{a\}$ separates $x$ from $w$ but is not an $a$-shutter since it is not adjacent to them. Set $\{z_1,z_2\}$ is an $a$-shutter of $\Omega$: it is the close $(z,y)$-separator, and $z$ is not adjacent to $C(a,\Omega \cup \{z_1,z_2\})$.} (b0);
\end{tikzpicture}

%% file: tikz/preimage-0.tex
\begin{tikzpicture}

\node[noeud] (x) at (5,5.5) {};
\node[noeud] (y) at (5,4) {};
\node[noeud] (z) at (5,2.5) {};
\node[noeud] (w) at (5,1) {};

\node[noeud, fill = green] (a) at (0.5,4) {};
\node[noeud, fill = green] (b) at (0,2.0) {};

\node[noeud] (t1) at (7,3) {};
\node[noeud] (t2) at (9,3) {};
\node[noeud] (t3) at (11,3) {};

\node (c) at (8,4.75) {};
\node[noeud] (c1) at (8,5) {};
\node[noeud] (c2) at (7.6,4.5) {};
\node[noeud] (c3) at (8.4,4.5) {};

\node[noeud] (d) at (10,6.2) {};

\node[noeud] (u1) at (7,0) {};
\node[noeud] (u2) at (8,0) {};
\node[noeud] (u3) at (9,0) {};

\node[noeud] (v1) at (7.0,1.6) {};

\node[noeud, fill = green] (x1) at (3,5.7) {};
\node[noeud, fill = green] (x2) at (3,5.1) {};
\node[noeud, fill = green] (x3) at (1.5,3) {};

\node[noeud, fill = green] (p) at (3.5,3.4) {};

\node[noeud, fill = green] (z1) at (3.2,2.7) {};
\node[noeud, fill = green] (z2) at (3.2,2.1) {};

\draw (x) node[above left,scale=1.5] {$x$};
\draw (y) node[below left,scale=1.5] {$y$};
\draw (z) node[below left,scale=1.5] {$z$};
\draw (w) node[below left,scale=1.5] {$w$};
\draw (a) node[above left,scale=1.5] {$a$};

\draw[rounded corners,line width=1.5mm, color = blue] (-0.5,-0.5) rectangle (12.0,7.2);

\draw[rounded corners,line width=1mm, color = red] (4.2,6.3) rectangle (5.8,0.2);
\draw (5,6.8) node[scale=2, color = red] {$\Omega$};

\draw[rounded corners,line width=1mm, color = red] (6.4,2.2) rectangle (11.8,3.7);




%

\draw [line width = 1.4pt] (b) -- (a);
\draw [out=-40,in=190,line width = 1.4pt] (b) to (w);

\draw [out=60,in=300,line width = 1.4pt] (w) to (y);

\draw [line width = 1.4pt] (c1) -- (c2);
\draw [line width = 1.4pt] (c2) -- (c3);
\draw [line width = 1.4pt] (c1) -- (c3);

\draw [line width = 1.4pt] (x) -- (c1);
\draw [out=60,in=200,line width = 1.4pt] (z) to (c2);
\draw [line width = 1.4pt] (t1) -- (c2);
\draw [line width = 1.4pt] (t2) -- (c3);

\draw [line width = 1.4pt] (x) -- (d);
\draw [line width = 1.4pt] (d) -- (t3);

\draw [line width = 1.4pt] (y) -- (t1);
\draw [line width = 1.4pt] (y) -- (t2);
\draw [out=0,in=160,line width = 1.4pt] (y) to (t3);
%
%
\draw [line width = 1.4pt] (u1) -- (u2);
\draw [line width = 1.4pt] (u2) -- (u3);
\draw [line width = 1.4pt] (w) -- (t1);
\draw [line width = 1.4pt] (w) -- (u1);
\draw [line width = 1.4pt] (t2) -- (u2);
\draw [line width = 1.4pt] (t3) -- (u3);
\draw [line width = 1.4pt] (z) -- (v1);
\draw [line width = 1.4pt] (v1) -- (t1);
\draw [out=0, in=230,line width = 1.4pt] (v1) to (t3);

\draw [line width = 1.4pt] (x) -- (x1);
\draw [line width = 1.4pt] (x) -- (x2);
\draw [line width = 1.4pt] (x) -- (x3);

\draw [line width = 1.4pt] (y) -- (x1);
\draw [line width = 1.4pt] (y) -- (x2);

\draw [line width = 1.4pt] (y) -- (a);
\draw [line width = 1.4pt] (x1) -- (a);
\draw [line width = 1.4pt] (x2) -- (a);
\draw [line width = 1.4pt] (x3) -- (a);

\draw [line width = 1.4pt] (z) -- (z1);
\draw [line width = 1.4pt] (z) -- (z2);
\draw [line width = 1.4pt] (z1) -- (z2);

\draw [line width = 1.4pt] (p) -- (z1);
\draw [line width = 1.4pt] (p) -- (x2);

\draw [out=140,in=230,line width = 1.4pt] (z2) to (x2);
\end{tikzpicture}

%% file: tikz/preimage-1.tex
\begin{tikzpicture}

\node[noeud] (x) at (5,5.5) {};
\node[noeud] (y) at (5,4) {};
\node[noeud] (z) at (5,2.5) {};
\node[noeud] (w) at (5,1) {};

\node[noeud, fill = green] (a) at (0.5,4) {};
\node[noeud, fill = green] (b) at (0,2.0) {};

\node[noeud] (t1) at (7,3) {};
\node[noeud] (t2) at (9,3) {};
\node[noeud] (t3) at (11,3) {};

\node (c) at (8,4.75) {};
\node[noeud] (c1) at (8,5) {};
\node[noeud] (c2) at (7.6,4.5) {};
\node[noeud] (c3) at (8.4,4.5) {};

\node[noeud] (d) at (10,6.2) {};

\node[noeud] (u1) at (7,0) {};
\node[noeud] (u2) at (8,0) {};
\node[noeud] (u3) at (9,0) {};

\node[noeud] (v1) at (7.0,1.6) {};

\node[noeud, fill = green] (x1) at (3,5.7) {};
\node[noeud, fill = green] (x2) at (3,5.1) {};
\node[noeud, fill = green] (x3) at (1.5,3) {};

\node[noeud, fill = green] (p) at (3.5,3.4) {};

\node[noeud] (z1) at (3.2,2.7) {};
\node[noeud] (z2) at (3.2,2.1) {};

\draw (x) node[above left,scale=1.5] {$x$};
\draw (y) node[below left,scale=1.5] {$y$};
\draw (w) node[below left,scale=1.5] {$w$};
\draw (a) node[above left,scale=1.5] {$a$};
\draw (z1) node[below right,scale=1.5] {$z_1$};
\draw (z2) node[below right,scale=1.5] {$z_2$};

\draw[rounded corners,line width=1mm] (-0.5,-0.5) rectangle (12.0,7.2);

\draw[line width=1mm, color = red] (w) circle (0.3);
\draw[rounded corners,line width=1mm, color = red] (4.3,6.2) rectangle (5.7,3.3);
\draw (5,6.8) node[scale=2, color = red] {$\Omega^{(1)}$};

\draw[rounded corners,line width=1mm, color = red] (6.5,2.3) rectangle (11.7,3.6);



\draw[rounded corners,line width=1mm, color = red] (2.8,1.5) rectangle (4.0,3.1);
\draw [line width = 1.4pt] (b) -- (a);
\draw [out=-40,in=190,line width = 1.4pt] (b) to (w);

\draw [out=60,in=300,line width = 1.4pt] (w) to (y);

\draw [line width = 1.4pt] (c1) -- (c2);
\draw [line width = 1.4pt] (c2) -- (c3);
\draw [line width = 1.4pt] (c1) -- (c3);

\draw [line width = 1.4pt] (x) -- (c1);
\draw [out=60,in=200,line width = 1.4pt] (z) to (c2);
\draw [line width = 1.4pt] (t1) -- (c2);
\draw [line width = 1.4pt] (t2) -- (c3);

\draw [line width = 1.4pt] (x) -- (d);
\draw [line width = 1.4pt] (d) -- (t3);

\draw [line width = 1.4pt] (y) -- (t1);
\draw [line width = 1.4pt] (y) -- (t2);
\draw [out=0,in=160,line width = 1.4pt] (y) to (t3);
%
%
\draw [line width = 1.4pt] (u1) -- (u2);
\draw [line width = 1.4pt] (u2) -- (u3);
\draw [line width = 1.4pt] (w) -- (t1);
\draw [line width = 1.4pt] (w) -- (u1);
\draw [line width = 1.4pt] (t2) -- (u2);
\draw [line width = 1.4pt] (t3) -- (u3);
\draw [line width = 1.4pt] (z) -- (v1);
\draw [line width = 1.4pt] (v1) -- (t1);
\draw [out=0, in=230,line width = 1.4pt] (v1) to (t3);

\draw [line width = 1.4pt] (x) -- (x1);
\draw [line width = 1.4pt] (x) -- (x2);
\draw [line width = 1.4pt] (x) -- (x3);

\draw [line width = 1.4pt] (y) -- (x1);
\draw [line width = 1.4pt] (y) -- (x2);

\draw [line width = 1.4pt] (y) -- (a);
\draw [line width = 1.4pt] (x1) -- (a);
\draw [line width = 1.4pt] (x2) -- (a);
\draw [line width = 1.4pt] (x3) -- (a);

\draw [line width = 1.4pt] (z) -- (z1);
\draw [line width = 1.4pt] (z) -- (z2);
\draw [line width = 1.4pt] (z1) -- (z2);

\draw [line width = 1.4pt] (p) -- (z1);
\draw [line width = 1.4pt] (p) -- (x2);

\draw [out=140,in=230,line width = 1.4pt] (z2) to (x2);
\end{tikzpicture}

%% file: tikz/preimage-2.tex
\begin{tikzpicture}

\node[noeud] (x) at (5,5.5) {};
\node[noeud] (y) at (5,4) {};
\node[noeud] (z) at (5,2.5) {};
\node[noeud] (w) at (5,1) {};

\node[noeud, fill = green] (a) at (0.5,4) {};
\node[noeud, fill = green] (b) at (0,2.0) {};

\node[noeud] (t1) at (7,3) {};
\node[noeud] (t2) at (9,3) {};
\node[noeud] (t3) at (11,3) {};

\node (c) at (8,4.75) {};
\node[noeud] (c1) at (8,5) {};
\node[noeud] (c2) at (7.6,4.5) {};
\node[noeud] (c3) at (8.4,4.5) {};

\node[noeud] (d) at (10,6.2) {};

\node[noeud] (u1) at (7,0) {};
\node[noeud] (u2) at (8,0) {};
\node[noeud] (u3) at (9,0) {};

\node[noeud] (v1) at (7.0,1.6) {};

\node[noeud, fill = green] (x1) at (3,5.7) {};
\node[noeud, fill = green] (x2) at (3,5.1) {};
\node[noeud, fill = green] (x3) at (1.5,3) {};

\node[noeud] (p) at (3.5,3.4) {};

\node[noeud] (z1) at (3.2,2.7) {};
\node[noeud] (z2) at (3.2,2.1) {};

\draw (x) node[above left,scale=1.5] {$x$};
\draw (y) node[below left,scale=1.5] {$y$};
\draw (w) node[below left,scale=1.5] {$w$};
\draw (a) node[above left,scale=1.5] {$a$};
\draw (z2) node[below right,scale=1.5] {$z_2$};
\draw (p) node[below right,scale=1.5] {$u$};

\draw[rounded corners,line width=1mm] (-0.5,-0.5) rectangle (12.0,7.2);

\draw[line width=1mm, color = red] (w) circle (0.3);
\draw[rounded corners,line width=1mm, color = red] (4.3,6.2) rectangle (5.7,3.3);
\draw (5,6.8) node[scale=2, color = red] {$\Omega^{(2)}$};

\draw[rounded corners,line width=1mm, color = red] (6.5,2.3) rectangle (11.7,3.6);




\draw[line width=1mm, color = red] (p) circle (0.3);
\draw[line width=1mm, color = red] (z2) circle (0.3);
\draw [line width = 1.4pt] (b) -- (a);
\draw [out=-40,in=190,line width = 1.4pt] (b) to (w);

\draw [out=60,in=300,line width = 1.4pt] (w) to (y);

\draw [line width = 1.4pt] (c1) -- (c2);
\draw [line width = 1.4pt] (c2) -- (c3);
\draw [line width = 1.4pt] (c1) -- (c3);

\draw [line width = 1.4pt] (x) -- (c1);
\draw [out=60,in=200,line width = 1.4pt] (z) to (c2);
\draw [line width = 1.4pt] (t1) -- (c2);
\draw [line width = 1.4pt] (t2) -- (c3);

\draw [line width = 1.4pt] (x) -- (d);
\draw [line width = 1.4pt] (d) -- (t3);

\draw [line width = 1.4pt] (y) -- (t1);
\draw [line width = 1.4pt] (y) -- (t2);
\draw [out=0,in=160,line width = 1.4pt] (y) to (t3);
%
%
\draw [line width = 1.4pt] (u1) -- (u2);
\draw [line width = 1.4pt] (u2) -- (u3);
\draw [line width = 1.4pt] (w) -- (t1);
\draw [line width = 1.4pt] (w) -- (u1);
\draw [line width = 1.4pt] (t2) -- (u2);
\draw [line width = 1.4pt] (t3) -- (u3);
\draw [line width = 1.4pt] (z) -- (v1);
\draw [line width = 1.4pt] (v1) -- (t1);
\draw [out=0, in=230,line width = 1.4pt] (v1) to (t3);

\draw [line width = 1.4pt] (x) -- (x1);
\draw [line width = 1.4pt] (x) -- (x2);
\draw [line width = 1.4pt] (x) -- (x3);

\draw [line width = 1.4pt] (y) -- (x1);
\draw [line width = 1.4pt] (y) -- (x2);

\draw [line width = 1.4pt] (y) -- (a);
\draw [line width = 1.4pt] (x1) -- (a);
\draw [line width = 1.4pt] (x2) -- (a);
\draw [line width = 1.4pt] (x3) -- (a);

\draw [line width = 1.4pt] (z) -- (z1);
\draw [line width = 1.4pt] (z) -- (z2);
\draw [line width = 1.4pt] (z1) -- (z2);

\draw [line width = 1.4pt] (p) -- (z1);
\draw [line width = 1.4pt] (p) -- (x2);

\draw [out=140,in=230,line width = 1.4pt] (z2) to (x2);
\end{tikzpicture}

%% file: tikz/preimage-3.tex
\begin{tikzpicture}

\node[noeud] (x) at (5,5.5) {};
\node[noeud] (y) at (5,4) {};
\node[noeud] (z) at (5,2.5) {};
\node[noeud] (w) at (5,1) {};

\node[noeud, fill = green] (a) at (0.5,4) {};
\node[noeud, fill = green] (b) at (0,2.0) {};

\node[noeud] (t1) at (7,3) {};
\node[noeud] (t2) at (9,3) {};
\node[noeud] (t3) at (11,3) {};

\node (c) at (8,4.75) {};
\node[noeud] (c1) at (8,5) {};
\node[noeud] (c2) at (7.6,4.5) {};
\node[noeud] (c3) at (8.4,4.5) {};

\node[noeud] (d) at (10,6.2) {};

\node[noeud] (u1) at (7,0) {};
\node[noeud] (u2) at (8,0) {};
\node[noeud] (u3) at (9,0) {};

\node[noeud] (v1) at (7.0,1.6) {};

\node[noeud, fill = green] (x1) at (3,5.7) {};
\node[noeud] (x2) at (3,5.1) {};
\node[noeud, fill = green] (x3) at (1.5,3) {};

\node[noeud] (p) at (3.5,3.4) {};

\node[noeud] (z1) at (3.2,2.7) {};
\node[noeud] (z2) at (3.2,2.1) {};

\draw (x) node[above left,scale=1.5] {$x$};
\draw (y) node[below left,scale=1.5] {$y$};
\draw (w) node[below left,scale=1.5] {$w$};
\draw (a) node[above left,scale=1.5] {$a$};
\draw (x2) node[below,yshift = -2.4mm,xshift = -0.4mm,scale=1.2] {$x_2$};

\draw[rounded corners,line width=1mm] (-0.5,-0.5) rectangle (12.0,7.2);

\draw[line width=1mm, color = red] (w) circle (0.3);
\draw[rounded corners,line width=1mm, color = red] (4.3,6.2) rectangle (5.7,3.3);
\draw (5,6.8) node[scale=2, color = red] {$\Omega^{(3)}$};

\draw[rounded corners,line width=1mm, color = red] (6.5,2.3) rectangle (11.7,3.6);




\draw[line width=1mm, color = red] (x2) circle (0.3);
\draw [line width = 1.4pt] (b) -- (a);
\draw [out=-40,in=190,line width = 1.4pt] (b) to (w);

\draw [out=60,in=300,line width = 1.4pt] (w) to (y);

\draw [line width = 1.4pt] (c1) -- (c2);
\draw [line width = 1.4pt] (c2) -- (c3);
\draw [line width = 1.4pt] (c1) -- (c3);

\draw [line width = 1.4pt] (x) -- (c1);
\draw [out=60,in=200,line width = 1.4pt] (z) to (c2);
\draw [line width = 1.4pt] (t1) -- (c2);
\draw [line width = 1.4pt] (t2) -- (c3);

\draw [line width = 1.4pt] (x) -- (d);
\draw [line width = 1.4pt] (d) -- (t3);

\draw [line width = 1.4pt] (y) -- (t1);
\draw [line width = 1.4pt] (y) -- (t2);
\draw [out=0,in=160,line width = 1.4pt] (y) to (t3);
%
%
\draw [line width = 1.4pt] (u1) -- (u2);
\draw [line width = 1.4pt] (u2) -- (u3);
\draw [line width = 1.4pt] (w) -- (t1);
\draw [line width = 1.4pt] (w) -- (u1);
\draw [line width = 1.4pt] (t2) -- (u2);
\draw [line width = 1.4pt] (t3) -- (u3);
\draw [line width = 1.4pt] (z) -- (v1);
\draw [line width = 1.4pt] (v1) -- (t1);
\draw [out=0, in=230,line width = 1.4pt] (v1) to (t3);

\draw [line width = 1.4pt] (x) -- (x1);
\draw [line width = 1.4pt] (x) -- (x2);
\draw [line width = 1.4pt] (x) -- (x3);

\draw [line width = 1.4pt] (y) -- (x1);
\draw [line width = 1.4pt] (y) -- (x2);

\draw [line width = 1.4pt] (y) -- (a);
\draw [line width = 1.4pt] (x1) -- (a);
\draw [line width = 1.4pt] (x2) -- (a);
\draw [line width = 1.4pt] (x3) -- (a);

\draw [line width = 1.4pt] (z) -- (z1);
\draw [line width = 1.4pt] (z) -- (z2);
\draw [line width = 1.4pt] (z1) -- (z2);

\draw [line width = 1.4pt] (p) -- (z1);
\draw [line width = 1.4pt] (p) -- (x2);

\draw [out=140,in=230,line width = 1.4pt] (z2) to (x2);
\end{tikzpicture}

%% file: tikz/preimage-4.tex
\begin{tikzpicture}

\node[noeud] (x) at (5,5.5) {};
\node[noeud] (y) at (5,4) {};
\node[noeud] (z) at (5,2.5) {};
\node[noeud] (w) at (5,1) {};

\node[noeud, fill = green] (a) at (0.5,4) {};
\node[noeud, fill = green] (b) at (0,2.0) {};

\node[noeud] (t1) at (7,3) {};
\node[noeud] (t2) at (9,3) {};
\node[noeud] (t3) at (11,3) {};

\node (c) at (8,4.75) {};
\node[noeud] (c1) at (8,5) {};
\node[noeud] (c2) at (7.6,4.5) {};
\node[noeud] (c3) at (8.4,4.5) {};

\node[noeud] (d) at (10,6.2) {};

\node[noeud] (u1) at (7,0) {};
\node[noeud] (u2) at (8,0) {};
\node[noeud] (u3) at (9,0) {};

\node[noeud] (v1) at (7.0,1.6) {};

\node[noeud] (x1) at (3,5.7) {};
\node[noeud] (x2) at (3,5.1) {};
\node[noeud] (x3) at (1.5,3) {};

\node[noeud] (p) at (3.5,3.4) {};

\node[noeud] (z1) at (3.2,2.7) {};
\node[noeud] (z2) at (3.2,2.1) {};

\draw (a) node[above left,scale=1.5] {$a$};

\draw[rounded corners,line width=1mm] (-0.5,-0.5) rectangle (12.0,7.1);

\draw[line width=1mm, color = red] (y) circle (0.3);
\draw[line width=1mm, color = red] (w) circle (0.3);
\draw (9.8,4.1) node[scale=2, color = red] {$\Omega^{(4)}$};

\draw[rounded corners,line width=1mm, color = red] (6.5,2.3) rectangle (11.7,3.6);




\draw[rounded corners,line width=1mm, color = red] (2.7,4.8) rectangle (3.3,6.0);
\draw[line width=1mm, color = red] (x3) circle (0.3);
\draw [line width = 1.4pt] (b) -- (a);
\draw [out=-40,in=190,line width = 1.4pt] (b) to (w);

\draw [out=60,in=300,line width = 1.4pt] (w) to (y);

\draw [line width = 1.4pt] (c1) -- (c2);
\draw [line width = 1.4pt] (c2) -- (c3);
\draw [line width = 1.4pt] (c1) -- (c3);

\draw [line width = 1.4pt] (x) -- (c1);
\draw [out=60,in=200,line width = 1.4pt] (z) to (c2);
\draw [line width = 1.4pt] (t1) -- (c2);
\draw [line width = 1.4pt] (t2) -- (c3);

\draw [line width = 1.4pt] (x) -- (d);
\draw [line width = 1.4pt] (d) -- (t3);

\draw [line width = 1.4pt] (y) -- (t1);
\draw [line width = 1.4pt] (y) -- (t2);
\draw [out=0,in=160,line width = 1.4pt] (y) to (t3);
%
%
\draw [line width = 1.4pt] (u1) -- (u2);
\draw [line width = 1.4pt] (u2) -- (u3);
\draw [line width = 1.4pt] (w) -- (t1);
\draw [line width = 1.4pt] (w) -- (u1);
\draw [line width = 1.4pt] (t2) -- (u2);
\draw [line width = 1.4pt] (t3) -- (u3);
\draw [line width = 1.4pt] (z) -- (v1);
\draw [line width = 1.4pt] (v1) -- (t1);
\draw [out=0, in=230,line width = 1.4pt] (v1) to (t3);

\draw [line width = 1.4pt] (x) -- (x1);
\draw [line width = 1.4pt] (x) -- (x2);
\draw [line width = 1.4pt] (x) -- (x3);

\draw [line width = 1.4pt] (y) -- (x1);
\draw [line width = 1.4pt] (y) -- (x2);

\draw [line width = 1.4pt] (y) -- (a);
\draw [line width = 1.4pt] (x1) -- (a);
\draw [line width = 1.4pt] (x2) -- (a);
\draw [line width = 1.4pt] (x3) -- (a);

\draw [line width = 1.4pt] (z) -- (z1);
\draw [line width = 1.4pt] (z) -- (z2);
\draw [line width = 1.4pt] (z1) -- (z2);

\draw [line width = 1.4pt] (p) -- (z1);
\draw [line width = 1.4pt] (p) -- (x2);

\draw [out=140,in=230,line width = 1.4pt] (z2) to (x2);
\end{tikzpicture}

%% file: tikz/preimage-5.tex
\begin{tikzpicture}

\node[noeud] (x) at (5,5.5) {};
\node[noeud] (y) at (5,4) {};
\node[noeud] (z) at (5,2.5) {};
\node[noeud] (w) at (5,1) {};

\node[noeud] (a) at (0.5,4) {};
\node[noeud] (b) at (0,2.0) {};

\node[noeud] (t1) at (7,3) {};
\node[noeud] (t2) at (9,3) {};
\node[noeud] (t3) at (11,3) {};

\node (c) at (8,4.75) {};
\node[noeud] (c1) at (8,5) {};
\node[noeud] (c2) at (7.6,4.5) {};
\node[noeud] (c3) at (8.4,4.5) {};

\node[noeud] (d) at (10,6.2) {};

\node[noeud] (u1) at (7,0) {};
\node[noeud] (u2) at (8,0) {};
\node[noeud] (u3) at (9,0) {};

\node[noeud] (v1) at (7.0,1.6) {};

\node[noeud] (x1) at (3,5.7) {};
\node[noeud] (x2) at (3,5.1) {};
\node[noeud] (x3) at (1.5,3) {};

\node[noeud] (p) at (3.5,3.4) {};

\node[noeud] (z1) at (3.2,2.7) {};
\node[noeud] (z2) at (3.2,2.1) {};

\draw (a) node[above left,scale=1.5] {$a$};

\draw[rounded corners,line width=1.5mm, color = green] (-0.5,-0.5) rectangle (12.0,7.1);

\draw[line width=1mm, color = red] (y) circle (0.3);
\draw[line width=1mm, color = red] (w) circle (0.3);
\draw (9.8,4.1) node[scale=2, color = red] {$\Omega^{(5)}$};

\draw[rounded corners,line width=1mm, color = red] (6.5,2.3) rectangle (11.7,3.6);




\draw[line width=1mm, color = red] (a) circle (0.3);
\draw [line width = 1.4pt] (b) -- (a);
\draw [out=-40,in=190,line width = 1.4pt] (b) to (w);

\draw [out=60,in=300,line width = 1.4pt] (w) to (y);

\draw [line width = 1.4pt] (c1) -- (c2);
\draw [line width = 1.4pt] (c2) -- (c3);
\draw [line width = 1.4pt] (c1) -- (c3);

\draw [line width = 1.4pt] (x) -- (c1);
\draw [out=60,in=200,line width = 1.4pt] (z) to (c2);
\draw [line width = 1.4pt] (t1) -- (c2);
\draw [line width = 1.4pt] (t2) -- (c3);

\draw [line width = 1.4pt] (x) -- (d);
\draw [line width = 1.4pt] (d) -- (t3);

\draw [line width = 1.4pt] (y) -- (t1);
\draw [line width = 1.4pt] (y) -- (t2);
\draw [out=0,in=160,line width = 1.4pt] (y) to (t3);
%
%
\draw [line width = 1.4pt] (u1) -- (u2);
\draw [line width = 1.4pt] (u2) -- (u3);
\draw [line width = 1.4pt] (w) -- (t1);
\draw [line width = 1.4pt] (w) -- (u1);
\draw [line width = 1.4pt] (t2) -- (u2);
\draw [line width = 1.4pt] (t3) -- (u3);
\draw [line width = 1.4pt] (z) -- (v1);
\draw [line width = 1.4pt] (v1) -- (t1);
\draw [out=0, in=230,line width = 1.4pt] (v1) to (t3);

\draw [line width = 1.4pt] (x) -- (x1);
\draw [line width = 1.4pt] (x) -- (x2);
\draw [line width = 1.4pt] (x) -- (x3);

\draw [line width = 1.4pt] (y) -- (x1);
\draw [line width = 1.4pt] (y) -- (x2);

\draw [line width = 1.4pt] (y) -- (a);
\draw [line width = 1.4pt] (x1) -- (a);
\draw [line width = 1.4pt] (x2) -- (a);
\draw [line width = 1.4pt] (x3) -- (a);

\draw [line width = 1.4pt] (z) -- (z1);
\draw [line width = 1.4pt] (z) -- (z2);
\draw [line width = 1.4pt] (z1) -- (z2);

\draw [line width = 1.4pt] (p) -- (z1);
\draw [line width = 1.4pt] (p) -- (x2);

\draw [out=140,in=230,line width = 1.4pt] (z2) to (x2);
\end{tikzpicture}

%% file: tikz/solutionpath2.tex
\begin{tikzpicture}
\node (p0) at (2,30) {\scalebox{0.6}{\input{tikz/preimage-0}}};
\node (p1) at (10,30) {\scalebox{0.6}{\input{tikz/preimage-1alt}}};

\node (p2) at (10,23) {\scalebox{0.6}{\input{tikz/preimage-2alt}}};

\node (a2) at (8.5,25.5) {};
\node (b1) at (8.5,27.5) {};

\node (a1) at (7,27.5) {};
\node (b0) at (3,27.5) {};

\node[scale = 1.2,color = blue] (o) at (0.2,27.2) {Target PMC reached !};
\node[scale = 1.2,color = green] (o5) at (12.3,20.2) {$a \in \Omega^{**}$};

\draw[->,line width = 2pt] (a2)to node[text centered,text width=60mm,xshift=30mm]{$(\Omega^{**},\Omega^*) \in \mathcal{N}$: $\{a\}$ is an $a$-shutter which separates pairs $(b,x),(b,z),(y,z)$ closely.} (b1);

\draw[->,line width = 2pt,out=210, in=-30] (a1)to node[text centered,text width=60mm,xshift=-20mm, yshift=-12mm]{$(\Omega^*,\Omega) \in \mathcal{N}$: $\{b\}$ is an $a$-shutter which separates pair $(x,w)$ closely. $w$ is adjacent to $b$ and is not adjacent to $C(a,\Omega \cup \{b\})$.} (b0);
\end{tikzpicture}

%% file: tikz/preimage-1alt.tex
\begin{tikzpicture}

\node[noeud] (x) at (5,5.5) {};
\node[noeud] (y) at (5,4) {};
\node[noeud] (z) at (5,2.5) {};
\node[noeud] (w) at (5,1) {};

\node[noeud, fill = green] (a) at (0.5,4) {};
\node[noeud] (b) at (0,2.0) {};

\node[noeud] (t1) at (7,3) {};
\node[noeud] (t2) at (9,3) {};
\node[noeud] (t3) at (11,3) {};

\node (c) at (8,4.75) {};
\node[noeud] (c1) at (8,5) {};
\node[noeud] (c2) at (7.6,4.5) {};
\node[noeud] (c3) at (8.4,4.5) {};

\node[noeud] (d) at (10,6.2) {};

\node[noeud] (u1) at (7,0) {};
\node[noeud] (u2) at (8,0) {};
\node[noeud] (u3) at (9,0) {};

\node[noeud] (v1) at (7.0,1.6) {};

\node[noeud, fill = green] (x1) at (3,5.7) {};
\node[noeud, fill = green] (x2) at (3,5.1) {};
\node[noeud, fill = green] (x3) at (1.5,3) {};

\node[noeud, fill = green] (p) at (3.5,3.4) {};

\node[noeud, fill = green] (z1) at (3.2,2.7) {};
\node[noeud, fill = green] (z2) at (3.2,2.1) {};

\draw (x) node[above left,scale=1.5] {$x$};
\draw (y) node[below left,scale=1.5] {$y$};
\draw (z) node[below left,scale=1.5] {$z$};
\draw (a) node[above left,scale=1.5] {$a$};
\draw (b) node[above right,scale=1.5] {$b$};

\draw[rounded corners,line width=1mm] (-0.5,-0.5) rectangle (12.0,7.2);

\draw[rounded corners,line width=1mm, color = red] (4.3,6.2) rectangle (5.7,1.8);
\draw (5,6.8) node[scale=2, color = red] {$\Omega^*$};

\draw[rounded corners,line width=1mm, color = red] (6.5,2.3) rectangle (11.7,3.6);



\draw[line width=1mm, color = red] (b) circle (0.3);
\draw [line width = 1.4pt] (b) -- (a);
\draw [out=-40,in=190,line width = 1.4pt] (b) to (w);

\draw [out=60,in=300,line width = 1.4pt] (w) to (y);

\draw [line width = 1.4pt] (c1) -- (c2);
\draw [line width = 1.4pt] (c2) -- (c3);
\draw [line width = 1.4pt] (c1) -- (c3);

\draw [line width = 1.4pt] (x) -- (c1);
\draw [out=60,in=200,line width = 1.4pt] (z) to (c2);
\draw [line width = 1.4pt] (t1) -- (c2);
\draw [line width = 1.4pt] (t2) -- (c3);

\draw [line width = 1.4pt] (x) -- (d);
\draw [line width = 1.4pt] (d) -- (t3);

\draw [line width = 1.4pt] (y) -- (t1);
\draw [line width = 1.4pt] (y) -- (t2);
\draw [out=0,in=160,line width = 1.4pt] (y) to (t3);
%
%
\draw [line width = 1.4pt] (u1) -- (u2);
\draw [line width = 1.4pt] (u2) -- (u3);
\draw [line width = 1.4pt] (w) -- (t1);
\draw [line width = 1.4pt] (w) -- (u1);
\draw [line width = 1.4pt] (t2) -- (u2);
\draw [line width = 1.4pt] (t3) -- (u3);
\draw [line width = 1.4pt] (z) -- (v1);
\draw [line width = 1.4pt] (v1) -- (t1);
\draw [out=0, in=230,line width = 1.4pt] (v1) to (t3);

\draw [line width = 1.4pt] (x) -- (x1);
\draw [line width = 1.4pt] (x) -- (x2);
\draw [line width = 1.4pt] (x) -- (x3);

\draw [line width = 1.4pt] (y) -- (x1);
\draw [line width = 1.4pt] (y) -- (x2);

\draw [line width = 1.4pt] (y) -- (a);
\draw [line width = 1.4pt] (x1) -- (a);
\draw [line width = 1.4pt] (x2) -- (a);
\draw [line width = 1.4pt] (x3) -- (a);

\draw [line width = 1.4pt] (z) -- (z1);
\draw [line width = 1.4pt] (z) -- (z2);
\draw [line width = 1.4pt] (z1) -- (z2);

\draw [line width = 1.4pt] (p) -- (z1);
\draw [line width = 1.4pt] (p) -- (x2);

\draw [out=140,in=230,line width = 1.4pt] (z2) to (x2);
\end{tikzpicture}

%% file: tikz/preimage-2alt.tex
\begin{tikzpicture}

\node[noeud] (x) at (5,5.5) {};
\node[noeud] (y) at (5,4) {};
\node[noeud] (z) at (5,2.5) {};
\node[noeud] (w) at (5,1) {};

\node[noeud] (a) at (0.5,4) {};
\node[noeud] (b) at (0,2.0) {};

\node[noeud] (t1) at (7,3) {};
\node[noeud] (t2) at (9,3) {};
\node[noeud] (t3) at (11,3) {};

\node (c) at (8,4.75) {};
\node[noeud] (c1) at (8,5) {};
\node[noeud] (c2) at (7.6,4.5) {};
\node[noeud] (c3) at (8.4,4.5) {};

\node[noeud] (d) at (10,6.2) {};

\node[noeud] (u1) at (7,0) {};
\node[noeud] (u2) at (8,0) {};
\node[noeud] (u3) at (9,0) {};

\node[noeud] (v1) at (7.0,1.6) {};

\node[noeud] (x1) at (3,5.7) {};
\node[noeud] (x2) at (3,5.1) {};
\node[noeud] (x3) at (1.5,3) {};

\node[noeud] (p) at (3.5,3.4) {};

\node[noeud] (z1) at (3.2,2.7) {};
\node[noeud] (z2) at (3.2,2.1) {};

\draw (x) node[above left,scale=1.5] {$x$};
\draw (y) node[below left,scale=1.5] {$y$};
\draw (z) node[below left,scale=1.5] {$z$};
\draw (a) node[above left,scale=1.5] {$a$};

\draw[rounded corners,line width=1.5mm, color = green] (-0.5,-0.5) rectangle (12.0,7.2);

\draw[rounded corners,line width=1mm, color = red] (4.3,6.2) rectangle (5.7,1.8);
\draw (5,6.8) node[scale=2, color = red] {$\Omega^{**}$};

\draw[rounded corners,line width=1mm, color = red] (6.5,2.3) rectangle (11.7,3.6);



\draw[line width=1mm, color = red] (a) circle (0.3);
\draw [line width = 1.4pt] (b) -- (a);
\draw [out=-40,in=190,line width = 1.4pt] (b) to (w);

\draw [out=60,in=300,line width = 1.4pt] (w) to (y);

\draw [line width = 1.4pt] (c1) -- (c2);
\draw [line width = 1.4pt] (c2) -- (c3);
\draw [line width = 1.4pt] (c1) -- (c3);

\draw [line width = 1.4pt] (x) -- (c1);
\draw [out=60,in=200,line width = 1.4pt] (z) to (c2);
\draw [line width = 1.4pt] (t1) -- (c2);
\draw [line width = 1.4pt] (t2) -- (c3);

\draw [line width = 1.4pt] (x) -- (d);
\draw [line width = 1.4pt] (d) -- (t3);

\draw [line width = 1.4pt] (y) -- (t1);
\draw [line width = 1.4pt] (y) -- (t2);
\draw [out=0,in=160,line width = 1.4pt] (y) to (t3);
%
%
\draw [line width = 1.4pt] (u1) -- (u2);
\draw [line width = 1.4pt] (u2) -- (u3);
\draw [line width = 1.4pt] (w) -- (t1);
\draw [line width = 1.4pt] (w) -- (u1);
\draw [line width = 1.4pt] (t2) -- (u2);
\draw [line width = 1.4pt] (t3) -- (u3);
\draw [line width = 1.4pt] (z) -- (v1);
\draw [line width = 1.4pt] (v1) -- (t1);
\draw [out=0, in=230,line width = 1.4pt] (v1) to (t3);

\draw [line width = 1.4pt] (x) -- (x1);
\draw [line width = 1.4pt] (x) -- (x2);
\draw [line width = 1.4pt] (x) -- (x3);

\draw [line width = 1.4pt] (y) -- (x1);
\draw [line width = 1.4pt] (y) -- (x2);

\draw [line width = 1.4pt] (y) -- (a);
\draw [line width = 1.4pt] (x1) -- (a);
\draw [line width = 1.4pt] (x2) -- (a);
\draw [line width = 1.4pt] (x3) -- (a);

\draw [line width = 1.4pt] (z) -- (z1);
\draw [line width = 1.4pt] (z) -- (z2);
\draw [line width = 1.4pt] (z1) -- (z2);

\draw [line width = 1.4pt] (p) -- (z1);
\draw [line width = 1.4pt] (p) -- (x2);

\draw [out=140,in=230,line width = 1.4pt] (z2) to (x2);
\end{tikzpicture}

%% file: tikz/minabSepExt_instance.tex
\begin{tikzpicture}
\node[scale=1.5] at (-0.6,-1) {$a$};
\node[draw, circle, minimum height=0.2cm, minimum width=0.2cm, fill=black] (a) at (0,-1) {};

\node[draw, rectangle, minimum height=0.2cm, minimum width=0.2cm, fill=red] (3) at (1,0) {};
\draw (3) node[above left,scale=1.5, color = red] {$f_1$};

\node[scale=1.5, color = black!30!green] at (2,0.4) {$m_1$};
\node[draw, diamond, minimum height=0.2cm, minimum width=0.2cm, fill=black!30!green] (4) at (2,0) {};
\node[draw, circle, minimum height=0.2cm, minimum width=0.2cm, fill=black] (1) at (2,1.1) {};
\node[draw, circle, minimum height=0.2cm, minimum width=0.2cm, fill=black] (5) at (1,-1) {};
\node[draw, circle, minimum height=0.2cm, minimum width=0.2cm, fill=black] (6) at (2.1,-0.8) {};
\node[draw, circle, minimum height=0.2cm, minimum width=0.2cm, fill=black] (2) at (4,1) {};
\node[draw, circle, minimum height=0.2cm, minimum width=0.2cm, fill=black] (7) at (4,-0.5) {};
\node[scale=1.5, color = black!30!green] at (2,-1.5) {$m_2$};
\node[draw, diamond, minimum height=0.2cm, minimum width=0.2cm, fill=black!30!green] (8) at (2,-2) {};

\node[draw, rectangle, minimum height=0.2cm, minimum width=0.2cm, fill=red] (9) at (3,-2) {};
\draw (9) node[below right,scale=1.5, color = red] {$f_2$};

\node[draw, circle, minimum height=0.2cm, minimum width=0.2cm, fill=black] (10) at (4,-1.5) {};

\node[scale=1.5] at (5.6,0) {$b$};
\node[draw, circle, minimum height=0.2cm, minimum width=0.2cm, fill=black] (b) at (5,0) {};

\draw[line width = 1.4pt] (3) -- (1);
\draw[line width = 1.4pt] (a) -- (3);
\draw[line width = 1.4pt] (a) -- (5);
\draw[line width = 1.4pt] (a) -- (8);

\draw[line width = 1.4pt] (1) -- (2);

\draw[line width = 1.4pt] (2) -- (b);

\draw[line width = 1.4pt] (3) -- (4);
\draw[line width = 1.4pt] (3) -- (6);

\draw[line width = 1.4pt] (4) -- (2);
\draw[line width = 1.4pt] (4) -- (6);

\draw[line width = 1.4pt] (5) -- (6);

\draw[line width = 1.4pt] (6) -- (7);
\draw[line width = 1.4pt] (6) -- (9);

\draw[line width = 1.4pt] (7) -- (b);

\draw[line width = 1.4pt] (8) -- (9);

\draw[line width = 1.4pt] (9) -- (10);

\draw[line width = 1.4pt] (10) -- (b);

\node[scale=2, color = blue] at (2,2) {$S$};
\draw[line width = 1.4pt, color = blue] (2,-0.5) ellipse (0.65 and 2.1);

\end{tikzpicture}

%% file: tikz/gadget_xi.tex
\begin{tikzpicture}
\node[draw, circle, minimum height=0.2cm, minimum width=0.2cm, fill=black] (a) at (0,0) {};
\draw (a) node[below left,scale=1.5] {$a$};

\node[draw, circle, minimum height=0.2cm, minimum width=0.2cm, fill=black] (xiv1) at (-1,3) {};
\draw (xiv1) node[above left,scale=1.5] {$\beta_i$};

\node[draw, circle, minimum height=0.2cm, minimum width=0.2cm, fill=black] (xif1) at (1,3) {};
\draw (xif1) node[above right,scale=1.5] {$\notxi{\beta}{i}$};

\node[draw, circle, minimum height=0.2cm, minimum width=0.2cm, fill=black] (xiv2) at (1,2) {};
\draw (xiv2) node[below right,scale=1.5] {$\alpha_i$};

\node[draw, circle, minimum height=0.2cm, minimum width=0.2cm, fill=black] (xif2) at (-1,2) {};
\draw (xif2) node[below left,scale=1.5] {$\notxi{\alpha}{i}$};

\node[draw, circle, minimum height=0.2cm, minimum width=0.2cm, fill=black] (b) at (0,5) {};
\draw (b) node[above right,scale=1.5] {$b$};

\node[draw, diamond, minimum height=0.2cm, minimum width=0.2cm, fill=black!30!green] (m1) at (0,2) {};
\draw (m1) node[below left,scale=1.5, color = black!30!green] {$a_i$};

\node[draw, diamond, minimum height=0.2cm, minimum width=0.2cm, fill=black!30!green] (m2) at (0,3) {};
\draw (m2) node[above right,scale=1.5, color = black!30!green] {$b_i$};

%

\draw [line width = 1.4pt] (a) -- (xiv2);
\draw [line width = 1.4pt] (a) -- (xif2);
\draw [line width = 1.4pt] (a) -- (m1);

\draw [line width = 1.4pt] (b) -- (xiv1);
\draw [line width = 1.4pt] (b) -- (xif1);
\draw [line width = 1.4pt] (b) -- (m2);

\draw [line width = 1.4pt] (xiv1) -- (xif2);
\draw [line width = 1.4pt] (xif1) -- (xiv2);

\draw [line width = 1.4pt] (m2) -- (xiv2);
\draw [line width = 1.4pt] (m2) -- (xif2);

\draw [line width = 1.4pt] (m1) -- (xiv1);
\draw [line width = 1.4pt] (m1) -- (xif1);

%

\end{tikzpicture}

%% file: tikz/clause_ci.tex
\begin{tikzpicture}
\node[draw, circle, minimum height=0.2cm, minimum width=0.2cm, fill=black] (a) at (0,0) {};
\draw (a) node[below left,scale=1.5] {$a$};

\node[draw, circle, minimum height=0.2cm, minimum width=0.2cm, fill=black] (xiv1) at (-4,3) {};
\draw (xiv1) node[above left,scale=1.5] {$\beta_i$};

\node[draw, circle, minimum height=0.2cm, minimum width=0.2cm, fill=black] (xif1) at (-2,3) {};
\draw (xif1) node[right,scale=1.5] {$\notxi{\beta}{i}$};

\node[draw, circle, minimum height=0.2cm, minimum width=0.2cm, fill=black] (xiv2) at (-2,2) {};
\draw (xiv2) node[right,scale=1.5] {$\alpha_i$};

\node[draw, circle, minimum height=0.2cm, minimum width=0.2cm, fill=black] (xif2) at (-4,2) {};
\draw (xif2) node[below left,scale=1.5] {$\notxi{\alpha}{i}$};

\node[draw, circle, minimum height=0.2cm, minimum width=0.2cm, fill=black] (xjv1) at (2,3) {};
\draw (xjv1) node[left,scale=1.5] {$\beta_k$};

\node[draw, circle, minimum height=0.2cm, minimum width=0.2cm, fill=black] (xjf1) at (4,3) {};
\draw (xjf1) node[above right,scale=1.5] {$\notxi{\beta}{k}$};

\node[draw, circle, minimum height=0.2cm, minimum width=0.2cm, fill=black] (xjv2) at (4,2) {};
\draw (xjv2) node[below right,scale=1.5] {$\alpha_k$};

\node[draw, circle, minimum height=0.2cm, minimum width=0.2cm, fill=black] (xjf2) at (2,2) {};
\draw (xjf2) node[left,scale=1.5] {$\notxi{\alpha}{k}$};

\node[draw, circle, minimum height=0.2cm, minimum width=0.2cm, fill=black] (b) at (0,5) {};
\draw (b) node[above right,scale=1.5] {$b$};

\node[draw, diamond, minimum height=0.2cm, minimum width=0.2cm, fill=black!30!green] (m1) at (-3,2) {};
\node[draw, diamond, minimum height=0.2cm, minimum width=0.2cm, fill=black!30!green] (n1) at (3,2) {};

\node[draw, diamond, minimum height=0.2cm, minimum width=0.2cm, fill=black!30!green] (m2) at (-3,3) {};
\node[draw, diamond, minimum height=0.2cm, minimum width=0.2cm, fill=black!30!green] (n2) at (3,3) {};

\node[draw, diamond, minimum height=0.2cm, minimum width=0.2cm, fill=black!30!green] (cv) at (0,2.5) {};
\draw (cv) node[above,scale=1.5, color = black!30!green] {$C_j$};
%

\draw [line width = 1.4pt] (a) -- (xiv2);
\draw [line width = 1.4pt] (a) -- (xif2);
\draw [line width = 1.4pt] (a) -- (m1);

\draw [line width = 1.4pt] (a) -- (xjv2);
\draw [line width = 1.4pt] (a) -- (xjf2);
\draw [line width = 1.4pt] (a) -- (n1);

\draw [line width = 1.4pt] (b) -- (xiv1);
\draw [line width = 1.4pt] (b) -- (xif1);
\draw [line width = 1.4pt] (b) -- (m2);

\draw [line width = 1.4pt] (b) -- (xjv1);
\draw [line width = 1.4pt] (b) -- (xjf1);
\draw [line width = 1.4pt] (b) -- (n2);

\draw [line width = 1.4pt] (xiv1) -- (xif2);
\draw [line width = 1.4pt] (xif1) -- (xiv2);

\draw [line width = 1.4pt] (xjv1) -- (xjf2);
\draw [line width = 1.4pt] (xjf1) -- (xjv2);

\draw [line width = 1.4pt] (m2) -- (xiv2);
\draw [line width = 1.4pt] (m2) -- (xif2);

\draw [line width = 1.4pt] (n2) -- (xjv2);
\draw [line width = 1.4pt] (n2) -- (xjf2);

\draw [line width = 1.4pt] (m1) -- (xiv1);
\draw [line width = 1.4pt] (m1) -- (xif1);

\draw [line width = 1.4pt] (n1) -- (xjv1);
\draw [line width = 1.4pt] (n1) -- (xjf1);

\draw [out = 180, in = -70, line width = 1.4pt] (cv) to (xif1);
\draw [out = 0, in = 250, line width = 1.4pt] (cv) to (xjv1);
\draw [line width = 1.4pt] (cv) -- (a);
%

\end{tikzpicture}

%% file: tikz/reduc_3sat.tex
\begin{tikzpicture}
\node[draw, circle, minimum height=0.2cm, minimum width=0.2cm, fill=black] (a) at (10,-4) {};
\draw (a) node[right,scale=2,xshift = 1mm] {$a$};
\node[draw, circle, minimum height=0.2cm, minimum width=0.2cm, fill=black] (b) at (10,6.5) {};
\draw (b) node[above,scale=2,yshift = 1mm] {$b$};

\node[scale=1.7] at (-1.8,2) {$\beta_1$};
\node[draw, circle, minimum height=0.2cm, minimum width=0.2cm, fill=blue] (x1v1) at (-1,2) {};

\node[scale=1.7] at (1.8,2) {$\notxi{\beta}{1}$};
\node[draw, circle, minimum height=0.2cm, minimum width=0.2cm, fill=black] (x1f1) at (1,2) {};

\node[scale=1.7] at (1.8,1) {$\alpha_1$};
\node[draw, circle, minimum height=0.2cm, minimum width=0.2cm, fill=blue] (x1v2) at (1,1) {};

\node[scale=1.7] at (-1.8,1) {$\notxi{\alpha}{1}$};
\node[draw, circle, minimum height=0.2cm, minimum width=0.2cm, fill=black] (x1f2) at (-1,1) {};

\node[draw, diamond, minimum height=0.2cm, minimum width=0.2cm, fill=black!30!green] (m11) at (0,1) {};

\node[draw, diamond, minimum height=0.2cm, minimum width=0.2cm, fill=black!30!green] (m12) at (0,2) {};

\draw [line width = 0.4pt] (a) -- (x1v2);
\draw [line width = 0.4pt] (a) -- (x1f2);
\draw [line width = 0.4pt] (a) -- (m11);

\draw [line width = 0.4pt] (b) -- (x1v1);
\draw [line width = 0.4pt] (b) -- (x1f1);
\draw [line width = 0.4pt] (b) -- (m12);

\draw [line width = 1.4pt] (x1v1) -- (x1f2);
\draw [line width = 1.4pt] (x1f1) -- (x1v2);

\draw [line width = 1.4pt] (m12) -- (x1v2);
\draw [line width = 1.4pt] (m12) -- (x1f2);

\draw [line width = 1.4pt] (m11) -- (x1v1);
\draw [line width = 1.4pt] (m11) -- (x1f1);

\node[scale=1.7] at (3.2,2) {$\beta_2$};
\node[draw, circle, minimum height=0.2cm, minimum width=0.2cm, fill=blue] (x2v1) at (4,2) {};

\node[scale=1.7] at (6.8,2) {$\notxi{\beta}{2}$};
\node[draw, circle, minimum height=0.2cm, minimum width=0.2cm, fill=black] (x2f1) at (6,2) {};

\node[scale=1.7] at (6.8,1) {$\alpha_2$};
\node[draw, circle, minimum height=0.2cm, minimum width=0.2cm, fill=blue] (x2v2) at (6,1) {};

\node[scale=1.7] at (3.2,1) {$\notxi{\alpha}{2}$};
\node[draw, circle, minimum height=0.2cm, minimum width=0.2cm, fill=black] (x2f2) at (4,1) {};

\node[draw, diamond, minimum height=0.2cm, minimum width=0.2cm, fill=black!30!green] (m21) at (5,1) {};

\node[draw, diamond, minimum height=0.2cm, minimum width=0.2cm, fill=black!30!green] (m22) at (5,2) {};

\draw [line width = 0.4pt] (a) -- (x2v2);
\draw [line width = 0.4pt] (a) -- (x2f2);
\draw [line width = 0.4pt] (a) -- (m21);

\draw [line width = 0.4pt] (b) -- (x2v1);
\draw [line width = 0.4pt] (b) -- (x2f1);
\draw [line width = 0.4pt] (b) -- (m22);

\draw [line width = 1.4pt] (x2v1) -- (x2f2);
\draw [line width = 1.4pt] (x2f1) -- (x2v2);

\draw [line width = 1.4pt] (m22) -- (x2v2);
\draw [line width = 1.4pt] (m22) -- (x2f2);

\draw [line width = 1.4pt] (m21) -- (x2v1);
\draw [line width = 1.4pt] (m21) -- (x2f1);

\node[scale=1.7] at (8.2,2) {$\beta_3$};
\node[draw, circle, minimum height=0.2cm, minimum width=0.2cm, fill=black] (x3v1) at (9,2) {};

\node[scale=1.7] at (11.8,2) {$\notxi{\beta}{3}$};
\node[draw, circle, minimum height=0.2cm, minimum width=0.2cm, fill=blue] (x3f1) at (11,2) {};

\node[scale=1.7] at (11.8,1) {$\alpha_3$};
\node[draw, circle, minimum height=0.2cm, minimum width=0.2cm, fill=black] (x3v2) at (11,1) {};

\node[scale=1.7] at (8.2,1) {$\notxi{\alpha}{3}$};
\node[draw, circle, minimum height=0.2cm, minimum width=0.2cm, fill=blue] (x3f2) at (9,1) {};

\node[draw, diamond, minimum height=0.2cm, minimum width=0.2cm, fill=black!30!green] (m31) at (10,1) {};

\node[draw, diamond, minimum height=0.2cm, minimum width=0.2cm, fill=black!30!green] (m32) at (10,2) {};

\draw [line width = 0.4pt] (a) -- (x3v2);
\draw [line width = 0.4pt] (a) -- (x3f2);
\draw [line width = 0.4pt] (a) -- (m31);

\draw [line width = 0.4pt] (b) -- (x3v1);
\draw [line width = 0.4pt] (b) -- (x3f1);
\draw [line width = 0.4pt] (b) -- (m32);

\draw [line width = 1.4pt] (x3v1) -- (x3f2);
\draw [line width = 1.4pt] (x3f1) -- (x3v2);

\draw [line width = 1.4pt] (m32) -- (x3v2);
\draw [line width = 1.4pt] (m32) -- (x3f2);

\draw [line width = 1.4pt] (m31) -- (x3v1);
\draw [line width = 1.4pt] (m31) -- (x3f1);

\node[scale=1.7] at (13.2,2) {$\beta_4$};
\node[draw, circle, minimum height=0.2cm, minimum width=0.2cm, fill=black] (x4v1) at (14,2) {};

\node[scale=1.7] at (16.8,2) {$\notxi{\beta}{4}$};
\node[draw, circle, minimum height=0.2cm, minimum width=0.2cm, fill=blue] (x4f1) at (16,2) {};

\node[scale=1.7] at (16.8,1) {$\alpha_4$};
\node[draw, circle, minimum height=0.2cm, minimum width=0.2cm, fill=black] (x4v2) at (16,1) {};

\node[scale=1.7] at (13.2,1) {$\notxi{\alpha}{4}$};
\node[draw, circle, minimum height=0.2cm, minimum width=0.2cm, fill=blue] (x4f2) at (14,1) {};

\node[draw, diamond, minimum height=0.2cm, minimum width=0.2cm, fill=black!30!green] (m41) at (15,1) {};

\node[draw, diamond, minimum height=0.2cm, minimum width=0.2cm, fill=black!30!green] (m42) at (15,2) {};

\draw [line width = 0.4pt] (a) -- (x4v2);
\draw [line width = 0.4pt] (a) -- (x4f2);
\draw [line width = 0.4pt] (a) -- (m41);

\draw [line width = 0.4pt] (b) -- (x4v1);
\draw [line width = 0.4pt] (b) -- (x4f1);
\draw [line width = 0.4pt] (b) -- (m42);

\draw [line width = 1.4pt] (x4v1) -- (x4f2);
\draw [line width = 1.4pt] (x4f1) -- (x4v2);

\draw [line width = 1.4pt] (m42) -- (x4v2);
\draw [line width = 1.4pt] (m42) -- (x4f2);

\draw [line width = 1.4pt] (m41) -- (x4v1);
\draw [line width = 1.4pt] (m41) -- (x4f1);

\node[scale=1.7] at (18.2,2) {$\beta_5$};
\node[draw, circle, minimum height=0.2cm, minimum width=0.2cm, fill=black] (x5v1) at (19,2) {};

\node[scale=1.7] at (21.8,2) {$\notxi{\beta}{5}$};
\node[draw, circle, minimum height=0.2cm, minimum width=0.2cm, fill=blue] (x5f1) at (21,2) {};

\node[scale=1.7] at (21.8,1) {$\alpha_5$};
\node[draw, circle, minimum height=0.2cm, minimum width=0.2cm, fill=black] (x5v2) at (21,1) {};

\node[scale=1.7] at (18.2,1) {$\notxi{\alpha}{5}$};
\node[draw, circle, minimum height=0.2cm, minimum width=0.2cm, fill=blue] (x5f2) at (19,1) {};

\node[draw, diamond, minimum height=0.2cm, minimum width=0.2cm, fill=black!30!green] (m51) at (20,1) {};

\node[draw, diamond, minimum height=0.2cm, minimum width=0.2cm, fill=black!30!green] (m52) at (20,2) {};

\draw [line width = 0.4pt] (a) -- (x5v2);
\draw [line width = 0.4pt] (a) -- (x5f2);
\draw [line width = 0.4pt] (a) -- (m51);

\draw [line width = 0.4pt] (b) -- (x5v1);
\draw [line width = 0.4pt] (b) -- (x5f1);
\draw [line width = 0.4pt] (b) -- (m52);

\draw [line width = 1.4pt] (x5v1) -- (x5f2);
\draw [line width = 1.4pt] (x5f1) -- (x5v2);

\draw [line width = 1.4pt] (m52) -- (x5v2);
\draw [line width = 1.4pt] (m52) -- (x5f2);

\draw [line width = 1.4pt] (m51) -- (x5v1);
\draw [line width = 1.4pt] (m51) -- (x5f1);

\node[scale=1.7, black!30!green] at (7,-7.6) {$C_1$};
\node [draw, diamond, minimum height=0.2cm, minimum width=0.2cm, fill=black!30!green] (c1) at (7,-7) {};
\node[scale=1.7, black!30!green] at (10,-7.6) {$C_2$};
\node [draw, diamond, minimum height=0.2cm, minimum width=0.2cm, fill=black!30!green] (c2) at (10,-7) {};
\node[scale=1.7, black!30!green] at (13,-7.6) {$C_3$};
\node [draw, diamond, minimum height=0.2cm, minimum width=0.2cm, fill=black!30!green] (c3) at (13,-7){};

\draw [line width = 1.4pt] (c1) -- (a);
\draw [line width = 1.4pt] (c1) -- (x1f1);
\draw [line width = 1.4pt] (c1) -- (x2f1);
\draw [line width = 1.4pt] (c1) -- (x3f1);

\draw [line width = 1.4pt] (c2) -- (a);
\draw [line width = 1.4pt] (c2) -- (x1v1);
\draw [line width = 1.4pt] (c2) -- (x2f1);
\draw [line width = 1.4pt] (c2) -- (x5f1);

\draw [line width = 1.4pt] (c3) -- (a);
\draw [line width = 1.4pt] (c3) -- (x3v1);
\draw [line width = 1.4pt] (c3) -- (x4v1);
\draw [line width = 1.4pt] (c3) -- (x5v1);
\end{tikzpicture}